\documentclass{article}
\usepackage[a4paper,margin=2cm,footskip=.5cm]{geometry}
\usepackage{setspace}
\usepackage{graphicx, epstopdf, parskip, latexsym, amsmath, amssymb, mathtools, times, amsthm, url, multicol, rotating, natbib, algorithm, algpseudocode, bm}
\usepackage[english]{babel}
\usepackage[utf8]{inputenc}
\usepackage{latexsym}
\usepackage{fancyhdr}
\usepackage{color}

\theoremstyle{plain}
\newtheorem{thm}{Theorem}
\newtheorem{Lem}{Lemma}
\newtheorem{Cor}{Corollary}

\theoremstyle{definition}
\newtheorem{Def}{Definition}
\newtheorem{Assum}{Assumption}

\theoremstyle{remark}
\newtheorem{Rem}{Remark}
\newtheorem{Eg}{Example}
\newtheorem{Con}{Condition}

\DeclareMathOperator{\Var}{Var}
\DeclareMathOperator{\Cov}{Cov}
\DeclareMathOperator{\Corr}{Corr}

\DeclareMathOperator{\Leb}{Leb}

\allowdisplaybreaks

\title{Modelling spatial heteroskedasticity\\by volatility modulated moving averages}
\author{MICHELE NGUYEN and ALMUT E. D. VERAART\\
	\textit{Department of Mathematics, Imperial College London}
	}
\providecommand{\keywords}[1]{\textbf{\textit{Keywords:}} #1}

\date{}

\begin{document}

\maketitle
\pagenumbering{arabic}

\thispagestyle{fancy}

\begin{abstract}
Spatial heteroskedasticity refers to stochastically changing variances and covariances in space. Such features have been observed in, for example, air pollution and vegetation data. We study how volatility modulated moving averages can model this by developing theory, simulation and statistical inference methods. For illustration, we also apply our procedure to sea surface temperature anomaly data from the International Research Institute for Climate and Society. 
\end{abstract}

\keywords{moments-based inference, moving averages, stochastic simulation, stochastic volatility, spatial processes}

\section{Introduction}

A classical assumption made when dealing with spatial data is that the variance is a constant and the covariance between measurements at two locations is a function of their distance apart. In practice, however, it has been observed that this does not hold for many data sets and accounting for spatial heteroskedasticity or spatial volatility has multiple benefits. 
\\
The first benefit is the better representation of the data. In a recent paper, it was shown that including spatial volatility in road topography models better captures the hilliness features of the roads  \cite[]{JPRS2016}. This has implications on estimating the risk of vehicle damage and simulating fuel consumption. In some settings, the presence of spatial volatility can be also explained. For example, in a study of sulphur dioxide concentrations by \cite{FS2001}, it was found that states which lie close to several coal power plants tend to have high variability in their readings. This was attributed to the dependence of the levels on the wind speed, the wind direction, as well as the atmospheric stability.
\\
A second benefit of modelling spatial volatility is the potential for improving prediction. This was seen by \cite{HWBR2011} when they fitted a Gaussian process with volatility to vegetation and nitrate deposition data. In the case of agriculture yields, prediction intervals accounting for spatial volatility will be useful for insurance companies when they set crop insurance prices \cite[]{Yan2007}. 
\\
Another way of using spatial volatility would be as an indicator of regime change. Such an approach has been taken in desertification and urban planning studies \cite[]{SD2015, Getis2015}. In the first case, regions of high volatility demarcate the bare and the extensive vegetative cover; while in the second case, it is used to identify slum areas. 
\\
In this paper, we introduce stochastic volatility to the well-known Gaussian moving average (GMA) or process convolution model:
\begin{equation}
Y(\mathbf{x}) = \int_{\mathbb{R}^{d}}g(\mathbf{x}-\bm{\xi})W(\mathrm{d}\bm{\xi}), \label{eqn:GMA}
\end{equation}
where $\mathbf{x}\in\mathbb{R}^{d}$ for some $d\in\mathbb{N}$, $g$ is a deterministic (kernel) function and $W$ is the white noise on $\mathbb{R}^{d}$ or a homogeneous standard Gaussian basis whose L\'evy seed (which we shall define in Section \ref{sec:btheory}) has mean $0$ and variance $1$. This results in the so-called volatility modulated moving average (VMMA): 
\begin{equation}
Y(\mathbf{x}) = \int_{\mathbb{R}^{d}}g(\mathbf{x}-\bm{\xi})\sigma(\bm{\xi})W(\mathrm{d}\bm{\xi}), \label{eqn:VMMA}
\end{equation}
where $\{\sigma^{2}(\bm{\xi}):\bm{\xi} \in\mathbb{R}^{d}\}$ is a stationary stochastic volatility field, independent of $W$. In \cite{HWBR2011}, the stochastic volatility is multiplied as a factor to the main spatial process; here, it appears as an integrand. As such, $Y$ can sometimes be identified as a solution to a stochastic partial differential equation. Following similar arguments to those on page 559 of  \cite{Bolin2014}, we find that $Y$ can be viewed as a solution to:
\begin{equation*}
(\kappa^{2} - \Delta)^{\alpha/2} Y(\mathbf{x}) = \sigma(\mathbf{x})\dot{W}(\mathbf{x}),
\end{equation*}
where $\alpha > d/2$, $\kappa > 0$, $\Delta = \sum_{i = 1}^{d}\partial^{2}/\partial x_{i}^{2} $ is the Laplacian operator and $\dot{W}$ is Gaussian white noise, when $g$ is a Mat\'ern kernel defined by:
\begin{equation}
g(\mathbf{x} - \bm{\xi}) = 2^{1 - (\alpha-d)/2}(\kappa|\mathbf{x} - \bm{\xi}|)^{(\alpha - d)/2}K_{(\alpha - d)/2}(\kappa|\mathbf{x} - \bm{\xi}|)/[(4\pi)^{d/2}\Gamma(\alpha/2)\kappa^{\alpha - d}], \label{eqn:Mker}
\end{equation} 
and $K_{(\alpha - d)/2}$ is the modified Bessel function of the second kind. 
\\
VMMAs can be seen as an extension of the Type G L\'evy moving average (LMA) recently studied by \cite{Bolin2014} and \cite{WB2015}: 
\begin{equation}
Y(\mathbf{x}) = \int_{\mathbb{R}^{d}}g(\mathbf{x}-\bm{\xi})L(\mathrm{d}\bm{\xi}), \label{eqn:LMA}
\end{equation}
where $L$ is a (homogeneous) Type G L\'evy basis. This means that the L\'evy seed, $L' \stackrel{d}{=} V^{1/2}Z$ where $V$ is an infinitely divisible random variable and $Z$ is a standard normal random variable independent of $V$.  This is equivalent to restricting $\sigma^2$ in the definition of our VMMA to be infinitely divisible and independent across locations.  
\\
The VMMA is also a special case of the volatility modulated mixed moving average studied by \cite{Veraart2014}:
\begin{align*}
\int_{\mathcal{X}\times\mathbb{R}^{d}}g(z, \mathbf{x}-\bm{\xi})&\sigma(\bm{\xi})W(\mathrm{d}z, \mathrm{d}\bm{\xi}), 
\end{align*}
where $\mathcal{X} \subset \mathbb{R}^{k}$ for some $k\in\mathbb{N}$ and $W$ is a more general Gaussian basis \cite[]{Veraart2014}. It would be useful to study the simulation and inference procedures for the VMMA before moving on to this case where the parameters in the kernel function are randomised.
\\
Convolution models such as GMAs and LMAs have been used in Geostatistics for designing spatial correlation structures. Apart from the classical stationary and isotropic correlation functions, other specifications can be made to construct non-stationary correlations with for example, locally varying geometric anisotoropy \cite[]{FDR2016}. In this paper, we show that VMMAs give us the ability to model stationary and non-stationary correlations since the process is stationary but conditional on the volatility, non-stationary. Specifically, introducing $\sigma$ to a GMA to form a VMMA retains the correlation constructed by our choice of $g$ when we integrate or average over the realisations of $\sigma$. On the other hand, conditional on the realisation of $\sigma$, the VMMA has varying correlation structures over space.
\\
Another usefulness of VMMAs is that they enable us to model clustered extremes which are seen in many environmental data sets. This is because $\sigma(\bm{\xi})$ acts like a local standard deviation for the driving noise $W$ and $\sigma^{2}(\bm{\xi})$ is modelled as a process with correlation. If clustered extremes are not suitable for the context, a good fit of the data to a VMMA could reveal missing covariates in the mean trend, location-dependent explanatory variables or areas where accurate measurements are hard to make.    

\paragraph{Outline}
We begin in Section \ref{sec:btheory} by summarising the $\mathcal{L}_{0}$ integration theory in \cite{RR1989} and providing an integrability condition for the VMMA defined in (\ref{eqn:VMMA}). Next, we develop the theoretical properties of VMMAs in Section \ref{sec:tprop}. The main contributions of our research lie in the following two sections. In Section \ref{sec:sim}, we use discrete convolution ideas to design a simulation algorithm for VMMAs. This is illustrated for a VMMA with a layered structure: the stochastic volatility field is an LMA and the kernels at both the VMMA and LMA levels are Gaussian. We provide a semi-explicit expression for the mean squared error and study cases where an explicit upper bound as well as its order of convergence can be obtained. In Section \ref{sec:est}, we tackle the problem of inference for VMMAs. We develop a two-step moments-matching estimation method which involves a moving window to obtain local variance estimates. Simulation experiments are conducted and the consistency of the estimators is proved under suitable double asymptotics. Next, we apply our method to sea surface temperature anomaly data in Section \ref{sec:emp} to illustrate benefits of using VMMAs instead of GMAs in this case. Finally, we conclude and discuss future steps for action in Section \ref{sec:confurther}. 

\section{The $\mathcal{L}_{0}$ integration and integrability conditions} \label{sec:btheory}

To construct the required stochastic integrals, we first need define our integrator: the homogeneous L\'evy basis.

\subsection{Homogeneous L\'evy bases}

Let $\mathcal{B}(\mathbb{R}^{d})$ denote the Borel $\sigma$-algebra on $\mathbb{R}^{d}$ and $\mathcal{B}_{b}(\mathbb{R}^{d}) = \{E \in \mathcal{B}(\mathbb{R}^{d}): \Leb(E) < \infty\}$ where $\Leb$ represents the Lebesgue measure. We work in the probability space $(\Omega, \mathcal{F}, P)$. To understand what a homogeneous L\'evy basis is, we first define a L\'evy basis \cite[]{BBV2012, Sato2007}: 
\\
\begin{Def}[L\'evy basis] \label{def:Lbasis}\hfill \\
Let $\{E_{i}: i \in \mathbb{N}\}$ be any sequence of the disjoint elements of $\mathcal{B}_{b}(\mathbb{R}^{d})$. Suppose that $L$ is a set of $\mathbb{R}$-valued random variables indexed by such sets, i.e. $\{L(E): E \in \mathcal{B}_{b}(\mathbb{R}^{d})\}$ where, for $\bigcup\limits_{j = 1}^{\infty} E_{j} \in \mathcal{B}_{b}(\mathbb{R}^{d})$, $L(\bigcup\limits_{j = 1}^{\infty} E_{j}) = \sum_{j = 1}^{\infty} L(E_j)$ almost surely. Then, $L$ is called a random measure. 
\\
A random measure $L$ is said to be a L\'evy basis on $(\mathbb{R}^{d}, \mathcal{B}(\mathbb{R}^{d}))$ if:
\begin{enumerate}
\item it is independently scattered: $L(E_{1})$, $L(E_{2})$, ... are independent;
\item and it is infinitely divisible: the random vector $\mathbf{L} = (L(B_{1}), ..., L(B_{m}))$, where $B_{1}, ..., B_{m}$ are elements of $\mathcal{B}_{b}(\mathbb{R}^{d})$, is infinitely divisible. This means that there exists a law $\mu_{n}$, for any $n \in \mathbb{N}$, such that the law of $\mathbf{L}$ can be written as $\mu = \mu_{n}^{*n}$, the n-fold convolution of $\mu_{n}$ with itself.
\end{enumerate}
\end{Def}

Now, we specify a homogeneous L\'evy basis as well as its so-called L\'evy seed using the notion of a cumulant generating function (CGF).  Note that the L\'evy seed is important for defining the distributions of a L\'evy basis and its associated moving average processes. 
\\
\begin{Def}[CGF, homogeneous L\'evy basis and its seed] \hfill \label{defn:lseed} \\
The CGF of a random variable $Z$, which is denoted by $C(\theta; Z)$, is defined as the distinguished logarithm of its characteristic function, i.e. $\log\mathbb{E}\left[\exp\left(i\theta Z\right)\right]$. \\
Let $L$ be a L\'evy basis. Suppose that there exists a random variable $L'$ such that $C(\theta ; L(E)) = C(\theta ; L') \Leb(E)$ for all $E\in\mathcal{B}_{b}(S)$. Then, we say that $L$ is homogeneous and $L'$ is its L\'evy seed. 
\end{Def}
\vspace{5mm}
\begin{Eg}
The standard Gaussian basis in (\ref{eqn:VMMA}) is a homogeneous L\'evy basis. In this case, $W(E) \sim N(0, \Leb(E))$ for any $E\in \mathcal{B}_{b}(\mathbb{R}^{d})$.
\end{Eg}
\vspace{5mm}
\begin{Eg}
In this paper, we will also use the inverse Gaussian (IG) L\'evy basis in our simulation studies. The parameterisation chosen is such that, for $z > 0$, $\delta >0$ and $\gamma >0$, the probability density function of a random variable $Z$ with an $IG\big(\delta, \gamma)$ distribution is:
\begin{equation*}
f(z; \delta, \gamma) = \frac{\delta}{\sqrt{2\pi z^{3}}} \exp\left({\delta\gamma - \frac{1}{2}\left(\frac{\delta^{2}}{z} + \gamma^{2}z\right)}\right).
\end{equation*}
In this case, if $L$ is an IG basis whose seed has an $IG\big(\delta, \gamma)$ distribution, $L(E) \sim IG\big(\delta\Leb(E), \gamma\big)$ for any $E\in \mathcal{B}_{b}(\mathbb{R}^{d})$.
\end{Eg}

\subsection{Summary of the theory in \cite{RR1989}}

Since we only deal with homogeneous L\'evy bases for the GMAs, LMAs and VMMAs, we present the integration theory for this case. As usual, a stochastic integral is built up as a limit of those defined by so-called simple functions:
\\
\begin{Def}[Simple function on $\mathbb{R}^{d}$, the stochastic integral of a simple function] \hfill \\
Consider $\{y_{j} \in \mathbb{R} : j = 1, ..., n\}$ and $\{E_{j}: j = 1,..., n\}$, a collection of disjoint sets of $\mathcal{B}_{b}(\mathbb{R}^{d})$. Then, $f(\mathbf{x}) = \sum_{j = 1}^{n} y_{j}\mathbf{1}_{E_{j}}(\mathbf{x})$ is called a simple function on $\mathbb{R}^{d}$ where $\mathbf{1}_{E_{j}}(\mathbf{x}) = 1$ if $\mathbf{x} \in E_{j}$ and $0$ otherwise.
\\
The stochastic integral of $f$ over $A \in \mathcal{B}(\mathbb{R}^{d})$ is defined as $\int_{A}f(\bm{\xi})L(\mathrm{d}\bm{\xi}) = \sum_{j = 1}^{n} y_{j} L(A \cap E_{j})$.
\end{Def}
The stochastic integral of a measurable function is a simple extension of this:
\\
\begin{Def}[$L$-integrability and the stochastic integral of a measurable function] \hfill \\
Let $f:(\mathbb{R}^{d}, \mathcal{B}(\mathbb{R}^{d})) \rightarrow (\mathbb{R}, \mathcal{B}(\mathbb{R}))$  be a measurable function. Then, $f$ is $L$-integrable if there exists a sequence of simple functions $\{f_{m}\}$ such that:
\begin{itemize}
\item[(i)] $f_{m}$ converges to $f$ almost everywhere with respect to the Lebesgue measure;
\item[(ii)] the sequence $\{\int_{A}f_{m}(\bm{\xi})L(\mathrm{d}\bm{\xi})\}$ converges in probability for every $A \in \mathcal{B}(\mathbb{R}^{d})$. 
\end{itemize}

For an $L$-integrable function $f$, we define:
\begin{equation*}
\int_{A}f(\bm{\xi}) L(\mathrm{d}\bm{\xi}) = P-\lim_{m\rightarrow\infty} \int_{A} f_{m}(\bm{\xi}) L(\mathrm{d}\bm{\xi}).
\end{equation*}
This construction is well-defined because the limit does not depend on the sequence $\{f_{m}\}$. 
\end{Def}

Theorem 2.7 in \cite{RR1989} provides us with explicit conditions for $\mathcal{L}_{0}$ integrability. When the kernel $g$ is Lebesgue integrable and square-integrable, these conditions are satisfied for GMAs and LMAs whose L\'evy bases have finite second moments. To use the $\mathcal{L}_{0}$ theory to construct VMMA, we condition on the realisation of $\sigma^{2}$ and treat $g(\mathbf{x} - \bm{\xi})\sigma(\bm{\xi})$ as a measurable function. The condition required for a well-defined VMMA is then given by:
\\
\begin{Con}\label{con1}
$\int_{\mathbb{R}^{d}}g^{2}(\mathbf{x} - \bm{\xi})\sigma^{2}(\bm{\xi})<\infty$. 
\end{Con}
\vspace{5mm}
As will be shown later, this quantity is equal to the conditional variance of $Y$ at $\mathbf{x}$. This is easy to see that the condition holds whenever $g$ is square-integrable and $\sigma^{2}$ takes finite values.
\\
\begin{Eg}[Two-tiered model]
An example of a well-defined VMMA is the so-called two-tiered model:
\begin{equation}
\left.\begin{aligned}
Y(\mathbf{x})&= \int_{\mathbb{R}^{d}}g(\mathbf{x}-\bm{\xi})\sigma(\bm{\xi})W(\mathrm{d}\bm{\xi}),\\
\text{where } \sigma^{2}(\bm{\xi})&= \int_{\mathbb{R}^{d}}h(\bm{\xi}-\mathbf{u})L(\mathrm{d}\mathbf{u}). 
\end{aligned}
\right\}
\qquad \label{eqn:twotier}
\end{equation}
Here, $g$ and $h$ are Lebesgue integrable and square-integrable kernel functions, and $L$ is a subordinator with finite second moments (so that $\sigma^{2}$ is well-defined). As before, $W$ is a homogeneous standard Gaussian basis independent of $\sigma^{2}$. 
\\
Note that we model $\sigma^{2}$ directly. In comparison, treating the volatility as a multiplicative factor as is done in \cite{HWBR2011} is synonymous with modelling the conditional variance $\int_{\mathbb{R}^{d}} g^{2}(\mathbf{x}-\bm{\xi})\sigma^{2}(\bm{\xi}) \mathrm{d}\bm{\xi}$. In this case, we have chosen to model $\sigma^{2}$ as an LMA because it is convenient for deriving the second order distributional properties of $\sigma^{2}$. When the VMMA can be viewed as a solution to an SPDE, $\sigma^{2}$ could correspond to a process of special interest. For example, it could represent the cumulative effect which the wind speed, direction and atmospheric stability has on the spatial heteroskedasticity of air pollution. Thus, in Section \ref{sec:est}, we are particularly interested in estimating the parameters of $\sigma^{2}$ that determine its variance and correlation structure. 
\end{Eg}
\vspace{5mm}
\begin{Eg} \label{eg:weg}
As our main illustration example, we will use a more precise model:
\begin{equation}
\left.\begin{aligned}
Y(\mathbf{x})&= \int_{\mathbb{R}^{2}}\frac{\lambda}{\pi}\exp\left(-\lambda\left(\mathbf{x}-\bm{\xi}\right)^{T}\left(\mathbf{x}-\bm{\xi}\right)\right)\sigma(\bm{\xi})W(\mathrm{d}\bm{\xi}),\\
\text{where } \sigma^{2}(\bm{\xi})&= \int_{\mathbb{R}^{2}}\frac{\eta}{\pi}\exp\left(-\eta\left(\bm{\xi}-\mathbf{u}\right)^{T}\left(\bm{\xi}-\mathbf{u}\right)\right)L(\mathrm{d}\mathbf{u}).
\end{aligned}
\right\}
\qquad \label{eqn:weg}
\end{equation}
The stochastic volatility field is an LMA with the same kernel structure as the VMMA itself but a different rate parameter. Here, we choose Gaussian kernels with rate parameters $\lambda$, $\eta > 0$. As mentioned in \cite{Higdon1998}, these kernels are computationally convenient and are supported by physical ocean dynamics. 
\end{Eg}

\section{Theoretical properties of VMMAs} \label{sec:tprop}

In this section, we prove several distributional properties of VMMAs including stationarity, cumulant and covariance structures. These will be useful for the estimation method which we develop later. 

\subsection{Marginal distribution}

\subsubsection{Conditional distribution and cumulants}

Let $\mathcal{F}^{\sigma}$ be the $\sigma$-algebra generated by the stochastic volatility $\sigma^{2}$. As we have assumed that $\sigma$ and $W$ are independent, $Y(\mathbf{x})|\mathcal{F}^{\sigma} \sim N(0, \sigma^{2}_{I}(\mathbf{x}))$, where $\sigma^{2}_{I}(\mathbf{x}) = \int_{\mathbb{R}^{d}} g^{2}(\mathbf{x}-\bm{\xi})\sigma^{2}(\bm{\xi})\mathrm{d}\bm{\xi}$ denotes the conditional variance. 
\\
Recall that the cumulants of the VMMA, $\kappa_{l}\left(Y\left(\mathbf{x}\right)\right)$, are defined though its CGF. That is, $C(\theta; Y) = \sum_{l = 1}^{\infty} \kappa_{l}\left(Y\left(\mathbf{x}\right)\right)\frac{(i\theta)^{l}}{l!}$. In this case, since we have a Gaussian distribution, the conditional cumulants are $\kappa_{1}^{\sigma} = \mathbb{E}\left[Y\left(\mathbf{x}\right)|\mathcal{F}^{\sigma}\right] = 0$, $\kappa_{2}^{\sigma} = \Var\left[Y\left(\mathbf{x}\right)|\mathcal{F}^{\sigma}\right] = \sigma^{2}_{I}(\mathbf{x})$ and $\kappa_{l}^{\sigma} = 0$ for $l\geq 3$.

\subsubsection{Unconditional distribution and cumulants}
\label{section:unconddistr}

From the first conditional cumulant, we get $\kappa_{1} = \mathbb{E}\left[Y\left(\mathbf{x}\right)\right] = \mathbb{E}\left[\kappa_{1}^{\sigma}\right] = 0$. So, $\kappa_{2} = \mathbb{E}\left[Y^{2}\left(\mathbf{x}\right)\right] = \mathbb{E}\left[\sigma^{2}_{I}(\mathbf{x})\right]$. Higher order unconditional cumulants can be calculated in similar ways; however, beyond the third cumulant, they are typically not be equal to $0$ unlike their conditional counterparts. This is because the unconditional marginal distribution of a VMMA is generally not Gaussian. 
\\
\begin{Eg} \label{eg:cum}
For Model (\ref{eqn:weg}), we have $\kappa_{1} = 0$. Let $\mathbb{E}[L'] = a < \infty$. The next three cumulants are given by:
\begin{align*}
\kappa_{2}^{\sigma} = \sigma^{2}_{I}(\mathbf{x}) &\Rightarrow \kappa_{2} = \int_{\mathbb{R}^{2}} g^{2}(\mathbf{x}-\bm{\xi})\mathbb{E}\left[\sigma^{2}(\bm{\xi})\right] \mathrm{d}\bm{\xi} = a \int_{\mathbb{R}^{2}} \frac{\lambda^{2}}{\pi^{2}} \exp\left(-2\lambda\left(\mathbf{x}-\bm{\xi}\right)^{T}\left(\mathbf{x}-\bm{\xi}\right)\right)\mathrm{d}\bm{\xi} = \frac{a\lambda}{2\pi}. \\
\kappa^{\sigma}_{3} = 0 &\Rightarrow \mathbb{E}\left[Y^{3}(\mathbf{x})|\mathcal{F}^{\sigma}\right] - 3 \mathbb{E}\left[Y^{2}(\mathbf{x})|\mathcal{F}^{\sigma}\right] \mathbb{E}\left[Y(\mathbf{x})|\mathcal{F}^{\sigma}\right] + 2 \mathbb{E}^{3}\left[Y(\mathbf{x})|\mathcal{F}^{\sigma}\right]= 0. \Rightarrow \mathbb{E}\left[Y^{3}(\mathbf{x})\right] = \mathbb{E}\left[Y^{3}(\mathbf{x})|\mathcal{F}^{\sigma}\right] = 0. \\
&\Rightarrow \kappa_{3} = \mathbb{E}\left[Y^{3}(\mathbf{x})\right] - 3 \mathbb{E}\left[Y^{2}(\mathbf{x})\right] \mathbb{E}\left[Y(\mathbf{x})\right] + 2 \mathbb{E}^{3}\left[Y(\mathbf{x})\right] = 0. \\
\kappa^{\sigma}_{4} = 0 &\Rightarrow \mathbb{E}\left[Y^{4}(\mathbf{x})|\mathcal{F}^{\sigma}\right] - 4 \mathbb{E}\left[Y^{3}(\mathbf{x})|\mathcal{F}^{\sigma}\right] \mathbb{E}\left[Y(\mathbf{x})|\mathcal{F}^{\sigma}\right] - 3 \mathbb{E}^{2}\left[Y^{2}(\mathbf{x})|\mathcal{F}^{\sigma}\right] + 12 \mathbb{E}\left[Y^{2}(\mathbf{x})|\mathcal{F}^{\sigma}\right] \mathbb{E}^{2}\left[Y(\mathbf{x})|\mathcal{F}^{\sigma}\right] \\
& ~~~~~- 6 \mathbb{E}^{4}\left[Y(\mathbf{x})|\mathcal{F}^{\sigma}\right] = 0. \Rightarrow \mathbb{E}\left[Y^{4}(\mathbf{x})|\mathcal{F}^{\sigma}\right] = 3 \mathbb{E}^{2}\left[Y^{2}(\mathbf{x})|\mathcal{F}^{\sigma}\right] = 3(\sigma^{2}_{I}(\mathbf{x}))^{2}. \\
&\Rightarrow \kappa_{4} = \mathbb{E}\left[Y^{4}(\mathbf{x})\right] - 3 \mathbb{E}^{2}\left[Y^{2}(\mathbf{x})\right] = 3 \Var\left[\sigma^{2}_{I}(\mathbf{x})\right] = \frac{3b\lambda^{3}\eta}{4\pi^{3}(2\lambda + \eta)}, \text{ where } b = \Var(L') \text{ and we have used } (\ref{eqn:VCov}).
\end{align*}
\end{Eg}

\subsection{Finite dimensional distributions}

\subsubsection{Conditional joint distribution and correlation structure}

Now, we consider joint distributions of the process at different locations. This is characterised by the joint cumulant generating function (JCGF). To compute this for the VMMA, we introduce the concept of a generalised cumulant functional. This is a spatial extension of the concept given in \cite{BBV2012}:
\\
\begin{Def}[Generalised cumulant functional] \hfill \\
Let $Y = \{Y(\mathbf{x}): \mathbf{x}\in\mathbb{R}^{d}\}$ denote a stochastic process in $\mathbb{R}^{d}$, and let $v$ denote any non-random measure such that $v(Y) = \int_{\mathbb{R}^{d}} Y(\mathbf{x}) v(\mathrm{d}\mathbf{x})$, where the integral exists almost surely. The generalised cumulant functional (GCF) of $Y$ with respect to $v$ is given by: $C(\theta ; v(Y)) = \log\mathbb{E}\left[\exp\left(i\theta v\left(Y\right)\right)\right]$.
\end{Def}
To compute the JCGF of a VMMA, Y, we first condition on $\sigma^{2}(\bm{\xi})$ and obtain the conditional GCF: 
\\
\begin{thm}
\label{thm:GCF}
Let $Y(\mathbf{x})$ be a VMMA defined by (\ref{eqn:VMMA}). Assume that for all $\bm{\xi} \in \mathbb{R}^{d}$, $h(\bm{\xi}) = \int_{\mathbb{R}^{d}} g(\mathbf{x}-\bm{\xi})\sigma(\bm{\xi}) v(\mathrm{d}\mathbf{x}) < \infty$, and that $h(\bm{\xi})$ is integrable with respect to the Gaussian basis $W$. Then,  with $W'$ denoting the seed of $W$, the GCF of $Y|\sigma^{2}$ with respect to $v$ can be expressed as:
\begin{equation*}
C(\theta ; v(Y)|\mathcal{F}^{\sigma}) = \int_{\mathbb{R}^{d}}C(\theta h(\bm{\xi}) ; W') \mathrm{d}\bm{\xi} = - \frac{1}{2}\theta^{2}\int_{\mathbb{R}^{d}} h^{2}(\bm{\xi}) \mathrm{d}\bm{\xi}.
\end{equation*}
\end{thm}
\begin{proof}
This is analagous to the proof for Proposition 5 in \cite{BBV2012} with $h$ being defined differently to account for our definition of $Y(\mathbf{x})$, and with the L\'evy basis restricted to be standard Gaussian.
\end{proof}
Now, we use $v(\mathrm{d}\mathbf{x}) = \sum_{j = 1}^{n}\theta_{j}\delta_{\mathbf{x}_{j}}(\mathrm{d}\mathbf{x})$ so that $C(\theta ; v(Y) | \mathcal{F}^{\sigma})$ is the JCGF of $Y(\mathbf{x}_{1}), \dots, Y(\mathbf{x}_{n}) |\mathcal{F}^{\sigma}$:
\\
\begin{Cor} \label{cor:jcum}
Let $\mathbf{x}_{1}, \dots, \mathbf{x}_{n}$ be different locations in $\mathbb{R}^{d}$. The JCGF of $Y(\mathbf{x}_{1}), \dots, Y(\mathbf{x}_{n})|\mathcal{F}^{\sigma}$ is given by:
\begin{equation*}
\log\mathbb{E}\left[\exp\left(i \sum_{j = 1}^{n} \tilde{\theta}_{j}Y(\mathbf{x}_{j})\right)|\mathcal{F}^{\sigma}\right] = - \frac{1}{2}\sum_{j = 1}^{n} \sum_{k = 1}^{n}\tilde{\theta}_{j}\tilde{\theta}_{k}\int_{\mathbb{R}^{d}} g(\mathbf{x}_{j}-\bm{\xi})g(\mathbf{x}_{k} - \bm{\xi})\sigma^{2}(\bm{\xi}) \mathrm{d}\bm{\xi}.
\end{equation*}
\end{Cor}
This means that $Y(\mathbf{x}_{1}), \dots, Y(\mathbf{x}_{n})|\mathcal{F}^{\sigma}\sim N_{n}(\mathbf{0}, \Sigma)$ where $\Sigma_{jk} = \int_{\mathbb{R}^{d}}g(\mathbf{x}_{j} - \bm{\xi})g(\mathbf{x}_{k} - \bm{\xi})\sigma^{2}(\bm{\xi})\mathrm{d}\bm{\xi}$ for $j, k = 1, \dots, n$. 
\begin{proof}
We first compute $h(\bm{\xi})$:
\begin{equation*}
h(\bm{\xi}) = \int_{\mathbb{R}^{d}} g(\mathbf{x}-\bm{\xi})\sigma(\bm{\xi}) v(\mathrm{d}\mathbf{x}) = \int_{\mathbb{R}^{d}} g(\mathbf{x}-\bm{\xi})\sigma(\bm{\xi}) \sum_{j = 1}^{n}\theta_{j}\delta_{\mathbf{x}_{j}}(\mathrm{d}\mathbf{x}) = \sum_{j = 1}^{n}\theta_{j}g(\mathbf{x}_{j}-\bm{\xi})\sigma(\bm{\xi}). 
\end{equation*}
With $\tilde{\theta}_{j} = \theta\theta_{j}$, the JCGF of $Y(\mathbf{x}_{1}), \dots, Y(\mathbf{x}_{n})|\mathcal{F}^{\sigma}$ is given by:
\begin{equation}
\log\mathbb{E}\left[\exp\left(i \sum_{j = 1}^{n} \tilde{\theta}_{j}Y(\mathbf{x}_{j})\right)|\mathcal{F}^{\sigma}\right] = - \frac{1}{2}\theta^{2}\int_{\mathbb{R}^{d}} h^{2}(\bm{\xi}) \mathrm{d}\bm{\xi} = - \frac{1}{2}\sum_{j = 1}^{n} \sum_{k = 1}^{n}\tilde{\theta}_{j}\tilde{\theta}_{k}\int_{\mathbb{R}^{d}} g(\mathbf{x}_{j}-\bm{\xi})g(\mathbf{x}_{k} - \bm{\xi})\sigma^{2}(\bm{\xi}) \mathrm{d}\bm{\xi}. \label{eqn:CJCGF}
\end{equation}
This corresponds to a multivariate normal distribution with the parameters stated in the Theorem. 
\end{proof}
From the multivariate normal distribution, we can infer the covariance and correlation structures of our VMMA. Let $\mathbf{x}$ and $\mathbf{x}^{*}$ be two different locations in $\mathbb{R}^{d}$, then $\Cov(Y(\mathbf{x}), Y(\mathbf{x}^{*})|\mathcal{F}^{\sigma}) = \int_{\mathbb{R}^{d}}g(\mathbf{x} - \bm{\xi})g(\mathbf{x}^{*} - \bm{\xi})\sigma^{2}(\bm{\xi})\mathrm{d}\bm{\xi}$ and:
\begin{equation*}
\Cov(Y(\mathbf{x}), Y(\mathbf{x}^{*})) = \int_{\mathbb{R}^{d}}g(\mathbf{x} - \bm{\xi})g(\mathbf{x}^{*} - \bm{\xi})\mathbb{E}\left[\sigma^{2}(\bm{\xi})\right]\mathrm{d}\bm{\xi} = \mathbb{E}\left[\sigma^{2}(\mathbf{0})\right]\int_{\mathbb{R}^{d}}g(\mathbf{w})g(\mathbf{x}^{*} - \mathbf{x} + \mathbf{w})\mathrm{d}\mathbf{w},
\end{equation*}
since $\sigma^{2}(\bm{\xi})$ is stationary and where $\mathbf{w} = \mathbf{x} - \bm{\xi}$. As this is a function of the location difference and not the locations themselves, $Y(\mathbf{x})$ has second-order stationarity. From the covariance function, we also find that the correlation structure does not depend on the stochastic volatility:
\\
\begin{equation*}
\Corr(Y(\mathbf{x}), Y(\mathbf{x}^{*})) = \frac{\int_{\mathbb{R}^{d}}g(\mathbf{w})g(\mathbf{x}^{*} - \mathbf{x} + \mathbf{w})\mathrm{d}\mathbf{w}}{\int_{\mathbb{R}^{d}}g^{2}(\mathbf{w})\mathrm{d}\mathbf{w}}. 
\end{equation*}
An LMA with zero mean and the same kernel will also have this correlation structure. The effects of the stochastic volatility in the VMMA, and hence a difference from the GMA and the LMA, lies in the higher order correlations. 
\\
\begin{Eg}\label{eg:wegCorr}
For Model (\ref{eqn:weg}), we have $g(\mathbf{w}) = \frac{\lambda}{\pi}\exp\left(-\lambda \mathbf{w}^{T}\mathbf{w}\right)$. By completing the squares:
\begin{align*}
\int_{\mathbb{R}^{2}}g(\mathbf{w})g(\mathbf{x}^{*} - \mathbf{x} + \mathbf{w})\mathrm{d}\mathbf{w} &=\frac{\lambda^{2}}{\pi^{2}}\int_{\mathbb{R}^{2}}\exp\left(-\lambda \mathbf{w}^{T}\mathbf{w}\right)\exp\left(-\lambda \left(\mathbf{x}^{*} - \mathbf{x} + \mathbf{w}\right)^{T}\left(\mathbf{x}^{*} - \mathbf{x} + \mathbf{w}\right)\right)\mathrm{d}\mathbf{w} \\
&= \frac{\lambda^{2}}{\pi^{2}}\int_{\mathbb{R}} \exp\left(-\lambda \left[w_{1}^{2} + (w_{1} - x_{1} + x^{*}_{1})^{2}\right]\right) \mathrm{d}w_{1} \int_{\mathbb{R}} \exp\left(-\lambda \left[w_{2}^{2} + (w_{2} - x_{2} + x^{*}_{2})^{2}\right]\right) \mathrm{d}w_{2} \\
&= \frac{\lambda}{2\pi}\exp\left(- \frac{\lambda\left(\mathbf{x}-\mathbf{x}^{*}\right)^{T}\left(\mathbf{x}-\mathbf{x}^{*}\right)}{2}\right). \\
\Rightarrow \Corr(Y(\mathbf{x}), Y(\mathbf{x}^{*})) &= \exp\left(- \frac{\lambda\left(\mathbf{x}-\mathbf{x}^{*}\right)^{T}\left(\mathbf{x}-\mathbf{x}^{*}\right)}{2}\right).
\end{align*}
\end{Eg}
\vspace{5mm}
\begin{Cor}
Let $Y(\mathbf{x})$ be a VMMA and $\mathbf{x}, \mathbf{x}^{*}$ denote two arbitrary locations in $\mathbb{R}^{d}$. Then:
\begin{align*}
\Cov(Y^{2}(\mathbf{x}), Y^{2}(\mathbf{x}^{*})|\mathcal{F}^{\sigma}) &= \mathbb{E}\left[Y^{2}(\mathbf{x})Y^{2}(\mathbf{x}^{*})|\mathcal{F}^{\sigma}\right] - \sigma^{2}_{I}(\mathbf{x})\sigma^{2}_{I}(\mathbf{x}^{*}) = 2\left(\int_{\mathbb{R}^{d}} g(\mathbf{x}-\bm{\xi})g(\mathbf{x}^{*}-\bm{\xi})\sigma^{2}(\bm{\xi})\mathrm{d}\bm{\xi}\right)^{2}, \\
\text{and } \Cov(Y^{2}(\mathbf{x}), Y^{2}(\mathbf{x}^{*})) &= \mathbb{E}\left[\Cov(Y^{2}(\mathbf{x}), Y^{2}(\mathbf{x}^{*})|\mathcal{F}^{\sigma})\right] + \Cov(\sigma^{2}_{I}(\mathbf{x}), \sigma^{2}_{I}(\mathbf{x}^{*})) \\
&= 2\mathbb{E}\left[\left(\int_{\mathbb{R}^{d}} g(\mathbf{x}-\bm{\xi})g(\mathbf{x}^{*}-\bm{\xi})\sigma^{2}(\bm{\xi})\mathrm{d}\bm{\xi}\right)^{2}\right] + \Cov(\sigma^{2}_{I}(\mathbf{x}), \sigma^{2}_{I}(\mathbf{x}^{*})).
\end{align*}
\end{Cor}
\begin{proof}
We calculate $\mathbb{E}\left[Y^{2}(\mathbf{x})Y^{2}(\mathbf{x}^{*})|\mathcal{F}^{\sigma}\right]$ by setting $n = 2$, differentiating the conditional JCGF in (\ref{eqn:CJCGF}) with respect to $\tilde{\theta}_{1}$ and $\tilde{\theta}_{2}$ twice each, and setting these to be equal to $0$. The rest follows easily.
\end{proof}

\begin{Eg}\label{eg:2oCov}
For Model (\ref{eqn:weg}):
\begin{equation*}
\Cov(Y^{2}(\mathbf{x}), Y^{2}(\mathbf{x}^{*})) =A\exp\left(-\lambda \left(\mathbf{x} - \mathbf{x}^{*}\right)^{T}\left(\mathbf{x} - \mathbf{x}^{*}\right)\right) + B\exp\left(\frac{-\lambda\eta}{2\lambda + \eta}\left(\mathbf{x} - \mathbf{x}^{*}\right)^{T}\left(\mathbf{x} - \mathbf{x}^{*}\right)\right),
\end{equation*}
where $A =  (b\lambda\eta + a^{2}(2\lambda + \eta)\pi)\lambda^{2}(2\pi^{3}(2\lambda+ \eta)^{-1}$ and $B = b\lambda^{3}\eta(4\pi^{3}(2\lambda + \eta))^{-1}$. The details of the computation can be found in the Appendix.
\end{Eg}

\subsection{Unconditional joint distribution and stationarity}

By exponentiating the expression in Corollary \ref{cor:jcum} for the conditional JCGF and taking expectations with respect to $\sigma^{2}$, we obtain the unconditional joint characteristic function (JCF) of $Y(\mathbf{x}_{1}), \dots, Y(\mathbf{x}_{n})$: 
\begin{equation*}
\mathbb{E}\left[\exp\left(i \sum_{j = 1}^{n} \tilde{\theta}_{j}Y(\mathbf{x}_{j})\right)\right] = \mathbb{E}\left[\exp\left(- \frac{1}{2}\sum_{j = 1}^{n} \sum_{k = 1}^{n}\tilde{\theta}_{j}\tilde{\theta}_{k}\int_{\mathbb{R}^{d}} g(\mathbf{x}_{j}-\bm{\xi})g(\mathbf{x}_{k} - \bm{\xi})\sigma^{2}(\bm{\xi}) \mathrm{d}\bm{\xi}\right)\right]. 
\end{equation*}
For specific $g$ and $\sigma^{2}$, this can be expressed as a function of the location differences since $Y$ is stationary:
\\
\begin{thm} \label{thm:station}
Let $Y(\mathbf{x})$ be a VMMA. Then, $Y(\mathbf{x})$ is a stationary process in $\mathbb{R}^{d}$. 
\end{thm}
\begin{proof}
We present the proof for $\mathbf{x}$ in $\mathbb{R}$. The case for general $\mathbb{R}^{d}$ follows analogously with more involved notation. Let $x_{0} <  \dots < x_{n-1}$ denote $n$ arbitrary locations in $\mathbb{R}$. We show that for any $u\in\mathbb{R}$, the JCF of $Y(x_{0} + u), \dots, Y(x_{n-1} + u)$ is the same as that of $Y(x_{0}), \dots, Y(x_{n-1})$: 
\begin{align}
&\mathbb{E}\left[\exp\left(i \sum_{j = 0}^{n-1}\theta_{j}\int_{\mathbb{R}}g(x_{j} + u - \xi)\sigma(\xi)W(\mathrm{d}\xi)\right)\right] \nonumber\\
&= \mathbb{E}\left[ \mathbb{E}\left[\exp\left(i \sum_{j = 0}^{n-1}\theta_{j}\int_{\mathbb{R}}g(x_{j} + u - \xi)\sigma(\xi)W(\mathrm{d}\xi)\right)|\mathcal{F}^{\sigma}\right] \right] \nonumber \\
&= \mathbb{E}\left[ \mathbb{E}\left[\exp\left(i \sum_{j = 0}^{n-1}\theta_{j}\lim_{p\rightarrow\infty} \sum_{I =2}^{k_{p}-1}  g(x_{j} + u - y_{I}) \sigma(y_{I}) W\left(\left(y_{I} - (y_{I} - y_{I-1})/2, y_{I} + (y_{I+1} - y_{I})/2\right]\right)\right)|\mathcal{F}^{\sigma}\right] \right], 
\end{align}
where we follow the $\mathcal{L}_{0}$ integration theory and use an approximating sequence for $g(x_{j} + u - \xi)\sigma(\xi)$:
\begin{equation*}
f_{p}(\xi) = \sum_{I =2}^{k_{p}-1}  g(x_{j} + u - y_{I}) \sigma(y_{I}) \mathbf{1}_{(y_{I} - (y_{I} - y_{I-1})/2, y_{I} + (y_{I+1} - y_{I})/2]}(\xi),
\end{equation*}
where $(y_{k_{p}})$ is a sequence of partitions such that:
\begin{equation*}
-\infty < y_{1} < y_{2} < \dots < y_{k_{p}} < \infty,
\end{equation*}
$k_{p} \rightarrow \infty$ and $\max_{I \in \{2, \dots, k_{p}\}}\left(y_{I} - y_{I-1}\right) \rightarrow 0$ as $p\rightarrow \infty$. Note that we have assumed that $g$ takes finite values so that this approximating sequence can be evaluated. If $g$ has a finite number of singularities, for example in the case of the Mat\'ern kernel (\ref{eqn:Mker}) with $(\alpha-d)/2) \in (-1/2, 0)$, shifts can be made to the evaluation points to avoid these. \\
Define another sequence of partitions $(z_{k_{p}}) = (y_{k_{p}} - u)$. By changing the order of taking limits and sums, and using $h(z_{I}) = \sum_{j = 0}^{n-1}\theta_{j} g(x_{j} - z_{I}) \sigma(z_{I} + u)$, the JCF can be written as:
\begin{align}
&\lim_{p\rightarrow\infty} \mathbb{E}\left[ \mathbb{E}\left[\exp\left(i \sum_{I =2}^{k_{p}-1} h(z_{I}) W\left(\left(z_{I} + u - (z_{I} - z_{I-1})/2, z_{I} + u + (z_{I+1} - z_{I})/2 \right]\right)\right)|\mathcal{F}^{\sigma}\right] \right] \nonumber \\
&= \lim_{p\rightarrow\infty} \mathbb{E}\left[ \prod\limits_{I =2}^{k_{p}-1} \mathbb{E}\left[\exp\left(i h(z_{I}) W\left(\left(z_{I} + u - (z_{I} - z_{I-1})/2, z_{I} + u + (z_{I+1} - z_{I})/2 \right]\right)\right)|\mathcal{F}^{\sigma}\right] \right] \label{eqn:prodrep} \\
&= \lim_{p\rightarrow\infty} \mathbb{E}\left[ \prod\limits_{I =2}^{k_{p}-1} \exp\left(-\frac{z_{I} + (z_{I+1} - z_{I})/2 - z_{I} + (z_{I} - z_{I-1})/2)}{2} h^{2}(z_{I}) \right)\right] \label{eqn:simpprod}
\end{align}
where (\ref{eqn:prodrep}) and (\ref{eqn:simpprod}) hold because $W$ is independently scattered and homogeneous standard Gaussian.\\
Recall that $\sigma^{2}$ is stationary. Since the term inside the expectation in (\ref{eqn:simpprod}) is a Borel transformation of $(\sigma(z_{1} + u), \dots, \sigma(z_{k_{p}} + u))$, it has the same distribution and expectation as a similar expression with $(\sigma(z_{1}), \dots, \sigma(z_{k_{p}}))$ instead. Since (\ref{eqn:simpprod}) no longer depends on $u$, we conclude that $Y$ is stationary.
\end{proof}

\section{Simulation} \label{sec:sim}

\subsection{A discrete convolution algorithm}

We focus on cases in $\mathbb{R}^{2}$ where the kernel function in the VMMA takes finite values. Let $\mathbf{x} = \{(x_{1}^{i}, x_{2}^{j}) = (x_{1}^{*} + i\triangle, x_{2}^{*} + j\triangle): i, j = -p, \dots, N+p-1\}$ be our simulation grid where $(x_{1}^{*}, x_{2}^{*})$ is the starting point, $\triangle$ is the grid size, $p \in \mathbb{N}$ is a kernel truncation parameter and $N\in\mathbb{N}$ is the number of coordinates in each spatial axis so that the final sample size is $N^{2}$. By discretizing the stochastic integral in (\ref{eqn:VMMA}), we can view a VMMA as a filtered process where $g(\mathbf{x} - \bm{\xi})$ is the kernel or filter and $\sigma(\bm{\xi})W^{*}(\bm{\xi})$ is the signal. Here, $W^{*}(\bm{\xi}) \stackrel{d}{=} N(0, \triangle^{2})$ is a random variable representing the Gaussian noise over the grid square centered at $\bm{\xi}$. These $W^{*}$s are independent across locations. In practice, this means that we approximate our VMMA by:
\begin{equation}
Y(x_{1}^{i}, x_{2}^{j}) \approx \sum_{I = i-p}^{i+p}  \sum_{J = j-p}^{j+p} g({x}_{1}^{i} - x_{1}^{I}, {x}_{2}^{j} - x_{2}^{J}) \sigma(x_{1}^{I}, x_{2}^{J}) W^{*}(x_{1}^{I}, x_{2}^{J}), \label{eqn:DCapprox}
\end{equation}
for $i, j = 0, \dots, N-1$. If $g$ is square-integrable, $g(\mathbf{x}-\bm{\xi})$ typically decreases very fast to $0$ as $|\mathbf{x} - \bm{\xi}| \rightarrow 0$. Thus, only small errors are incurred by truncating the kernel. Following the $\mathcal{L}_{0}$ integration theory and the proof of Theorem \ref{thm:station}, (\ref{eqn:DCapprox}) can also be viewed as an approximation of $Y$ when $g$ is approximated using a particular equispaced partition and the increments of $W$ are replaced by independent, identically distributed random variables $W^{*}$. \\
Suppose that we know the values of $\sigma$ over $\mathbf{x}$, then we can generate values for $\mathbf{Y} = \{Y(x_{1}^{i}, x_{2}^{j}): i, j = 0, \dots, N-1\}$ using (\ref{eqn:DCapprox}). To begin, we create a $(2p+1)\times(2p+1)$ kernel matrix $K$ as follows:
\begin{equation}
K = \begin{pmatrix}
g(p\triangle, -p\triangle) & \dots & g(-p\triangle, -p\triangle) \\
\vdots & \vdots & \vdots \\
\dots & g(0, 0) & \dots \\
\vdots & \vdots & \vdots \\
g(p\triangle, p\triangle) & \dots & g(-p\triangle, p\triangle)
\end{pmatrix}.
\end{equation}
Then, we generate $\{W^{*}(x_{1}^{i}, x_{2}^{j}) \}$ for $i, j = -p, \dots, N+p - 1$. We multiply each of these $W^{*}$s to their corresponding $\sigma$ values and create a $(N + 2p)\times(N+2p)$ signal matrix:
\begin{equation}
\sigma W^{*} = \begin{pmatrix}
\sigma(x_{1}^{-p}, x_{2}^{N+p-1}) W^{*}(x_{1}^{-p}, x_{2}^{N+p-1})  & \dots & \sigma(x_{1}^{N+p-1}, x_{2}^{N+p-1}) W^{*}(x_{1}^{N+p-1}, x_{2}^{N+p-1})\\
\vdots & \vdots & \vdots \\
\sigma(x_{1}^{-p}, x_{2}^{-p}) W^{*}(x_{1}^{-p}, x_{2}^{-p})  & \dots & \sigma(x_{1}^{N+p-1}, x_{2}^{-p}) W^{*}(x_{1}^{N+p-1}, x_{2}^{-p})
\end{pmatrix}.
\end{equation}

From (\ref{eqn:DCapprox}), we see that a matrix of VMMA values $Y = \{Y(x_{1}^{i}, x_{2}^{j}): i, j = 0, ..., N-1\}$ can be obtained through a filtering of $\sigma W^{*}$ by $K$. To compute this efficiently, we can use the convolution theorem and fast Fourier transform (fft) schemes which are available in software such as R. As summarised in Algorithm \ref{alg:DCalg}, this involves zero-padding $K$ and $\sigma W^{*}$ to the size $(N+4p)\times(N+4p)$, taking the Fourier transforms of the resulting matrices and cropping the inverse Fourier transform of their element-wise product. 
\\
For Model (\ref{eqn:weg}), we can compute $\sigma^{2}(\bm{\xi})$ first by replacing $\sigma W^{*}$ in Algorithm \ref{alg:DCalg} by a $(N+2p)\times(N+2p)$ matrix of generated values for the L\'evy noise over each grid square. Figure \ref{fig:GMAvsVMMA}(a) shows a simulated stochastic volatility layer over the region $\mathbf{x} = [-1.5, 11.5] \times [-1.5, 11.5]$ where $\eta = 4$ and $L$ is an IG basis whose seed has mean $a = 1$ and variance $b = 2$. Here, we have chosen $N = 261$, $\triangle = 0.05$ and $p = 30$. After taking square roots of the volatility values and multiplying the results with the Gaussian realisations, we can use Algorithm \ref{alg:DCalg} again to compute $Y$. Assuming that the same kernel truncation parameter $p$ is used, this results in $(261-2p)^{2} = 201^{2}$ values. This scheme takes about two seconds to generate one VMMA data set using a PC with characteristics: Intel$^{\circledR}$ Core\texttrademark i7-3770 CPU Processor @ 3.40GHz; 8GB of RAM; Windows 8.1 64-bit. An example of a simulated realisation is shown in Figure \ref{fig:GMAvsVMMA}(b). 
\\
Such a simulation scheme can also be used for generating data for GMAs. In Figures \ref{fig:GMAvsVMMA}(c), we show the realisation from a GMA with the same underlying Gaussian noise as the VMMA in Plot (b). The VMMA exhibits clustered extremes where the values of its stochastic volatility are high in Plot (a). This in turn has the effect of smoothing the VMMA surface as seen from the contours in Figure \ref{fig:GVContours}(d). The first column of plots in Figure \ref{fig:GVContours} correspond to the same realisations of the stochastic volatility, VMMA and GMA as those in Figure \ref{fig:GMAvsVMMA}.  

\begin{algorithm}[tbp]
\caption{Discrete convolution via Fourier tranforms}\label{alg:DCalg}
\begin{algorithmic}[1]
\State $M1 \gets matrix(0, N+4p, N+4p)$ \Comment{We create a $(N+4p)\times(N+4p)$ matrix of $0$s.}
\State $M1[1:(2p+1), 1:(2p+1)] \gets K$ \Comment{We insert $K$ into the top left corner of $M1$.}
\State $M2 \gets matrix(0, N+4p, N+4p)$ \Comment{We create another $(N+4p)\times(N+4p)$ matrix of $0$s.}
\State $M2[1:(N + 2p), 1:(N+2p)] \gets \sigma W^{*}$ \Comment{We insert $\sigma W^{*}$ into the top left corner of $M2$.}
\State $FM1 \gets fft(M1, inverse = FALSE)$ \Comment{We compute the forward fft of M1.}
\State $FM2 \gets fft(M2, inverse = FALSE)$ \Comment{We compute the forward fft of M2.}
\State $FM \gets FM1*FM2$ \Comment{We multiply FM1 and FM2 element-wise.}
\State $Y \gets Re(fft(FM, inverse = TRUE)/(N+4p)^{2})$ \Comment{We take the real part of the inverse Fourier transform of FM.}
\State $Y \gets Y[(2p+1):(N+2p), (2p+1):(N+2p)]$ \Comment{We crop the matrix $Y$ to obtain the final filtered process.}
\end{algorithmic}
\end{algorithm}

\begin{figure}[tbp]
\centering
\includegraphics[width = 6in, height = 2.2in, trim = 0.7in 0.4in 0.3in 0.1in]{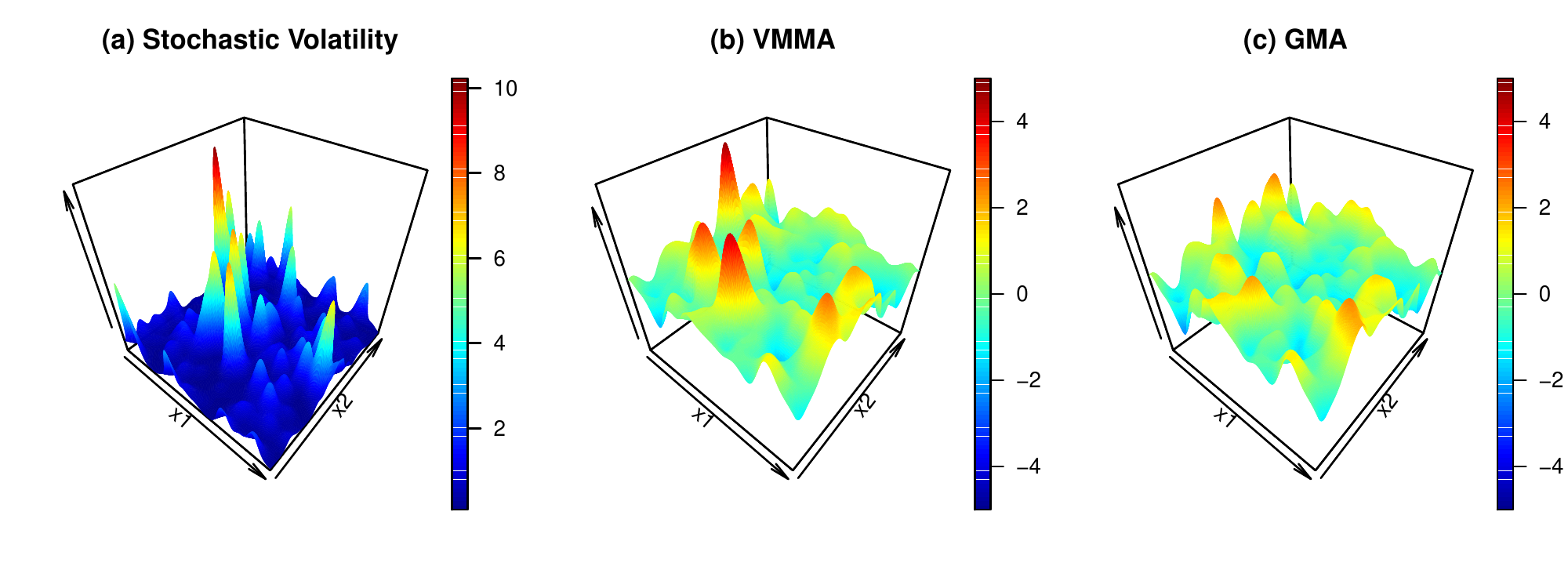}
\caption{Sample paths related to Model (\ref{eqn:weg}): (a) the stochastic volatility with $\eta = 4$; (b) the VMMA with $\lambda = 4$ and IG basis whose seed has mean $a = 1$ and variance $b = 2$; (c) the GMA with the same kernel structure as the VMMA. The same realisation of the Gaussian driving noise is used for the VMMA and GMA to facilitate comparison.}
\label{fig:GMAvsVMMA}
\end{figure}

\begin{figure}[tbp]
\centering
\includegraphics[width = 6in, height = 6.6in]{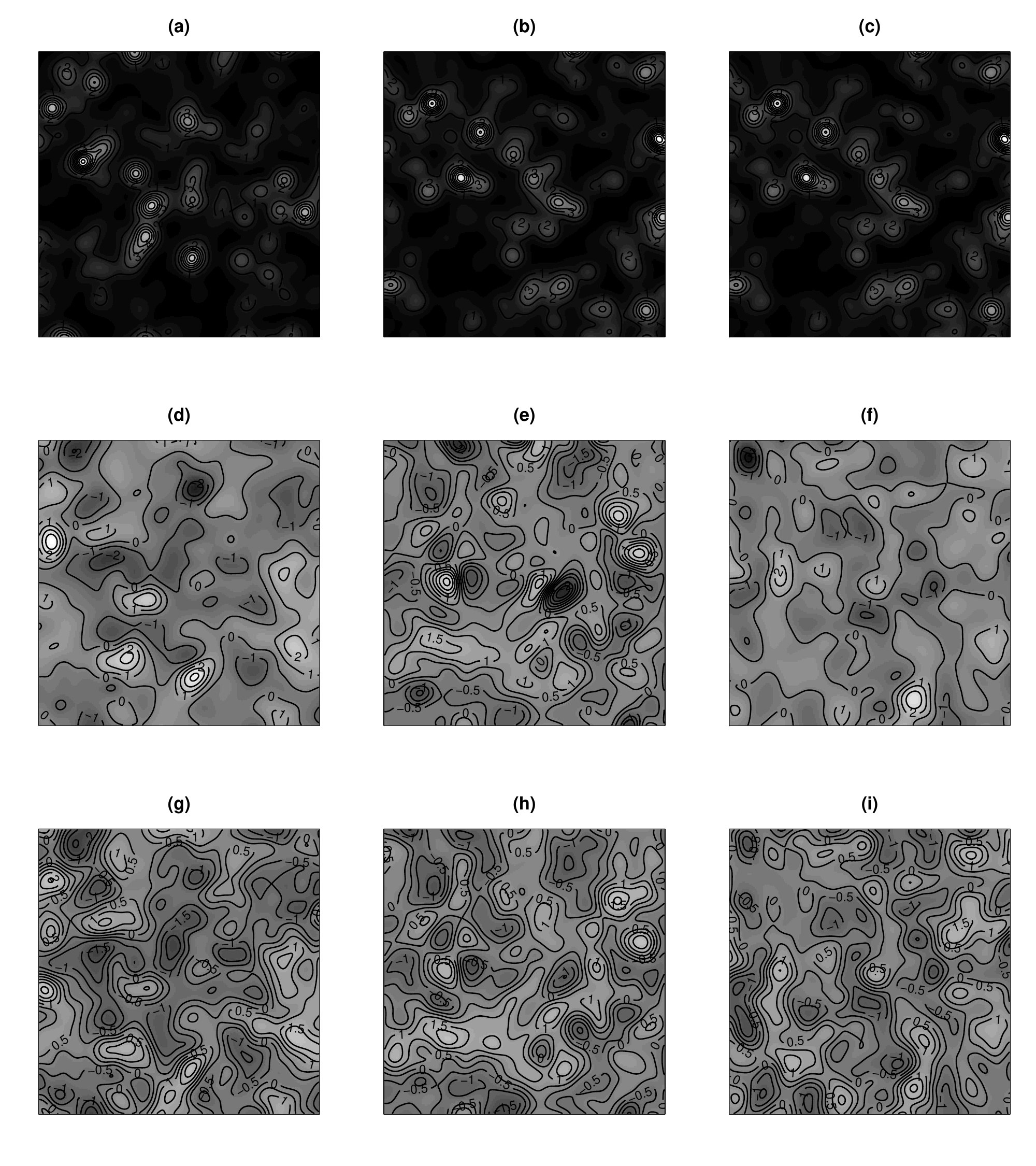}
\caption{Contour plots for three simulations from Model (\ref{eqn:weg}): (a)-(c) the stochastic volatility with $\eta = 4$; (d)-(f) the VMMA with $\lambda = 4$ and IG basis whose seed has mean $a = 1$ and variance $b = 2$; (g)-(i) the GMA with the same kernel structure as the VMMA. For the plots in the same column, the VMMA is constructed using associated stochastic volatility and the same realisation of the Gaussian driving noise is used for the VMMA and GMA to facilitate comparison.}
\label{fig:GVContours}
\end{figure}

\subsection{Mean-square error and its upper bound}

We have a two-step discrete convolution (TSDC) simulation algorithm for the two-tiered VMMA defined in (\ref{eqn:twotier}). Let $Y$ be such a VMMA and $Z =  \{Z(\mathbf{x})\}_{\mathbf{x}\in\mathbb{R}^{2}}$ be its TSDC approximation. Then, we can write $Z(\mathbf{x}) = \int_{\mathbb{R}^{2}} g_{\triangle}(\mathbf{x}, \bm{\xi})\sigma_{\triangle}(\bm{\xi})W(\mathrm{d}\bm{\xi})$, where:
\begin{align*}
g_{\triangle}(\mathbf{x}, \bm{\xi}) &=  \sum_{i = -p}^{p}\sum_{j = -p}^{p} \mathbf{1}_{\left[x_{1} + i\triangle- \frac{\triangle}{2}, x_{1} + i\triangle + \frac{\triangle}{2}\right)}(s_{1}) \mathbf{1}_{\left[x_{2} + j\triangle- \frac{\triangle}{2}, x_{2} + j\triangle + \frac{\triangle}{2}\right)}(s_{2}) g\left(i\triangle, j \triangle\right), \\
\sigma_{\triangle}^{2}(\bm{\xi}) &= \int_{\mathbb{R}^{2}}h_{\triangle}(\bm{\xi}, \mathbf{u})L(\mathrm{d}\mathbf{u}), \\
h_{\triangle}(\bm{\xi},\mathbf{u}) &=  \sum_{i = -\tilde{p}}^{\tilde{p}}\sum_{j = -\tilde{p}}^{\tilde{p}} \mathbf{1}_{\left[\xi_{1} + i\triangle- \frac{\triangle}{2}, \xi_{1} + i\triangle + \frac{\triangle}{2}\right)}(u_{1}) \mathbf{1}_{\left[\xi_{2} + j\triangle- \frac{\triangle}{2}, \xi_{2} + j\triangle + \frac{\triangle}{2}\right)}(u_{2}) h\left(i\triangle, j \triangle\right),
\end{align*}
and $\triangle$ is the grid size while $p$, $\tilde{p}$ are the kernel truncation parameters at the field and volatility layers respectively.
\\
Here, we give an analytical formula for the mean squared error (MSE) involved. Since this is difficult to evaluate in practice, we also give an upper bound which is useful in its own right. The corresponding proofs are given in the Appendix.

\clearpage

\begin{thm} \label{thm:mse}
Let $Y$ be a two-tier VMMA where the mean of the L\'evy seed $L'$ is given by $a > 0$ and let $Z$ be the TSDC approximation of $Y$. Then, $\mathbb{E}\left[|Y(\mathbf{x}) - Z(\mathbf{x})|^{2}\right] = T1 + T2 + T3$ where:
\begin{align*}
T1 &=   a\left(\int_{\mathbb{R}^{2}}g^{2}(\mathbf{x}-\bm{\xi})\mathrm{d}\bm{\xi}\right)\left(\int_{\mathbb{R}^{2}}h(\bm{\xi} - \mathbf{u})\mathrm{d}\mathbf{u} + \sum_{i, j = -\tilde{p}}^{\tilde{p}} h\left(i\triangle, j \triangle\right) \triangle^{2}\right) - 2\int_{\mathbb{R}^{2}}g^{2}(\mathbf{x}-\bm{\xi})\mathbb{E}\left[\sigma(\bm{\xi})\sigma_{\triangle}(\bm{\xi})\right]\mathrm{d}\bm{\xi}, \nonumber\\
T2 &= \left[\left(\int_{\mathbb{R}^{2}}g^{2}(\mathbf{x} - \bm{\xi})\mathrm{d}\bm{\xi} -  \sum_{i,j = -p}^{p} g^{2}\left(i\triangle, j \triangle\right)\triangle^{2}\right) + 2\sum_{i,j = -p}^{p} g\left(i\triangle, j \triangle\right)\left( g\left(i\triangle, j \triangle\right)\triangle^{2} - \int_{i\triangle- \frac{\triangle}{2}}^{i\triangle + \frac{\triangle}{2}} \int_{j\triangle- \frac{\triangle}{2}}^{j\triangle + \frac{\triangle}{2}} g(\mathbf{w})  \mathrm{d}\mathbf{w}\right)\right]\\
&\times \left(a \sum_{i', j' = -\tilde{p}}^{\tilde{p}} h\left(i'\triangle, j' \triangle\right) \triangle^{2}\right),\\
T3 &=   2\left[\int_{\mathbb{R}^{2}}g^{2}(\mathbf{x}-\bm{\xi})\mathbb{E}\left[\sigma(\bm{\xi})\sigma_{\triangle}(\bm{\xi})\right] \mathrm{d}\bm{\xi} - \int_{\mathbb{R}^{2}}g(\mathbf{x} - \bm{\xi})g_{\triangle}(\mathbf{x}, \bm{\xi})\mathbb{E}\left[\sigma(\bm{\xi})\sigma_{\triangle}(\bm{\xi})\right] \mathrm{d}\bm{\xi} \right.  \\ 
&\left. -  a \sum_{i', j'= -\tilde{p}}^{\tilde{p}} h\left(i'\triangle, j'\triangle\right) \triangle^{2}\left(\int_{\mathbb{R}^{2}}g^{2}(\mathbf{x}-\bm{\xi})\mathrm{d}\bm{\xi} - \sum_{i, j = -p}^{p}g\left(i\triangle, j \triangle\right) \int_{i\triangle- \frac{\triangle}{2}}^{i\triangle + \frac{\triangle}{2}} \int_{j\triangle- \frac{\triangle}{2}}^{j\triangle + \frac{\triangle}{2}} g(\mathbf{w})  \mathrm{d}\mathbf{w}\right)\right]. \nonumber
\end{align*}
By letting $\Psi_{L}(\theta)$ denote the Laplace exponent of $L'$ evaluated at $\theta$, we can also express $\mathbb{E}\left[\sigma(\bm{\xi})\sigma_{\triangle}(\bm{\xi})\right]$ as:
\begin{equation*}
\frac{1}{4\pi}\int_{0}^{\infty}\int_{0}^{\infty}\left[1 - e^{\int_{\mathbb{R}^{2}}\Psi_{L}(xh(\bm{\xi} - \mathbf{u}))\mathrm{d}\mathbf{u}} - e^{\int_{\mathbb{R}^{2}}\Psi_{L}(yh_{\triangle}(\bm{\xi}, \mathbf{u}))\mathrm{d}\mathbf{u}} + e^{\int_{\mathbb{R}^{2}}\Psi_{L}(xh(\bm{\xi} - \mathbf{u})+yh_{\triangle}(\bm{\xi}, \mathbf{u}))\mathrm{d}\mathbf{u}} \right] \frac{\mathrm{d} x \mathrm{d}l y}{x^{3/2}y^{3/2}}. 
\end{equation*}
\end{thm}
\begin{Rem}
$T1$ quantifies the part of the MSE that arises from the discrete convolution approximation of $\sigma(\bm{\xi})$ since if $\sigma_{\triangle}\rightarrow\sigma$, we expect it to decrease to zero. On the other hand, $T2$ gives us the part of the MSE that can be attributed to the kernel discretisation and truncation of $g$ since if $g_{\triangle}\rightarrow g$, $T2$ would decrease to zero. The error from the combined effect of simulating $\sigma$ and using $g_{\triangle}$ is represented by $T3$. 
\end{Rem} 
\vspace{4mm}
\begin{Rem}
Theorem \ref{thm:mse} gives us a semi-explicit formula for the MSE involved in our simulations. Although we can approximate this by numerical integrations, in general, it is hard to obtain a full analytic expression for this due to the Laplace exponent of $L'$. For example, if $L'$ has an IG distribution with mean and shape parameter $\mu > 0$ and $\alpha > 0$:
\begin{equation*}
\Psi_{L}(\theta) = \frac{\alpha}{\mu}\left(1 - \sqrt{1 + \frac{2\mu^{2}\theta}{\alpha}}\right) \Rightarrow  e^{\int_{\mathbb{R}^{2}}\Psi_{L}(xh(\bm{\xi} - \mathbf{u}))\mathrm{d}\mathbf{u}} = \exp\left(\frac{\alpha}{\mu} \int_{\mathbb{R}^{2}}\left(1 - \sqrt{1 + \frac{2\mu^{2}xh(\bm{\xi} - \mathbf{u})}{\alpha}}\right)\mathrm{d}\mathbf{u}\right),
\end{equation*}
which is hard to simplify in general.
\end{Rem}
\vspace{4mm} 
\begin{Rem}
If instead of modelling $\sigma^{2}$, we modelled $\sigma$ or $\log(\sigma^{2})$, $\mathbb{E}\left[\sigma(\bm{\xi})\sigma_{\triangle}(\bm{\xi})\right]$ would be easier to evaluate. The disadvantage of the latter strategies is that higher order moments or moments of transformations need to be used for the inference in Section \ref{sec:est}.
\end{Rem}

Although it is hard to calculate the MSE in practice due to the presence of the term $\mathbb{E}\left[\sigma(\bm{\xi})\sigma_{\triangle}(\bm{\xi})\right]$, we can obtain a useful upper bound by using the relationship between the harmonic, arithmetic and geometric means of $\sigma^{2}$ and $\sigma_{\triangle}^{2}$. 
\\
\begin{Cor} \label{cor:mseub}
Let $Y$ be a two-tier VMMA where the mean of the L\'evy seed $L'$ is given by $a > 0$ and let $Z$ be the TSDC approximation of $Y$. Then, $\mathbb{E}\left[|Y(\mathbf{x}) - Z(\mathbf{x})|^{2}\right] \leq T2 + T4 + T5$ where $T2$ is as defined in Theorem \ref{thm:mse}, and:
\begin{align*}
T4 &= \left( a\int_{\mathbb{R}^{2}}g^{2}(\mathbf{x}-\bm{\xi})\mathrm{d}\bm{\xi}\right) \frac{\left(\int_{\mathbb{R}^{2}}h(\bm{\xi} - \mathbf{u})\mathrm{d}\mathbf{u} -  \sum_{i = -\tilde{p}}^{\tilde{p}} \sum_{j = -\tilde{p}}^{\tilde{p}} h\left(i\triangle, j \triangle\right) \triangle^{2}\right)^{2}}{\int_{\mathbb{R}^{2}}h(\bm{\xi} - \mathbf{u})\mathrm{d}\mathbf{u} + \sum_{i = -\tilde{p}}^{\tilde{p}} \sum_{j = -\tilde{p}}^{\tilde{p}} h\left(i\triangle, j \triangle\right) \triangle^{2}}, \\
\text{while } T5 &=a\left[\int_{\mathbb{R}^{2}}h(\bm{\xi} - \mathbf{u})\mathrm{d}\mathbf{u} -  \sum_{i, j = -\tilde{p}}^{\tilde{p}} h\left(i\triangle, j \triangle\right) \triangle^{2}\right]\left[\left(\int_{\mathbb{R}^{2}}g^{2}(\mathbf{x}-\bm{\xi})\mathrm{d}\bm{\xi} - \sum_{i, j = -p}^{p}g^{2}\left(i\triangle, j \triangle\right) \triangle^{2} \right)\right. \\
&\left. + \sum_{i, j = -p}^{p}g\left(i\triangle, j \triangle\right) \left( g\left(i\triangle, j \triangle\right)\triangle^{2} - \int_{i\triangle- \frac{\triangle}{2}}^{i\triangle + \frac{\triangle}{2}} \int_{j\triangle- \frac{\triangle}{2}}^{j\triangle + \frac{\triangle}{2}} g(\mathbf{w})  \mathrm{d}\mathbf{w}\right) \right. \\
&\left.+ \left(\int_{\mathbb{R}^{2}}h(\bm{\xi} - \mathbf{u})\mathrm{d}\mathbf{u} - \sum_{i', j'= -\tilde{p}}^{\tilde{p}} h\left(i'\triangle, j'\triangle\right) \triangle^{2} \right) \left(\sum_{i, j = -p}^{p}g\left(i\triangle, j \triangle\right) \int_{i\triangle- \frac{\triangle}{2}}^{i\triangle + \frac{\triangle}{2}} \int_{j\triangle- \frac{\triangle}{2}}^{j\triangle + \frac{\triangle}{2}} g(\mathbf{w})  \mathrm{d}\mathbf{w}\right)\right]. 
\end{align*}
\end{Cor}
By using the following assumption, we can derive a result to help us find corresponding orders of convergence:
\begin{Assum} \label{a:bHm}
$g^{2}(\mathbf{w})$ and  $h^{2}(\mathbf{w})$ have bounded Hessian matrices.
\end{Assum}
\vspace{2mm}
\begin{Lem} \label{lem:kcomp}
Let $R = p\triangle$ and $\widetilde{R} = \tilde{p}\triangle$ be the fixed truncation ranges for $g$ and $h$ respectively. Under Assumption \ref{a:bHm}:
\begin{gather*}
\int_{\mathbb{R}^{2}}g^{2}(\mathbf{w})\mathrm{d}\mathbf{w} -  \sum_{i,j = -p}^{p} g^{2}\left(i\triangle, j \triangle\right)\triangle^{2} =  O(\triangle^{2}), ~~
g\left(i\triangle, j \triangle\right)\triangle^{2} - \int_{i\triangle- \frac{\triangle}{2}}^{i\triangle + \frac{\triangle}{2}} \int_{j\triangle- \frac{\triangle}{2}}^{j\triangle + \frac{\triangle}{2}} g(\mathbf{w})  \mathrm{d}\mathbf{w} = O(\triangle^4),\\ 
\text{and } \int_{\mathbb{R}^{2}}h(\mathbf{w})\mathrm{d}\mathbf{w} -  \sum_{i , j = -\tilde{p}}^{\tilde{p}} h\left(i\triangle, j \triangle\right) \triangle^{2} = O(\triangle^{2}).
\end{gather*}
\end{Lem}
The analysis so far has been for fixed $R$ and $\widetilde{R}$. Further suppose that:
\\
\begin{Assum} \label{a:Rdelta}
 $R = O\left(\triangle^{-r}\right)$ and $\widetilde{R} = O\left(\triangle^{-r}\right)$ where $0<r<3$. 
\end{Assum}
Now, $R, \widetilde{R} \rightarrow \infty$ as $\triangle\rightarrow 0$ so that the MSE converges to zero. As shown in the next two examples, the order of this convergence will depend on the forms of $g$ and $h$. 
\\
\begin{Eg} \label{eg:powereg} Since:
\begin{align*}
\int_{\mathbb{R}^{2}}g^{2}(\mathbf{w})\mathrm{d}\mathbf{w} -  \sum_{i,j = -p}^{p} g^{2}\left(i\triangle, j \triangle\right)\triangle^{2} &= \left(\int_{\mathbb{R}^{2}}g^{2}(\mathbf{w})\mathrm{d}\mathbf{w} -  \int_{-\left(R + \triangle/2\right)}^{R + \triangle/2}\int_{-\left(R + \triangle/2\right)}^{R + \triangle/2}g^{2}(\mathbf{w})\mathrm{d}\mathbf{w} \right) \\
&+ \left(\int_{-\left(R + \triangle/2\right)}^{R + \triangle/2}\int_{-\left(R + \triangle/2\right)}^{R + \triangle/2}g^{2}(\mathbf{w})\mathrm{d}\mathbf{w} -  \sum_{i,j = -\lfloor R/\triangle \rfloor}^{\lfloor R/\triangle \rfloor} g^{2}\left(i\triangle, j \triangle\right)\triangle^{2}\right),
\end{align*}
we find the order of convergence of $\int_{\mathbb{R}^{2}}g^{2}(\mathbf{w})\mathrm{d}\mathbf{w} -  \int_{-\left(R + \triangle/2\right)}^{R + \triangle/2}\int_{-\left(R + \triangle/2\right)}^{R + \triangle/2}g^{2}(\mathbf{w})\mathrm{d}\mathbf{w}$ for $R = O\left(\triangle^{-r}\right)$ and $\triangle \rightarrow 0$.
\\
Suppose that $g(w_{1}, w_{2}) = g(-w_{1}, w_{2}) = g(w_{1}, - w_{2})$, i.e. $g$ is symmetric about the axes, and $g$ is bounded over $\mathbb{R}^{2}$ . In addition, for a large and fixed value of $|w_{1}|$, $g(\mathbf{w}) \sim A_{1}(|w_{1}|)|w_{2}|^{-\beta}$ for some square integrable function $A_{1}:\mathbb{R}^{+}\rightarrow\mathbb{R}^{+}$ and $\beta > 1/2$. Similarly, for a large and fixed value of $|w_{2}|$, $g(\mathbf{w}) \sim A_{2}(|w_{2}|)|w_{1}|^{-\alpha}$ for some square integrable function $A_{2}:\mathbb{R}^{+}\rightarrow\mathbb{R}^{+}$ and $\alpha > 1/2$. Notice that this implies that $g$ behaves proportional to $|w_{1}|^{-\alpha}|w_{2}|^{-\beta}$ for large $|w_{1}|$ and $|w_{2}|$. Applying these conditions, we have:
\begin{align*}
&\int_{\mathbb{R}^{2}}g^{2}(\mathbf{w})\mathrm{d}\mathbf{w} -  \int_{-\left(R + \triangle/2\right)}^{R + \triangle/2}\int_{-\left(R + \triangle/2\right)}^{R + \triangle/2}g^{2}(\mathbf{w})\mathrm{d}\mathbf{w} \\
&= 4\int_{0}^{\infty}\int_{R + \triangle/2}^{\infty} g^{2}(\mathbf{w})\mathrm{d}\mathbf{w} + 4 \int_{R + \triangle/2}^{\infty}\int_{0}^{R + \triangle/2} g^{2}(\mathbf{w})\mathrm{d}\mathbf{w} \text{ by the symmetry about the axes,} \\
&=  4\left[\int_{R + \triangle/2}^{\infty}\int_{R + \triangle/2}^{\infty} g^{2}(\mathbf{w})\mathrm{d}\mathbf{w} + \int_{0}^{R + \triangle/2}\int_{R + \triangle/2}^{\infty} g^{2}(\mathbf{w})\mathrm{d}\mathbf{w} + \int_{R + \triangle/2}^{\infty}\int_{0}^{R + \triangle/2} g^{2}(\mathbf{w})\mathrm{d}\mathbf{w} \right] \\
&\sim 4\left[\int_{R + \triangle/2}^{\infty}\int_{R + \triangle/2}^{\infty} C_{1}|w_{1}|^{-2\alpha}|w_{2}|^{-2\beta} \mathrm{d}\mathbf{w} + \int_{0}^{R + \triangle/2}\int_{R + \triangle/2}^{\infty} A^{2}_{2}(|w_{2}|)|w_{1}|^{-2\alpha}\mathrm{d}\mathbf{w} + \int_{R + \triangle/2}^{\infty}\int_{0}^{R + \triangle/2} A^{2}_{1}(|w_{1}|)|w_{2}|^{-2\beta}\mathrm{d}\mathbf{w} \right] \\
&< 4\left[ \frac{C_{1}(R + \triangle/2)^{2 - 2\alpha - 2\beta}}{(2\alpha - 1)(2\beta - 1)} + \frac{C_{2}(R + \triangle/2)^{1 - 2\alpha}}{2\alpha - 1} + \frac{C_{3}(R + \triangle/2)^{1 - 2\beta}}{2\beta - 1} \right] \\
&= O(\triangle^{r(2\min(\alpha, \beta) - 1)}),
\end{align*}
where $C_{1}$ is a finite constant, $C_{2} =\int_{0}^{\infty} A^{2}_{2}(|w_{2}|)\mathrm{d}w_{2}$ and $C_{3} = \int_{0}^{\infty} A^{2}_{1}(|w_{1}|)\mathrm{d}w_{1}$. Adding this to the previous bound that we obtained by assuming a bounded Hessian for $g$, we have:
\begin{equation*}
\int_{\mathbb{R}^{2}}g^{2}(\mathbf{w})\mathrm{d}\mathbf{w} -  \sum_{i,j = -p}^{p} g^{2}\left(i\triangle, j \triangle\right)\triangle^{2} < O(\triangle^{\min(2, r(2\min(\alpha, \beta) - 1))}). 
\end{equation*}
\\
If $h$ is symmetric about the axes and shares the same asymptotic properties as $g$ but with parameters $\tilde{\alpha}$ and $\tilde{\beta}$ in place of $\alpha$ and $\beta$ respectively, we can use an analogous approach to obtain $\int_{\mathbb{R}^{2}}h(\mathbf{w})\mathrm{d}\mathbf{w} -  \sum_{i , j = -\tilde{p}}^{\tilde{p}} h\left(i\triangle, j \triangle\right) \triangle^{2} < O(\triangle^{\min(2, r(2\min(\tilde{\alpha}, \tilde{\beta}) - 1))})$. \\ 
Applying these bounds to $T2$, $T4$ and $T5$, we have:
\begin{align*}
T2 &= O(\triangle^{\min(2, r(2\min(\alpha, \beta) - 1), 3 - r)}), \text{ } T4 = O(\triangle^{2\min(2, r(2\min(\tilde{\alpha}, \tilde{\beta}) - 1))}), \\
\text{and } T5 &= O(\triangle^{\min(4, 2r(2\min(\tilde{\alpha}, \tilde{\beta}) - 1), 2r(\min(\alpha, \beta) + \min(\tilde{\alpha}, \tilde{\beta}) - 1), 2 +  r(2\min(\alpha, \beta) - 1), 5 - r)}).\\
\Rightarrow \mathbb{E}\left[|Y(\mathbf{x}) - Z(\mathbf{x})|^{2}\right] &\leq O(\triangle^{\min(2, r(2\min(\alpha, \beta) - 1), 2r(2\min(\tilde{\alpha}, \tilde{\beta}) - 1), 3 - r)}). 
\end{align*}
Note that since $p = \lfloor R/\triangle \rfloor = O(\triangle^{-r - 1})$, in $T2$ and $T5$, we have:
\begin{align*}
&\sum_{i, j = -p}^{p}g\left(i\triangle, j \triangle\right) \left( g\left(i\triangle, j \triangle\right)\triangle^{2} - \int_{i\triangle- \frac{\triangle}{2}}^{i\triangle + \frac{\triangle}{2}} \int_{j\triangle- \frac{\triangle}{2}}^{j\triangle + \frac{\triangle}{2}} g(\mathbf{w})  \mathrm{d}\mathbf{w}\right) \\
&< \sup_{\mathbf{w}\in\mathbb{R}^{2}}{g(\mathbf{w})}\sum_{i, j = -p}^{p} \left( g\left(i\triangle, j \triangle\right)\triangle^{2} - \int_{i\triangle- \frac{\triangle}{2}}^{i\triangle + \frac{\triangle}{2}} \int_{j\triangle- \frac{\triangle}{2}}^{j\triangle + \frac{\triangle}{2}} g(\mathbf{w})  \mathrm{d}\mathbf{w}\right) = O(\triangle^{3 - r}). 
\end{align*} 
\end{Eg}
\vspace{2mm}
\begin{Eg} \label{eg:isoeg}
If we assume that $g$ and $h$ are isotropic, and $g(\mathbf{w}) \sim |\mathbf{w}|^{-\chi}$ while $h(\mathbf{w}) \sim |\mathbf{w}|^{-\tilde{\chi}}$ for some $\chi, \tilde{\chi} > 1/2$ when $|\mathbf{w}|$ is large, a simpler convergence bound for the MSE can be obtained. With $C$ being a constant:
\begin{align*}
&\int_{\mathbb{R}^{2}}g^{2}(\mathbf{w})\mathrm{d}\mathbf{w} -  \int_{-\left(R + \triangle/2\right)}^{R + \triangle/2}\int_{-\left(R + \triangle/2\right)}^{R + \triangle/2}g^{2}(\mathbf{w})\mathrm{d}\mathbf{w} \leq  \int_{|\mathbf{w}|> R + \triangle/2}g^{2}(|\mathbf{w}|)\mathrm{d}|\mathbf{w}| = \frac{C(R + \triangle/2)^{1 - 2\chi}}{2\chi - 1} = O(\triangle^{r(2\chi - 1)}).  \\
&\Rightarrow \int_{\mathbb{R}^{2}}g^{2}(\mathbf{w})\mathrm{d}\mathbf{w} -  \sum_{i,j = -p}^{p} g^{2}\left(i\triangle, j \triangle\right)\triangle^{2} \leq O(\triangle^{\min(2, r(2\chi - 1))}). 
\end{align*}
Similarly, one can show that $\int_{\mathbb{R}^{2}}h(\mathbf{w})\mathrm{d}\mathbf{w} -  \sum_{i , j = -\tilde{p}}^{\tilde{p}} h\left(i\triangle, j \triangle\right) \triangle^{2} = O(\triangle^{\min(2, r(2\tilde{\chi} - 1))})$. 
\\
Correspondingly, we have:
\begin{gather*}
T2 = O(\triangle^{\min(2, r(2\chi - 1), 3 - r)}), \text{ } T4 = O(\triangle^{2\min(2, r(2\tilde{\chi} - 1))}) \text{ and } T5 = O(\triangle^{\min(4, 2r(2\tilde{\chi} - 1), 2r(\chi + \tilde{\chi} - 1), 2 +  r(2\chi - 1), 5 - r)}) \\
\Rightarrow \mathbb{E}\left[|Y(\mathbf{x}) - Z(\mathbf{x})|^{2}\right] \leq O(\triangle^{\min(2, r(2\chi - 1), 2r(2\tilde{\chi} - 1), 3 - r)}). 
\end{gather*}
\end{Eg}
\vspace{2mm}
\begin{Eg} \label{eg:simeg}
For illustration, we provide the explicit MSE upper bound for Model (\ref{eqn:weg}). With $\Phi(x; y)$ denoting the standard Normal distribution function evaluated between x and y, the contributing terms are:
\begin{align*}
T4 &=  \frac{a\lambda}{2\pi} \left[ 1 + \frac{\eta}{\pi}\left[1 + 2\sum_{i'=1}^{\tilde{p}} \exp\left(-\eta\triangle^{2}i'^{2}\right)\right]^{2}\triangle^{2} - 4\left(1 + \frac{\pi}{\eta\left[1 + 2\sum_{i' = 1}^{\tilde{p}} \exp\left(-\eta\triangle^{2}i'^{2}\right)\right]^{2}\triangle^{2}}\right)^{-1} \right] \\
T2 &= \frac{a\lambda\eta}{\pi^{2}}\triangle^{2}\left[1 + 2\sum_{i' = 1}^{\tilde{p}}\exp\left(-\eta\triangle^{2}i'^{2}\right) \right]^{2} \\
&\times \left[\frac{1}{2} + \frac{\lambda}{\pi}\triangle^{2}\left[1 + 2\sum_{i = 1}^{p} \exp\left(- 2\lambda\triangle^{2}i^{2}\right)\right]^{2}  - 2\left[\sum_{i = -p}^{p}e^{-\lambda\triangle^{2}i^{2}} \Phi\left(\sqrt{2\lambda}\left(i\triangle - \frac{\triangle}{2}\right) , \sqrt{2\lambda}\left(i\triangle + \frac{\triangle}{2}\right)\right)\right]^{2}\right], \\
\text{and } T5 &=  \frac{a\lambda}{\pi}\left[ \frac{1}{2} - \frac{\eta}{2\pi}\left[1 + 2\sum_{i'= 1}^{\tilde{p}} \exp\left(-\eta\triangle^{2}i'^{2}\right)\right]^{2}\triangle^{2} - 4\left[\sum_{i = -p}^{p}e^{-\lambda\triangle^{2}i^{2}} \Phi\left(\sqrt{2\lambda}\left(i\triangle - \frac{\triangle}{2}\right) , \sqrt{2\lambda}\left(i\triangle + \frac{\triangle}{2}\right)\right)\right]^{2} \right.\\
&\left.\times \left( \left(1 + \frac{\pi}{\eta\left[1 + 2\sum_{i' = 1}^{\tilde{p}} \exp\left(-\eta\triangle^{2}i'^{2}\right)\right]^{2}\triangle^{2}}\right)^{-1} -  \frac{\eta}{2\pi}\triangle^{2} \left[1 + 2\sum_{i'= 1}^{\tilde{p}} \exp\left(-\eta\triangle^{2}i'^{2}\right)\right]^{2}\right) \right]. 
\end{align*}
Using our simulation settings, i.e. $\lambda = \eta = 4$, $a = 1$, $b = 2$ and $p = \tilde{p}$, we examine the behaviour of $T2$, $T4$ and $T5$, as well as the resulting upper bound on the MSE for the case: $R = K\triangle^{-1}$ where $K = 0.05^{2}\times 30$ so that as $\triangle$ decreases, $R$ increases. We know from Examples \ref{eg:powereg} and \ref{eg:isoeg} that the MSE converges to zero at most as fast as $O(\triangle^{2})$ when $R = O(\triangle^{-1})$ since the squared exponential dominates any power function. 
\\
Figure \ref{fig:NCase3} shows the upper bound, $T4$, $T2$ and $T5$ values. We notice that $T2$ is much larger in magnitude than $T4$ and $T5$. This indicates that for our choice of $a = 1$ and $\lambda = \eta = 4$, the error due to the kernel truncation and discretisation of $g$ outweighs those due to the simulation of $\sigma$. As such, the asymptotic behaviour of our upper bound is driven largely by the behaviour of $T2$. The vertical dotted lines in each plot indicate our simulation choice of $\triangle = 0.05$ and $R = 1.5$ for the experiments in Section \ref{sec:est}. From the plots, we also see that when $R$ increases and $\triangle$ decreases, $T4$, $T2$ and $T5$ converge smoothly to $0$ so that the MSE upper bound (and the MSE itself) converges to zero as expected.
\end{Eg}

\begin{figure}[tbp]
\centering
\includegraphics[width = 5in, height = 4.5in]{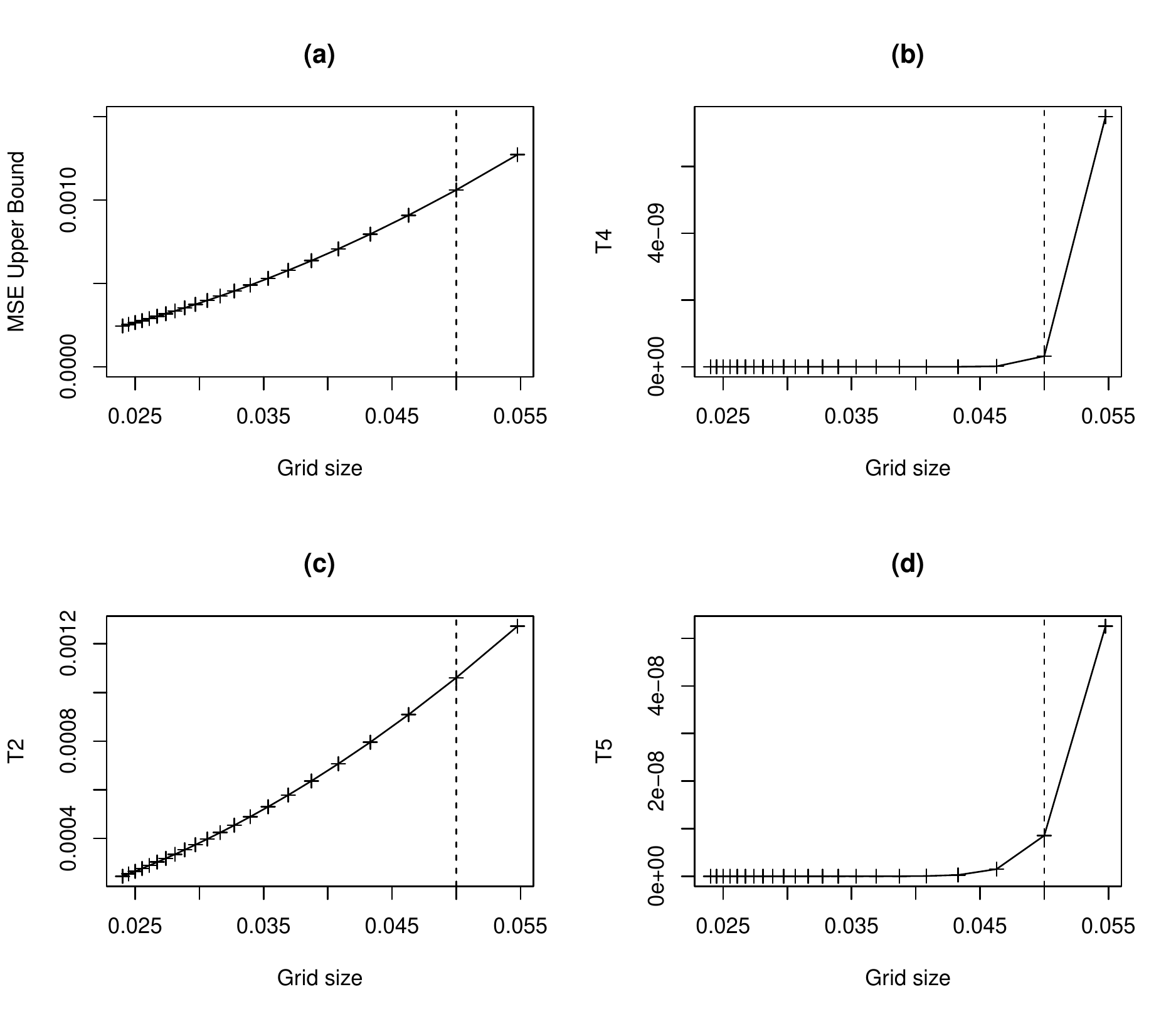}
\caption{(a) MSE upper bound against $\triangle$ for Model (\ref{eqn:weg}) when we set $\lambda = \eta = 4$, $a = 1$, $b = 2$, $p = \tilde{p}$ and $R = K\triangle^{-1}$ where $K = 0.05^{2}\times 30$; and the corresponding values of (b) $T4$, (c) $T2$  and (d) $T5$. The vertical dotted lines mark our chosen simulation setting.}
\label{fig:NCase3}
\end{figure}
\vspace{2mm}
\begin{Rem}
So far, we have assumed that $g$ takes finite values over $\mathbb{R}^{2}$. This holds for many choices of $g$. In the case of a singularity at zero which occurs for the Mat\'ern kernel (\ref{eqn:Mker}) with $(\alpha-d)/2 \in (-1/2, 0)$, a so-called hybrid scheme similar to that in \cite{BLP2015} can be used.
\end{Rem}

\section{Inference} \label{sec:est}

Before we introduce the two-step moments-matching method for VMMAs, we will look at several classical inference approaches and give reasons why it is difficult to implement them for these models.

\subsection{A note on classical methods}

In Econometrics where we have high-frequency financial data, a tool frequently used to estimate the parameters of the volatility is that of realised variance. This involves the sum of the squared increments of the log-price over small time intervals. Recently, this concept has been extended to two-dimensions in \cite{Pakk2014}. Now instead of taking increments over time intervals, we take increments over rectangles in our data region. Under certain assumptions, one can show that this two-dimensional realised variance properly normalised, converges to a weighted integral of the volatility as these rectangles get smaller. This could enable us to estimate the parameters of the volatility by moments-based or quasi-likelihood methods. Unfortunately, it was found that not all the required conditions hold for general VMMAs and it is hard to establish the convergence. 
\\
Another classical approach would be to use likelihood-based or Bayesian inference. In general, however, since we do not know the specific distribution of $\sigma^2$ and an approximation of the conditional variance does not guarantee valid covariance matrices, implementing such strategies for a VMMA is not straightforward. 
\\
Alternatively, one might consider a direct moments-matching method involving higher order moments. As we have found though, this may lead to parameter sign errors. That is, we could estimate a positive parameter as negative if no adjustments are made.
\\
Since these classical approaches are hard to implement, we will develop a two-step moments-matching estimation method for our VMMAs. 

\subsection{Moments-based estimation}

To illustrate our moments-based method, we use data generated from Model (\ref{eqn:weg}). Recall that our VMMA $Y$ is stationary, but when we condition on $\sigma^{2}$, $Y$ is non-stationary. Specifically, when we express its observations as a vector, $Y|\sigma^{2} \sim N(0, V)$ where $V = (V_{ij})$ with $V_{ij} = \int_{\mathbb{R}^{2}}\lambda^{2}\pi^{-2}\exp\left(-\lambda\left(\mathbf{x}_{i}-\bm{\xi}\right)^{T}\left(\mathbf{x}_{i}-\bm{\xi}\right) - \lambda \left(\mathbf{x}_{j}-\bm{\xi}\right)^{T}\left(\mathbf{x}_{j}-\bm{\xi}\right)\right)\sigma^{2}(\bm{\xi})\mathrm{d}\bm{\xi}$ and where $\mathbf{x}_{i}, \mathbf{x}_{j} \in \mathbb{R}^{2}$ are two data locations. Since $\sigma^{2}$ is stationary, we should expect to observe stationarity in $Y$ over a large region. This should allow us to estimate $\lambda$ and $a$ through the empirical normalised variogram and variance.
\\
By definition, the normalised variogram of our VMMA, $Y(\mathbf{x})$, is given by:
\begin{align}
\gamma(d_{\mathbf{x}}) &:= \frac{\mathbb{E}\left[\left(Y\left(\mathbf{x}\right)-Y\left(\mathbf{x}^{*}\right)\right)^{2}\right]}{\Var\left(Y\left(\mathbf{x}\right)\right) }= 2(1 - \Corr\left(Y\left(\mathbf{x}\right), Y\left(\mathbf{x}^{*}\right)\right)) = 2\left(1-\exp\left(-\frac{\lambda d_{\mathbf{x}}^{2}}{2} \right)\right), \label{eqn:nvar}
\end{align}
where $\mathbf{x}, \mathbf{x}^{*} \in \mathbb{R}^{2}$ and $|\mathbf{x} - \mathbf{x}^{*}| = d_{\mathbf{x}}$. 
\\
Let $\left(Y(x^{i}_{1}, x^{j}_{2})\right) \in \mathbb{R}^{M\times M}$ be our data matrix and  $Y\left(\mathbf{x}_{k}\right) \in \mathbb{R}^{M^2}$ be its corresponding data vector. With $N(d_{\mathbf{x}})$ denoting the set containing all the pairs of indices of sites with spatial distance $d_{\mathbf{x}}$, we can estimate the normalised variogram by:
\begin{align}
\hat{\gamma}(d_{\mathbf{x}}) = \frac{1}{|N(d_{\mathbf{x}})|} \sum_{(k, l) \in N(d_{\mathbf{x}})}\frac{(Y(\mathbf{x}_{k}) - Y(\mathbf{x}_{l}))^{2}}{\hat{\kappa}_{2}}, \label{eqn:envar}
\end{align}
where $\hat{\kappa}_{2} =  \frac{1}{M^{2}-1} \sum_{k = 1}^{M^{2}} (Y(\mathbf{x}_{k}) - \overline{Y})^{2}$ and $\overline{Y} = \frac{1}{M^{2}} \sum_{k = 1}^{M^{2}}Y(\mathbf{x}_{k})$. By matching (\ref{eqn:nvar}) and (\ref{eqn:envar}), we can estimate the rate parameter of our field by:
\begin{equation*}
\hat{\lambda} = -\frac{2\log(1 - \hat{\gamma}(\triangle)/2)}{\triangle^{2}},
\end{equation*}
where $\triangle$ is the simulation grid size. For Model (\ref{eqn:weg}), $\kappa_{2} = \frac{a\lambda}{2\pi}$. Thus, with $\hat{\lambda}$ at hand, we can estimate $a$ by:
\begin{equation*}
\hat{a}= \frac{2\pi\hat{\kappa}_{2}}{\hat{\lambda}}.
\end{equation*} 
Next, we want to obtain estimates for the parameters $b$ and $\eta$. These determine the extent of non-stationarity in $Y|\sigma^{2}$ which is shown through differing variance and covariance structures across subregions. Thus, it seems natural to infer about $b$ and $\eta$ by comparing estimated local variances. This requires some sort of subsetting. To retain the correlation between our local variance estimates for the next step of our inference, we use a moving window strategy. 
\\
Figure \ref{fig:EDConcept2} illustrates how the $q\times q$ moving window, which is represented by the small box on the bottom left corner of Plot (a), selects data points for local variance calculations to form a field of estimates in Plot (b). The current location at which the local variance is being calculated is represented by the circle and labelled $\bm{\xi}_{1}$. By moving the window from left to right and then up the rows of the $M\times M$ data matrix, we obtain the $(M-q+1)\times (M-q+1)$ field of local variance estimates, as illustrated in Plot (b). The parameter $q$ is a tuning parameter in our inference method. As will be evident later, different $q$ values lead to different inferred volatility cluster sizes which in turn are related to the values of $b$ and $\eta$.
\\
For $\mathbf{x} \in \mathbb{R}^{2}$, we can index the local variance estimator by $Q = (q-1)/2$ as follows:
\begin{equation*}
\hat{\sigma}_{I}^{2}(\mathbf{x}, Q)= \frac{1}{(2Q+1)^{2}} \sum_{l = -Q}^{Q} \sum_{k = -Q}^{Q} Y^{2}(\mathbf{x} + (l, k)\triangle).
\end{equation*}
The local variances estimates are estimates for the conditional variance at the centres of the subregions marked out by the moving window procedure. The analytical formula for the latter is given by:
\begin{equation}
\sigma^{2}_{I}(\bm{\xi}_{i}) = \int_{\mathbb{R}^{2}}\frac{\lambda^{2}}{\pi^{2}}\exp\left(-2\lambda\left(\bm{\xi}_{i}-\bm{\xi}\right)^{T}\left(\bm{\xi}_{i}-\bm{\xi}\right)\right)\sigma^{2}(\bm{\xi})\mathrm{d}\bm{\xi}, \label{eqn:s2I} 
\end{equation}
where $\bm{\xi}_{i}$ denotes the centre of the subregion $A_{i}$ for $i = 1, \dots, \widetilde{M}$ and $\widetilde{M} = (M-q+1)^2$.  
\\
By comparing the mean of our local variances to $a\lambda(2\pi)^{-1}$ and writing $\overline{\hat{\sigma}^{2}_{I}} = \frac{1}{\widetilde{M}}\sum_{i = 1}^{\widetilde{M}} \hat{\sigma}^{2}_{I}(\bm{\xi}_{i}, Q)$, we get another estimator for $a$:
\begin{equation*}
\hat{a}_{2} = \frac{2\pi\overline{\hat{\sigma}^{2}_{I}}}{\hat{\lambda}}. 
\end{equation*}
From the proof of Example \ref{eg:2oCov}, $\Cov(\sigma^{2}_{I}(\bm{\xi}_{i}), \sigma^{2}_{I}(\bm{\xi}_{j})) = A \exp(-B(\bm{\xi}_{i} - \bm{\xi}_{j})^{T}(\bm{\xi}_{i} - \bm{\xi}_{j}))$, 
where $A = b\lambda^{3}\eta(4\pi^{3}(2\lambda + \eta))^{-1}$ and $B = \lambda\eta(2\lambda + \eta)^{-1}$. Using the empirical variance and normalised variogram at the first lag of $\hat{\sigma}^{2}_{I}$, $\hat{\psi}(\triangle)$, we obtain:
\begin{equation*}
\widehat{A} = \frac{1}{\widetilde{M}-1}\sum_{i = 1}^{\widetilde{M}} ( \hat{\sigma}^{2}_{I}(\bm{\xi}_{i}) - \overline{\hat{\sigma}^{2}_{I}} )^{2} \text{ and } \widehat{B} = -\frac{\log(1 - \hat{\psi}(\triangle)/2)}{\triangle^{2}}.
\end{equation*}
This in turn gives us:
\begin{equation*}
\hat{b} = \frac{4\pi^{3}\hat{A}}{\hat{\lambda}^{2}\widehat{B}} \text{ and } \hat{\eta} = \frac{2\hat{\lambda}\widehat{B}}{\hat{\lambda} - \widehat{B}}.
\end{equation*}
\vspace{2mm}
\begin{Rem}\label{rem:general}
We have used Model (\ref{eqn:weg}) to illustrate our method. More generally, this strategy works when we have parameters representing the variance and correlation of $Y$ and $\sigma^{2}$ respectively, and when an analytical expression for $\Cov(\sigma^{2}_{I}(\bm{\xi}_{i}), \sigma^{2}_{I}(\bm{\xi}_{j}))$ is available. 
\end{Rem}

\begin{figure}[tbp]
\centering
\includegraphics[width =6.4in, height = 3in, trim = 0.4in 1in 0.4in 0.4in]{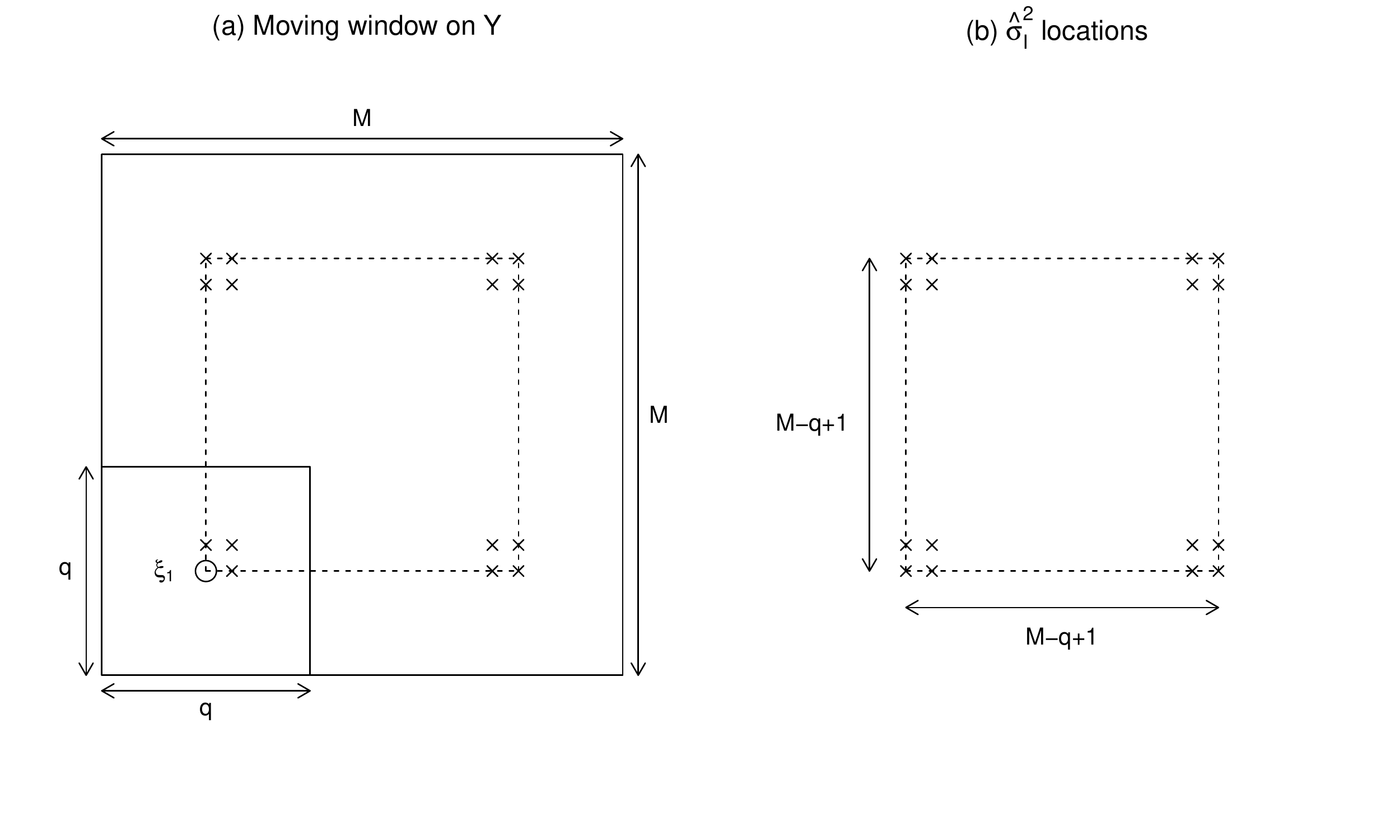}
\caption{Moving window strategy: (a) The $q\times q$ moving window, represented by the smaller box on the bottom left corner, slides across the $M\times M$ data region to select the data points for local variance estimation. (b) A $(M-q+1)\times(M-q+1)$ field of local variance estimates ($\hat{\sigma}^{2}_{I}$) is created.} 
\label{fig:EDConcept2}
\end{figure}

\subsection{Theoretical properties of the estimators}

In this subsection, we derive some properties of the local variance and moments-based parameter estimators. We show that under suitable conditions, the latter are consistent. Proofs of the results, if not shown, can be found in the Appendix. 
\\
\begin{thm} \label{thm:l2con}
Suppose that the following conditions hold:
\begin{enumerate}
\item  \begin{equation*} \frac{\sum_{\substack{l', k', l, k\\ = -Q}}^{Q}  \mathbb{E}\left[\left(\int_{\mathbb{R}^{2}} g(\mathbf{x} + (l, k)\triangle-\bm{\xi})g(\mathbf{x} + (l', k')\triangle-\bm{\xi})\sigma^{2}(\bm{\xi})\mathrm{d}\bm{\xi}\right)^{2}\right]}{(2Q+1)^{4}} \rightarrow 0;
\end{equation*}
\item  $C(d_{x_{1}}, d_{x_{2}}) := \Cov\left(\sigma_{I}^{2}(\mathbf{x}), \sigma_{I}^{2}(\mathbf{x} + (d_{x_{1}}, d_{x_{2}}))\right)$ has a finite gradient $\nabla C = (\partial C/\partial d_{x_{1}} , \partial C/\partial d_{x_{2}})$ over $\mathbb{R}^{2}$ 
\end{enumerate}
Then, the local variance estimator $\hat{\sigma}_{I}^{2}(\mathbf{x}, Q) \stackrel{\mathcal{L}_{2}}{\rightarrow} \sigma_{I}^{2}(\mathbf{x})$ when $Q\rightarrow\infty$ and $\triangle = O(Q^{-\tilde{r}})$ for $\tilde{r}>0$. 
\end{thm}
\vspace{4mm} 
\begin{Rem}
The conditions ``$Q\rightarrow\infty$ and $\triangle = O(Q^{-\tilde{r}})$ for $\tilde{r}>0$ " means that we require infill asymptotics. If $0<\tilde{r}<1$, the range of our moving window, $\tau = Q\triangle$ increases as $Q\rightarrow\infty$. On the other hand, if $\tilde{r} = 1$, we have a fixed range and if $\tilde{r}>1$, $\tau\rightarrow0$ as $Q\rightarrow\infty$. Based on the proof of Theorem \ref{thm:l2con}, the rate of the $\mathcal{L}_{2}$ convergence of our estimator increases as $\tilde{r}$ increases. 
\end{Rem}
\vspace{4mm} 
\begin{Cor}\label{cor:convar}
When Theorem \ref{thm:l2con} holds, the mean, variance and normalised variogram of the estimated local variance field converge to those of the true conditional variance, $\sigma^{2}_{I}$.
\end{Cor}
\vspace{4mm} 
\begin{Eg}\label{eg:t4eg}
We show that the assumptions required for Theorem \ref{thm:l2con} hold for Model (\ref{eqn:weg}). From Example \ref{eg:2oCov}, we know that:
\begin{equation*}
\Cov(Y^{2}(\mathbf{x}), Y^{2}(\mathbf{x}^{*})) = \frac{(b\lambda\eta + a^{2}(2\lambda + \eta)\pi)\lambda^{2}}{2\pi^{3}(2\lambda+ \eta)}\exp\left(-\lambda \left(\mathbf{x} - \mathbf{x}^{*}\right)^{T}\left(\mathbf{x} - \mathbf{x}^{*}\right)\right) + A\exp\left(-B \left(\mathbf{x} - \mathbf{x}^{*}\right)^{T}\left(\mathbf{x} - \mathbf{x}^{*}\right)\right). 
\end{equation*}
This implies that $(\ref{eqn:CCcondition})$ in the Appendix is equal to:
\begin{align*}
&\frac{1}{(2Q+1)^{4}}\sum_{\substack{l', k', l, k\\ = -Q}}^{Q} \left[\Cov\left(Y^{2}(\mathbf{x}+ (l, k)\triangle), Y^{2}(\mathbf{x} + (l', k')\triangle) \right) - \Cov\left(\sigma^{2}_{I}(\mathbf{x}+ (l, k)\triangle), \sigma^{2}_{I}(\mathbf{x} + (l', k')\triangle) \right)\right] \\
&= \frac{1}{(2Q+1)^{4}}\sum_{\substack{l', k', l, k\\ = -Q}}^{Q}\frac{(b\lambda\eta + a^{2}(2\lambda + \eta)\pi)\lambda^{2}}{2\pi^{3}(2\lambda+ \eta)}\exp\left(-\lambda \triangle^{2} \left[\left(l - l'\right)^2 + \left(k - k'\right)^2 \right]\right) \\
&= \frac{(b\lambda\eta + a^{2}(2\lambda + \eta)\pi)\lambda^{2}}{2\pi^{3}(2\lambda+ \eta)(2Q+1)^{4}}\left[\sum_{k = -Q}^{Q} \sum_{k'= -Q}^{Q} \exp\left(-\lambda \triangle^{2} \left(k - k'\right)^2 \right) \right]^{2} \\
&=  \frac{(b\lambda\eta + a^{2}(2\lambda + \eta)\pi)\lambda^{2}}{2\pi^{3}(2\lambda+ \eta)(2Q+1)^{4}}\left[(2Q+1) + \sum_{\substack{\left\{k, k' \in\left\{ -Q, Q\right\}:\right.\\ \left. |k - k'| = 1\right\}}} e^{-\lambda \triangle^{2} \left(k - k'\right)^2} + \dots + \sum_{\substack{\left\{k, k' \in\left\{ -Q, Q\right\}:\right.\\ \left. |k - k'| = 2Q\right\}}}  e^{-\lambda \triangle^{2} \left(k - k'\right)^2} \right]^{2} \\
 &=  \frac{(b\lambda\eta + a^{2}(2\lambda + \eta)\pi)\lambda^{2}}{2\pi^{3}(2\lambda+ \eta)(2Q+1)^{4}}\left[(2Q+1) + 4Q \exp\left(-\lambda \triangle^{2}\right) + \dots + 2 \exp\left(-\lambda \triangle^{2} \left(2Q \right)^2 \right) \right]^{2}.
\end{align*}
All of the terms in the square brackets of the last line behave like $O(Q)$ when $Q$ tends to infinity and $\triangle$ behaves like $O(Q^{-\tilde{r}})$ for $\tilde{r}>0$. So, term (\ref{eqn:CCcondition}) behaves like $O(Q^{-2})$ and converges to zero as required.  
\\
Next, we show that the second condition of Theorem \ref{thm:l2con} holds:
\begin{equation*}
C(d_{x_{1}}, d_{x_{2}}) = \frac{b\lambda^{3}\eta}{4\pi^{3}(2\lambda + \eta)}\exp\left(\frac{-\lambda\eta}{2\lambda + \eta} \left(d_{x_{1}}^{2} + d_{x_{2}}^{2}\right)\right) \Rightarrow \frac{\partial C}{\partial d_{x_{i}}} = \frac{-2b\lambda^{4}\eta^{2}d_{t_{i}}}{4\pi^{3}(2\lambda + \eta)^{2}}\exp\left(\frac{-\lambda\eta}{2\lambda + \eta} \left(d_{x_{1}}^{2} + d_{x_{2}}^{2}\right)\right),
\end{equation*}
for $i = 1, 2$. Using L' H\^{o}pital's Rule, $\partial C/\partial d_{x_{i}} \rightarrow 0$ as $d_{x_{i}}\rightarrow\infty$. Thus, $\nabla C = (\partial C/\partial d_{x_{1}}, \partial C/\partial d_{x_{2}})$ is finite over $\mathbb{R}^{2}$.
\end{Eg}
\vspace{4mm}
\begin{thm}
Suppose that Theorem \ref{thm:l2con} holds, i.e. $\hat{\sigma}_{I}^{2}(\mathbf{x}, Q) \stackrel{\mathcal{L}_{2}}{\rightarrow} \sigma_{I}^{2}(\mathbf{x})$ when $Q\rightarrow\infty$ and $\triangle = O(Q^{-\tilde{r}})$ for $\tilde{r}>0$, and that the number of local variance locations, $\widetilde{M} = O(\triangle^{\tilde{\tau}})$ where $\tilde{\tau} > 0$. This means that we have both infill and increasing domain asymptotics. Then if Slutsky's conditions for mean and covariance ergodicity hold for $Y$ and $\sigma^{2}_{I}$, i.e.:
\begin{gather*}
\lim_{X\rightarrow \infty}\frac{1}{X^{2}} \int_{0}^{X}\int_{0}^{X} \Cov\left(Y\left(\mathbf{x}\right), Y\left(\mathbf{x} + \mathbf{h}\right)\right) \mathrm{d}\mathbf{h} = 0, \text{ } \lim_{X\rightarrow \infty}\frac{1}{X^{2}} \int_{0}^{X}\int_{0}^{X}\Cov^{2}\left(Y\left(\mathbf{x}\right), Y\left(\mathbf{x} + \mathbf{h}\right)\right) \mathrm{d}\mathbf{h} = 0 \\
\lim_{X\rightarrow \infty}\frac{1}{X^{2}} \int_{0}^{X}\int_{0}^{X}\Cov\left(\sigma^{2}_{I}\left(\mathbf{x}\right), \sigma^{2}_{I}\left(\mathbf{x} + \mathbf{h}\right)\right) \mathrm{d}\mathbf{h} = 0 \text{ and } \lim_{X\rightarrow \infty} \frac{1}{X^{2}}  \int_{0}^{X}\int_{0}^{X} \Cov^{2}\left(\sigma^{2}_{I}\left(\mathbf{x}\right), \sigma^{2}_{I}\left(\mathbf{x} + \mathbf{h}\right)\right) \mathrm{d}\mathbf{h} = 0,
\end{gather*}
the empirical variance and normalised variogram of $Y$ are consistent. Furthermore, the empirical mean, variance and normalised variogram of $\hat{\sigma}^{2}_{I}$ converge in probability to the respective theoretical quantities of $\sigma^{2}_{I}$. 
\end{thm}
\begin{proof}
As suggested on page 57 of \cite{Cressie1993}, under infill and increasing domain asymptotics, Slutsky's conditions for the mean and covariance ergodicity of the discrete process formed by sampling a continuous process converges to those of the continuous process itself. For example:
\begin{align*}
&\lim_{X, N\rightarrow\infty}\frac{1}{N^{2}}\sum_{i = 1}^{N}\sum_{j = 1}^{N}\Cov\left(Y\left(\mathbf{x}\right), Y\left(\mathbf{x} + \left(\frac{iX}{N},  \frac{jX}{N}\right)\right)\right) = 0 \text{ where } X = O(N^{\tilde{s}}) \text{ with } \tilde{s}<1,\\
&\Leftrightarrow  \lim_{X\rightarrow \infty}\frac{1}{X^{2}}\lim_{N\rightarrow\infty}\frac{X^{2}}{N^{2}}\sum_{i = 1}^{N}\sum_{j = 1}^{N}\Cov\left(Y\left(\mathbf{x}\right),  Y\left(\mathbf{x} + \left(\frac{iX}{N}, \frac{jX}{N}\right)\right)\right) = 0 \\
&\Leftrightarrow\lim_{X\rightarrow \infty}\frac{1}{X^{2}} \int_{0}^{X}\int_{0}^{X} \Cov\left(Y\left(\mathbf{x}\right), Y\left(\mathbf{x} + \mathbf{h}\right)\right)\mathrm{d}\mathbf{h} = 0
\end{align*}
since in the second line, $X$ is fixed when we vary $N$. 
\\ 
Thus, if Slutsky's conditions hold, the sample mean and covariances for $Y$ and $\sigma_{I}^{2}$ converge to their theoretical values. Since the normalised variograms are formed from the covariances, the empirical variograms are also consistent. 
\\
Now we show that the empirical mean, variance and normalised variogram of $\hat{\sigma}^{2}_{I}$ converge to the respective theoretical equivalents of $\sigma^{2}_{I}$. Under the conditions of Theorem \ref{thm:l2con},  $\hat{\sigma}_{I}^{2}(\mathbf{x}, Q) \stackrel{p}{\rightarrow} \sigma_{I}^{2}(\mathbf{x})$ for arbitrary $\mathbf{x}\in\mathbb{R}^{2}$. Since for $X_{n}\stackrel{p}{\rightarrow}X$ and $Y_{n}\stackrel{p}{\rightarrow}Y$, we have that $aX_{n} + bY_{n} \stackrel{p}{\rightarrow} aX + bY$, it follows that the sample mean of an estimated local variance surface at $\widetilde{M}$ locations, $\widetilde{M}^{-1}\sum_{i = 1}^{\widetilde{M}}\hat{\sigma}_{I}^{2}(\mathbf{s}_{i})\stackrel{p}{\rightarrow} \widetilde{M}^{-1}\sum_{i = 1}^{\widetilde{M}}\sigma_{I}^{2}(\mathbf{s}_{i})$ as $Q\rightarrow \infty$. Since we have mean ergodicity, this in turn converges in probability to $\mathbb{E}\left[ \sigma^{2}(\mathbf{x})\right]$ as $\widetilde{M}\rightarrow \infty$. Following similar arguments with the appropriate use of the Continuous Mapping Theorem, we can show that the sample variance and normalised variogram of $\hat{\sigma}_{I}^{2}$ converge in probability to the theoretical variance and normalised variogram of $\sigma_{I}^{2}$.
\end{proof}
\vspace{2mm}
\begin{Rem}
Recall that $\Cov\left(Y\left(\mathbf{x}\right), Y\left(\mathbf{x} + \mathbf{h}\right)\right) =  \mathbb{E}\left[\sigma^{2}(\mathbf{0})\right]\int_{\mathbb{R}^{2}}g(\mathbf{w})g(\mathbf{w} + \mathbf{h})\mathrm{d}\mathbf{w}$ and:
\begin{align*}
\Cov\left(\sigma^{2}_{I}\left(\mathbf{x}\right), \sigma^{2}_{I}\left(\mathbf{x} + \mathbf{h}\right)\right)) &= \int_{\mathbb{R}^{2}} \int_{\mathbb{R}^{2}} g^{2}(\mathbf{x}-\bm{\xi})g^{2}(\mathbf{x} + \mathbf{h}-\bm{\xi}^{*})\Cov\left(\sigma^{2}(\bm{\xi}), \sigma^{2}(\bm{\xi}^{*})\right)\mathrm{d}\bm{\xi}\mathrm{d}\bm{\xi}^{*} \\
&= \int_{\mathbb{R}^{2}}\left[ \int_{\mathbb{R}^{2}} g^{2}(\mathbf{u}-\mathbf{w})g^{2}(\mathbf{u} + \mathbf{h})\mathrm{d}\mathbf{u}\right]\Cov\left(\sigma^{2}(\mathbf{w}),\sigma^{2}(\mathbf{0})\right)\mathrm{d}\mathbf{w}, 
\end{align*}
where $\mathbf{w} = \bm{\xi} - \bm{\xi}^{*}$, $\mathbf{u} = \mathbf{x} - \bm{\xi}^{*}$ and we have written the covariance of $\sigma^{2}$ in terms of the lag $\mathbf{w}$. Using these expressions, we can express Slutsky's conditions in terms of $g$ and the first two moments of $\sigma^{2}$. 
\end{Rem}
\vspace{2mm}
\begin{Lem}
For Model (\ref{eqn:weg}), our two-step moments-matching estimators, $\hat{\lambda}$, $\hat{a}$, $\hat{a}_{2}$, $\hat{b}$ and $\hat{\eta}$ are consistent under infill and increasing domain asymptotics.
\end{Lem}
\begin{proof}
Slutsky's ergodic conditions hold for Model (\ref{eqn:weg}) since the covariances of $Y$ and $\sigma^{2}_{I}$ can be written in terms of Gaussian densities. From Example \ref{eg:t4eg}, we also know that Theorem \ref{thm:l2con} holds. So, the empirical variance and normalised variogram of $Y$ are consistent, and the empirical mean, variance and normalised variogram of $\hat{\sigma}^{2}_{I}$ converge in probability to the theoretical quantities of $\sigma^{2}_{I}$. By repeated use of Slutsky's Theorem and the Continuous Mapping Theorem, it is easy to show that our parameter estimators are consistent. 
\end{proof}
\begin{Rem}
We proved the consistency of our parameter estimators for Model (\ref{eqn:weg}). As hinted at in Remark \ref{rem:general}, these consistency conditions can be checked for the two-step moments-matching estimators of other VMMAs when they have parameters representative of the variance and correlation of $Y$ and those of $\sigma^{2}$, and when an analytical expression for $\Cov(\sigma^{2}_{I}(\bm{\xi}_{i}), \sigma^{2}_{I}(\bm{\xi}_{j}))$ is available. 
\end{Rem}

\subsection{Estimation in practice}

Under infill and increasing domain asymptotics, our parameter estimators are consistent when the range of our moving window behaves in an appropriate manner with respect to the grid size. To use the moments-based method in practice, we need to fix this range by choosing our tuning parameter $q$. In this subsection, we introduce a way to select q by the so-called maximum regional variance, examine the potential and limitations of such a method as well as illustrate why we chose to define the local variance estimator by the mean of the squared data values instead of a sample variance formula. 

\subsubsection{Choosing $q$ via the maximum regional variance}

The size of the volatility clusters, i.e. the regions of high volatility, gives us information about the variance and correlation parameters of $\sigma^{2}$. In the case of Model (\ref{eg:weg}), these are $b$ and $\eta$ respectively. Thus, we choose $q$ to identify the volatility cluster size. Let us define the empirical regional variance at $\bm{\xi}_{i}\in\mathbb{R}^{2}$ by:
\begin{equation*}
\hat{\kappa}_{2}^{I}(\bm{\xi}_{i}) = \frac{1}{|Y_{A_{i}}|-1}\sum_{\mathbf{x}\in A_{i}}\left(Y(\mathbf{x}) - \overline{Y}_{A_{i}}\right)^{2},
\end{equation*}
where $Y_{A_{i}}$ denotes the data values within the subregion $A_{i}$ whose centre is $\bm{\xi}_{i}$ and $\overline{Y}_{A_{i}}$, the empirical mean of $Y_{A_{i}}$.
\\
If $q$ is too small, the similarity of the $Y$s within the capture windows will cause the $\hat{\kappa}_{2}^{I}$ values to be small. As we increase $q$ to the radius of a volatility cluster, there will be capture windows for which the $Y$s within have increasingly different values. This causes the maximum regional variance (MRV), i.e. the maximum of the $\hat{\kappa}_{2}^{I}$ values over the $\bm{\xi}_{i}$ values, to increase. If we increase $q$ further, the unconditional stationarity of our VMMA and the stationarity of the $\sigma^{2}$ layer will cause the $\hat{\kappa}_{2}^{I}$ values and hence its maximum over the subregions to drop. Thus, we can identify appropriate values of $q$ by calculating the MRV for a range of values and choosing the $q$ values at the MRV peaks. 
\\
Figure \ref{fig:MrvCompare} shows the MRV values of a data set, Data set 1, simulated from Model (\ref{eqn:weg}) with $\lambda = \eta = 4$, $a = 1$ and $b =2$ in Plot (a) and that simulated from a GMA with the same correlation structure and underlying Gaussian noise in Plot (b) for $q$ ranging from $9$ to $51$. To generate each data set, we chose $\triangle = 0.05$ and $M = 201$. We see that although both plots feature peaks in MRV, the peak for the VMMA is larger in magnitude. This indicates the presence of stochastic volatility: the higher the magnitude of the peak, the more prominent the associated volatility cluster. On the other hand, while the peak occurs at $q = 21$ for the VMMA, it occurs at $q = 33$ for the GMA. The ``volatility clusters'' identified in the GMA data have a larger radius. Since larger volatility clusters are less distinct over the same data region, this is indicative of its constant volatility. From this observation, in cases of multiple peaks, we prioritise peaks at lower $q$ values.

\subsubsection{Choice of the local variance estimator}

Here, we give some reasoning behind our choice of $\hat{\sigma}_{I}^{2}$ as the local variance estimator instead of $\hat{\kappa}_{2}^{I}$. 
\\
Although computing $\hat{\kappa}_{2}^{I}$ allows us to identify the volatility cluster sizes, these estimates do not work well as local variance estimates. In particular, they underestimate the local variances near the centres of volatility clusters. The key difference between $\hat{\kappa}_{2}^{I}$ and $\hat{\sigma}^{2}_{I}(\bm{\xi}_{i})$ is that instead of using $\overline{Y}_{A_{i}}$ to estimate regional mean, we set it to the theoretical zero.
\\
Figures \ref{fig:k2SigmaI}(b) and (e) show the fields of $\hat{\kappa}_{2}^{I}$ and $\hat{\sigma}^{2}_{I}$ obtained from Data set $1$ for $q = 21$. Comparing Plots (b) and (e) (and similar plots for other $q$ values), we notice that while high $\hat{\kappa}_{2}^{I}$ values mean that we are at the boundaries of the volatility clusters, high $\hat{\sigma}^{2}_{I}$ values mean that we are in the clusters themselves. Using the MRV to choose $q$ enables us to choose a capture window size so that we find the size of the highest volatility cluster which lies away from the boundaries (so as to be captured by the capture window). From the location of the red diamond in Plot (e), we see that the point whose regional variance is the MRV for $q = 21$ lies at the slope of this cluster. Note that the white border in the heat plots denote areas for which no estimates are obtained from the moving window approach and this widens as $q$ increases.

\subsubsection{Inference results and discussion} 

We apply our inference procedure to $100$ simulated data sets for Model (\ref{eqn:weg}). These were generated with random seeds $1$ to $100$. Figure \ref{fig:BP100} shows the estimates and selected $q$ values from the $100$ data sets. The outliers are labelled by their data set indices. Since the medians of the estimates (denoted by the bold black horizontal lines) lie close to the true parameter values (represented by the red horizontal lines) in all cases, our inference method works reasonably well.
\\
Despite the promising results, we note that the moving window approach is sensitive to the most prominent estimated volatility cluster which is in turn influenced by the realisation of the underlying Gaussian noise. In individual cases, special care is also required when we have overlapping clusters and sudden surges in amplitude. 
\\
In Figure \ref{fig:hatSIVG}(a), we show the true conditional variance surface of Data set 1. Plots (b) and (c) show its estimated local variance surface as well as that for its corresponding GMA data set. For the VMMA, $\hat{\sigma}^{2}_{I}$ locates the volatility clusters reasonably well. In addition, as can be seen from the absolute difference between $\sigma^{2}(\bm{\xi})$ and $\hat{\sigma}^{2}_{I}$ in Figure \ref{fig:hatSIVG}(d), the regions of higher error occur at areas of high conditional variances. Attributing this error to the realisation of the background Gaussian noise is also consistent with the fact that volatility clusters are also identified for the GMA in Figure \ref{fig:hatSIVG}(c).
\\
To illustrate the other limitations of the method which we mentioned, we examine several outliers in Figure \ref{fig:BP100}. First, we focus on the outliers in the values of $q$ such as that corresponding to Data set $38$. Figure \ref{fig:Seed38} shows its MRV chart, true conditional variance surface, and estimated local variance surface for the selected $q$ values and the median $q$ value calculated over the $100$ data sets, $23$. We see that the most prominent variance cluster (i.e. the brightest spot in Plot (b) which lies away from the boundaries) do not translate into the region of highest estimated local variance (i.e. the brightest spot in Plot (d)). This means that our inference method focuses its attention on theoretically less prominent clusters, in this case those at the bottom left corner of Plot (b). Since the estimates for these clusters happen to be concentrated together without much distinction from each other, our inference method eventually groups them together to form one big cluster and selects a larger $q$ value. This leads to a lower $\widehat{B}$ value since the overlaps in the capture window when computing $\hat{\sigma}_{I}^{2}$ determine the amount of correlation the values have.

\clearpage

\begin{figure}[tbp]
\centering
\includegraphics[width =5.8in, height = 3in, trim = 0.1in 0.4in 0.1in 0.1in]{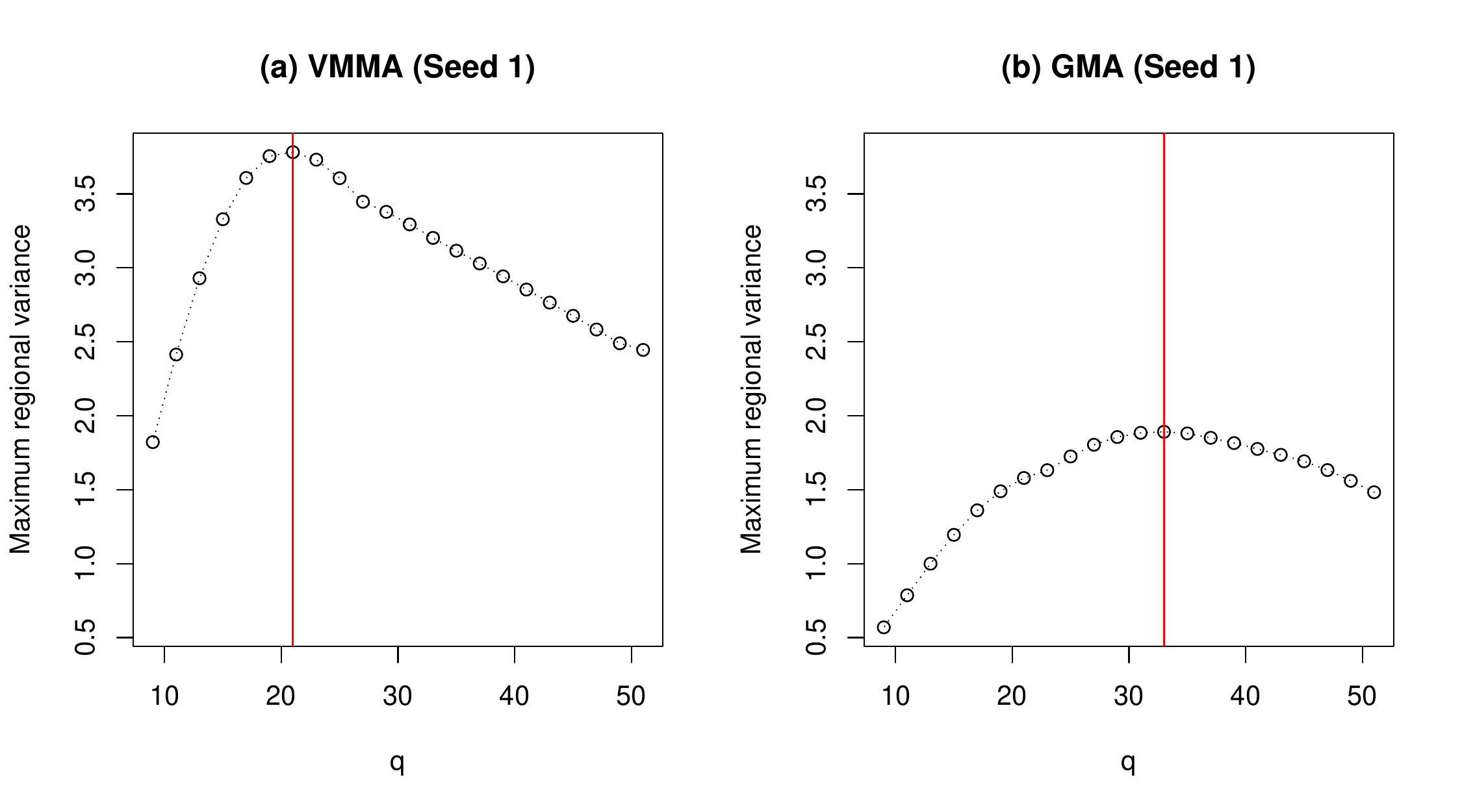}
\caption{MRV as a function of the capture window size $q$ for: (a) the VMMA data set; (b) the GMA data set. The red vertical lines denote the peaks of the MRVs.} 
\label{fig:MrvCompare}
\end{figure}

\begin{figure}[tbp]
\centering
\includegraphics[width =6in, height = 4.1in, trim = 0.5in 0.4in 0.5in 0.4in]{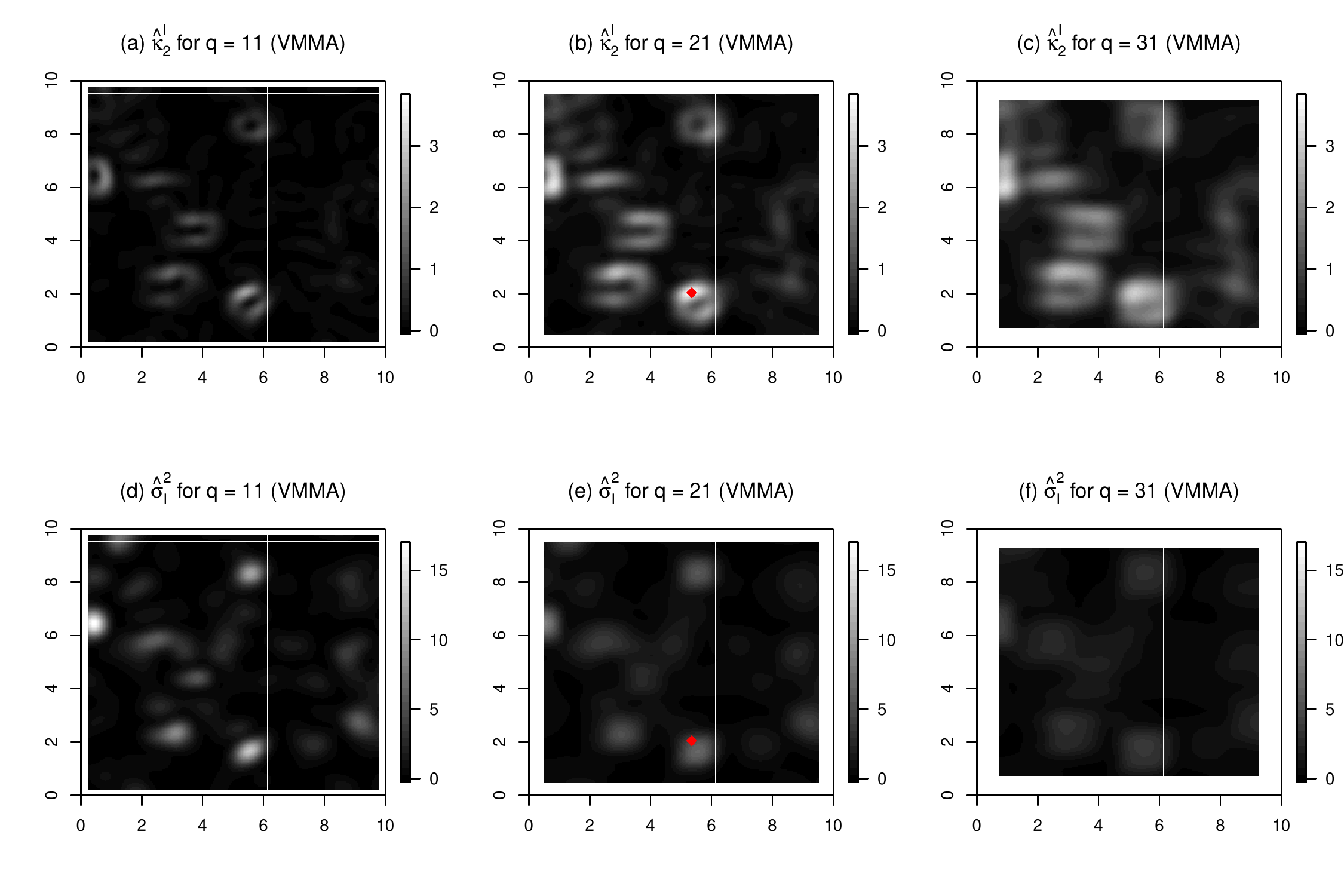}
\caption{Data set 1: The empirical regional variance ($\hat{\kappa}_{2}^{I}$) and the estimated local variance ($\hat{\sigma}^{2}_{I}$) surfaces for $q = 11$, $21$ and $31$. The red diamonds in Plots (b) and (e) denote the point whose regional variance is the MRV for $q = 21$.} 
\label{fig:k2SigmaI}
\end{figure}

\clearpage

\begin{figure}[tbp]
\centering
\includegraphics[width =5.4in, height = 2in, trim = 1in 0.1in 1in 0in]{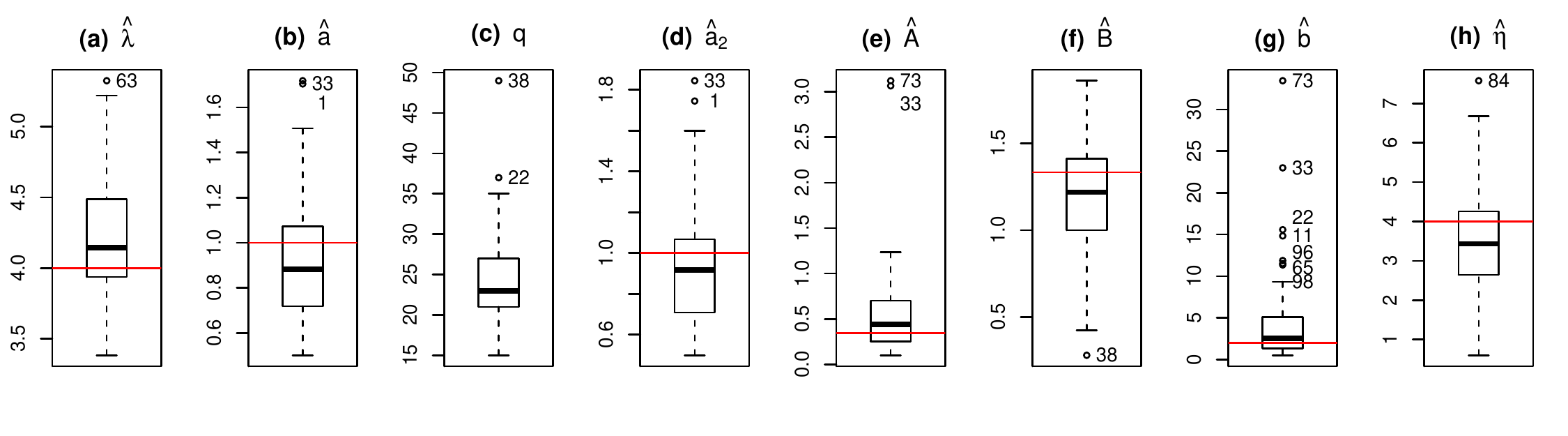}
\caption{Box plots of estimates and selected $q$ values from $100$ VMMA data sets. The red horizontal lines denote the true parameter values and the numbers in the box plots represent the index of the data sets which give rise to the corresponding outliers.} 
\label{fig:BP100}
\end{figure}

\begin{figure}[tbp]
\centering
\includegraphics[width =5in, height = 5.2in, trim = 1in 0.2in 1in 0.2in]{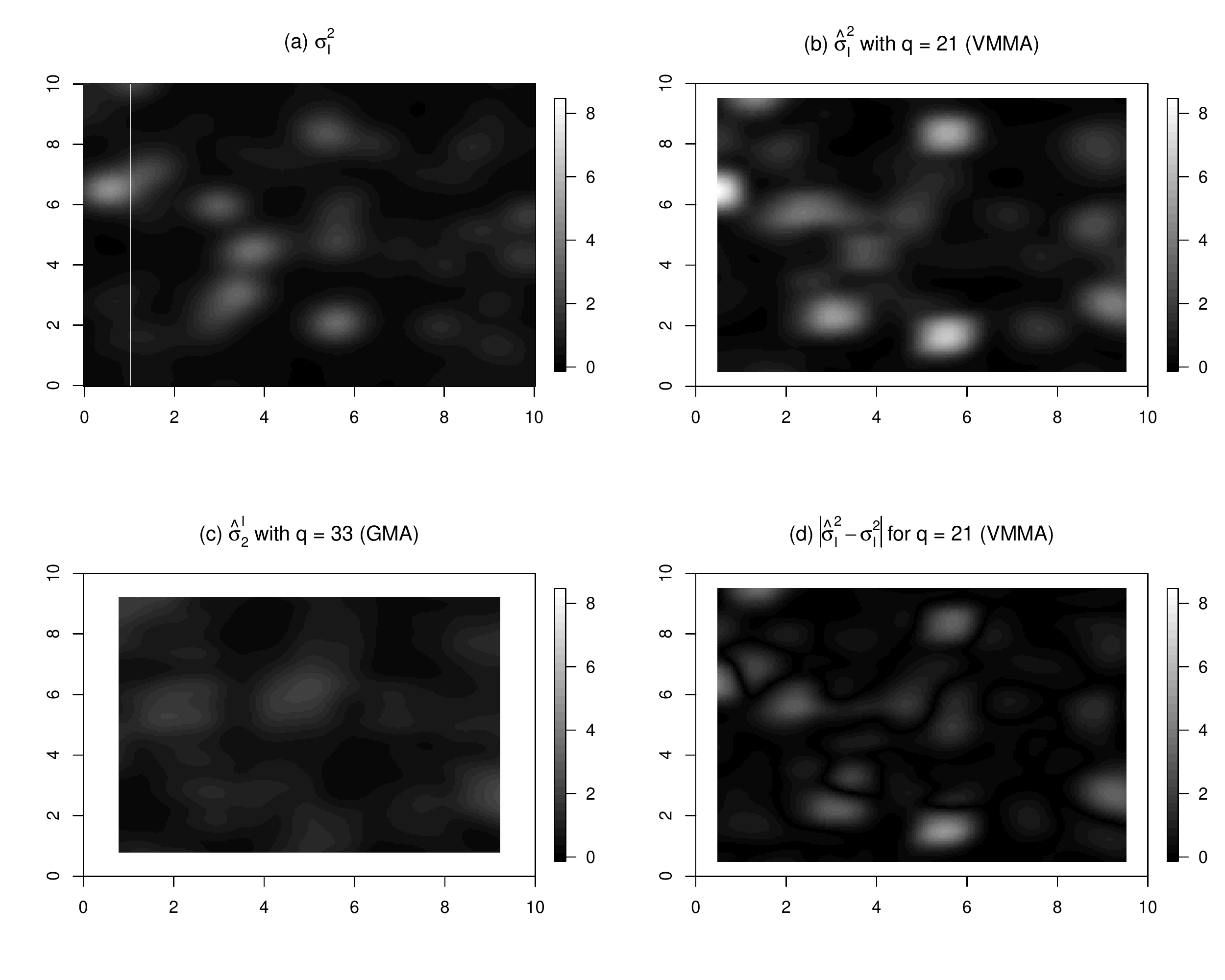}
\caption{Data set 1: (a) The true conditional variance surface ($\hat{\sigma}^{2}_{I}$); (b)-(c) the estimated local variance ($\hat{\sigma}^{2}_{I}$) surfaces for the VMMA data set with $q = 21$ and the GMA data set with $q = 33$ respectively; and (d) the absolute difference between $\hat{\sigma}^{2}_{I}$ and $\hat{\sigma}^{2}_{I}$.} 
\label{fig:hatSIVG}
\end{figure}

\clearpage

\begin{figure}[tbp]
\centering
\includegraphics[width =5in, height = 5.2in, trim = 1in 0.2in 1in 0.2in]{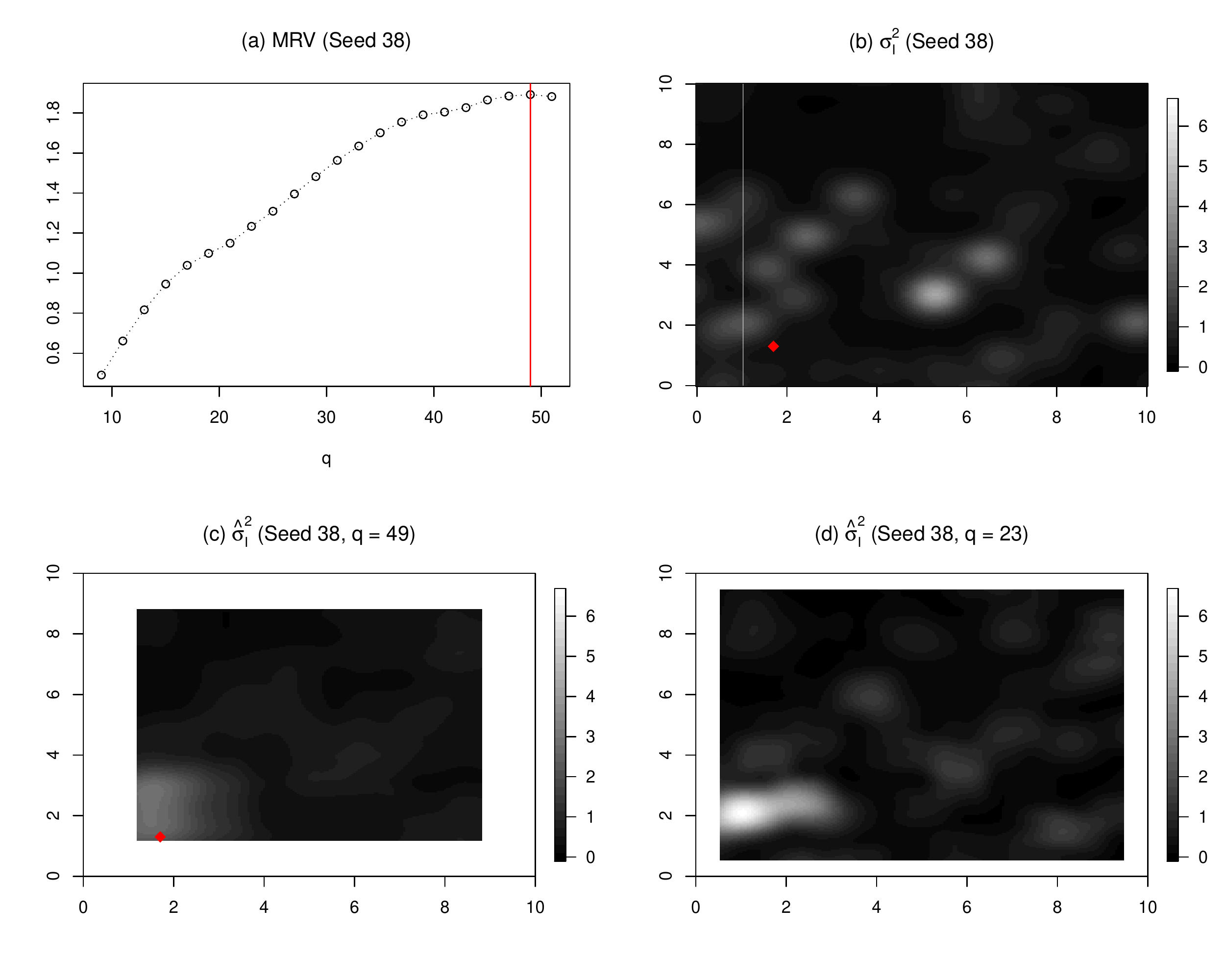}
\caption{VMMA data set 38: (a) The MRV chart for $q$ ranging from $9$ to $51$; (b) the true conditional variance ($\hat{\sigma}^{2}_{I}$) surface; (c)-(d) the estimated local variance ($\hat{\sigma}^{2}_{I}$) surfaces with the selected $q$ value ($49$) and the median $q$ value ($23$).}
\label{fig:Seed38}
\end{figure}

\begin{figure}[tbp]
\centering
\includegraphics[width =5.6in, height = 2.25in, trim = 1in 0.2in 1in 0.2in]{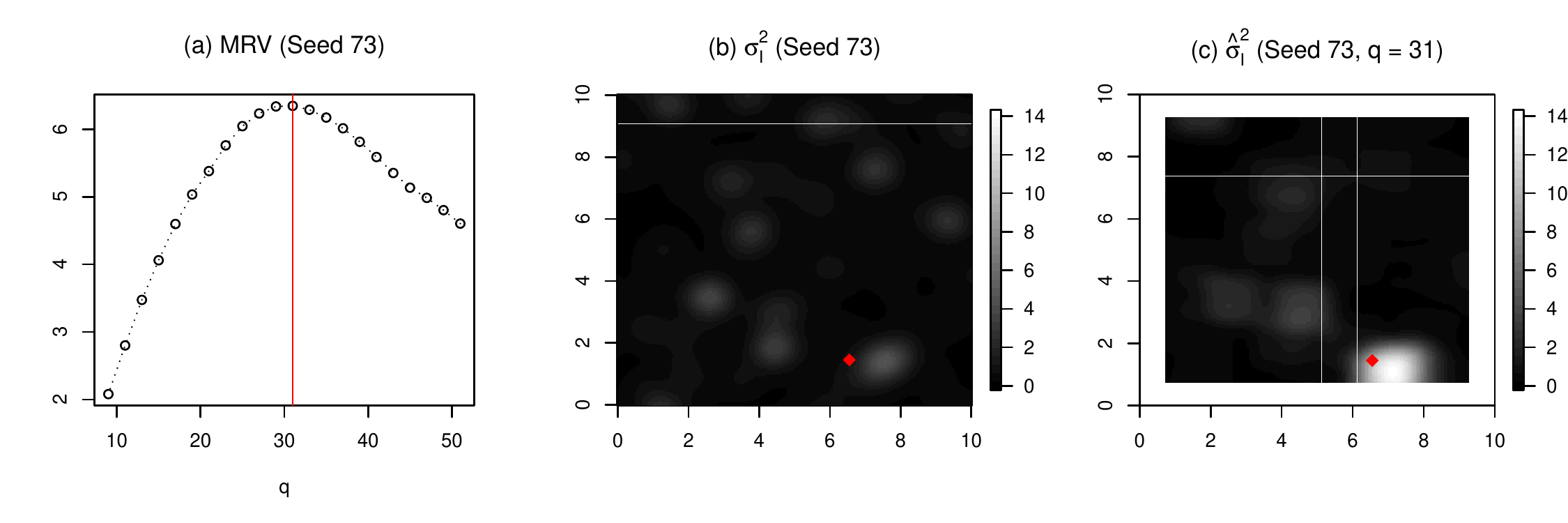}
\caption{VMMA data set 73: (a) The MRV chart for $q$ ranging from $9$ to $51$; (b) the true conditional variance ($\hat{\sigma}^{2}_{I}$) surface; (c) the estimated local variance ($\hat{\sigma}^{2}_{I}$) surface with the selected $q$ value. In Plot (a), the red vertical lines denote the peaks in the MRV chart. The lighter regions in the heat plots of (b) and (c) denote higher values.} 
\label{fig:Seed73}
\end{figure}

\clearpage

\begin{figure}[tbp]
\centering
\includegraphics[width = 3.2in, height = 2.2in, trim = 1in 0.4in 1in 0in]{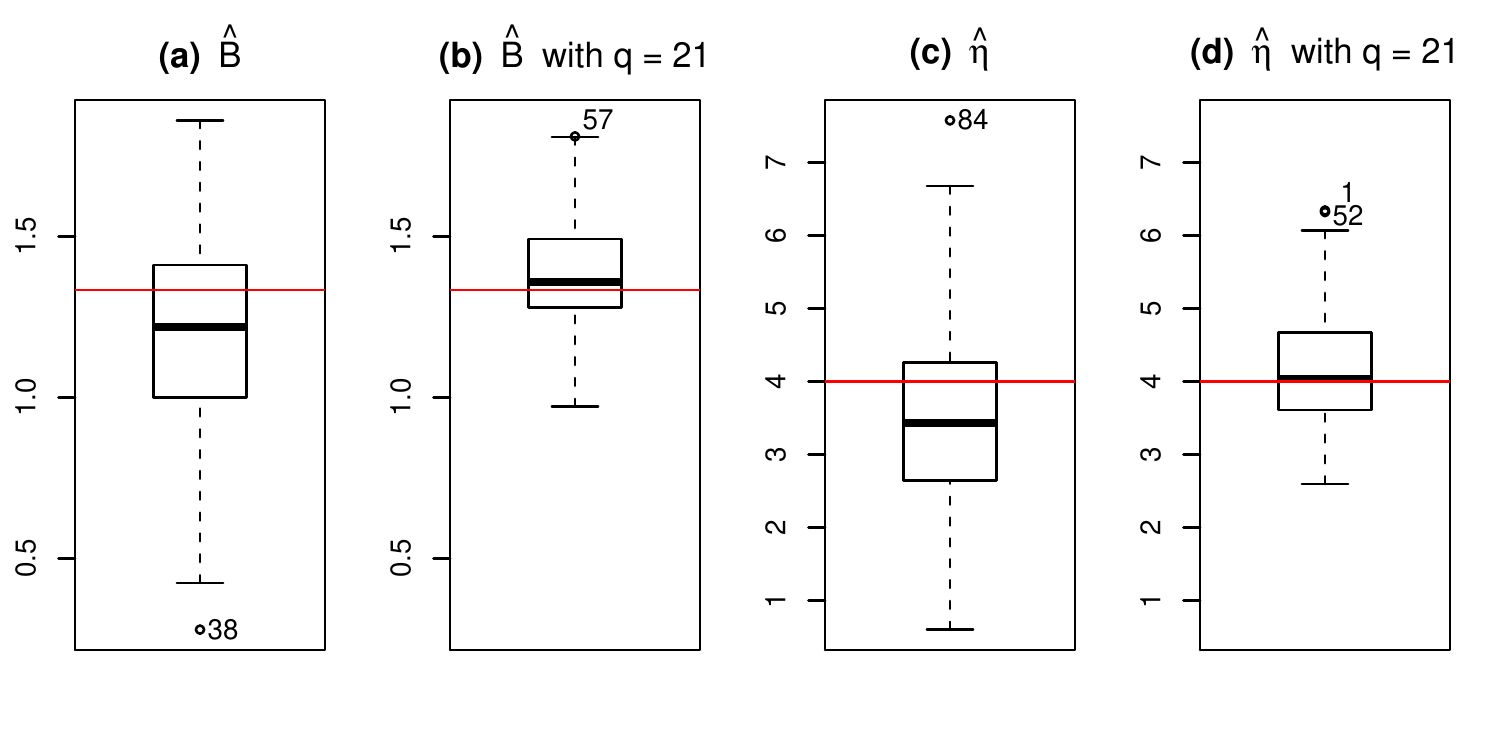}
\caption{Box plots of estimates from $100$ VMMA data sets when we set $q = 21$. The red horizontal lines denote the true parameter values.} 
\label{fig:BPfixq}
\end{figure}

The next interesting outlier corresponds to Data set $73$. In this case, $A$ which is the variance of $\sigma_{I}^{2}$ is overestimated. Looking at the $\hat{\sigma}_{I}^{2}$ surface in Figure \ref{fig:Seed73}(c), we see that this is because one sharp peak in $\sigma_{I}^{2}$ is very distinct from the rest of the values. Due to the overestimation of $A$, the parameter $b$ is also overestimated. 
\\
The two limitations mentioned are related to the multiplicative Gaussian noise in our simulations and further work is required to overcome these. Nevertheless, the potential of the estimation method, given a well-chosen $q$ value, can be seen from the better accuracy and precision in the $b$ and $\eta$ estimates in Figure \ref{fig:BPfixq} when $q = 21$. 

\section{Empirical example} \label{sec:emp}

In the previous section, promising inference results were obtained under strong heteroskedasticity. For Model (\ref{eqn:weg}), this means that the parameters $b$ and $\eta$ are of comparable magnitude to $\lambda$ so that there is sufficient variation in the conditional variance surface to identify the volatility clusters. In general, it is harder to estimate the parameters of $\sigma^{2}$ under low heteroskedasticity. In this section, we show that even in this case, modelling the data by a VMMA instead of a GMA can be beneficial. Specifically, better prediction intervals can be obtained by using the estimated local variance surface. 
\\
We illustrate this with a data set of sea surface temperature anomalies (SSTA) for the week $29^{\text{th}}$ May 2016 to $4^{\text{th}}$ May 2016 \cite[]{website:SSTA}. These are calculated with respect to the 1971-2000 climatology and thus indicate how SST has changed at different spatial locations. The data, which is pictured in Figure \ref{fig:SSTAmp}(a), is given in $^{\circ}$C and lies on a $1^{\circ}$ latitude/longitude grid in the Pacific Ocean between $150.5^{\circ}$E and $234.5^{\circ}$E, and $-69.5^{\circ}$N and $59.5^{\circ}$N.  
\\
Before we start our analysis, we randomly choose $100$ test points away from the boundaries (denoted by the black circles in Figure \ref{fig:SSTAmp}(a)) and remove them from our data. Next, we apply median polishing on the remaining data to obtain a spatial trend. Note that this also gives trend estimates where we had missing values. The median polish algorithm has been used for various data sets in \cite{Cressie1993}. From the estimated median polish surface in Figure \ref{fig:SSTAmp}(b), we see that the trend is more prominent in the direction of the latitude and captures some of the extreme values near $60^{\circ}$N. 
\\
By treating $10^{\circ}$ latitude/longitude as one unit, we fit Model (\ref{eqn:weg}) to the median polish residuals in Figure \ref{fig:SSTAmp}(c). All the averaging required in for example, mean and variance calculations, have been adapted to deal with the missing data. Recall that Model (\ref{eqn:weg}) is given by:
\begin{equation*}
\left.\begin{aligned}
Y(\mathbf{x})&= \int_{\mathbb{R}^{2}}\frac{\lambda}{\pi}\exp\left(-\lambda\left(\mathbf{x}-\bm{\xi}\right)^{T}\left(\mathbf{x}-\bm{\xi}\right)\right)\sigma(\bm{\xi})W(\mathrm{d}\bm{\xi}),\\
\text{where } \sigma^{2}(\bm{\xi})&= \int_{\mathbb{R}^{2}}\frac{\eta}{\pi}\exp\left(-\eta\left(\bm{\xi}-\mathbf{u}\right)^{T}\left(\bm{\xi}-\mathbf{u}\right)\right)L(\mathrm{d}\mathbf{u}).
\end{aligned}
\right\}
\qquad
\end{equation*}

\clearpage

\begin{figure}[tbp]
\centering
\includegraphics[width =5.8in, height = 2in, trim = 1in 0.2in 1in 0.3in]{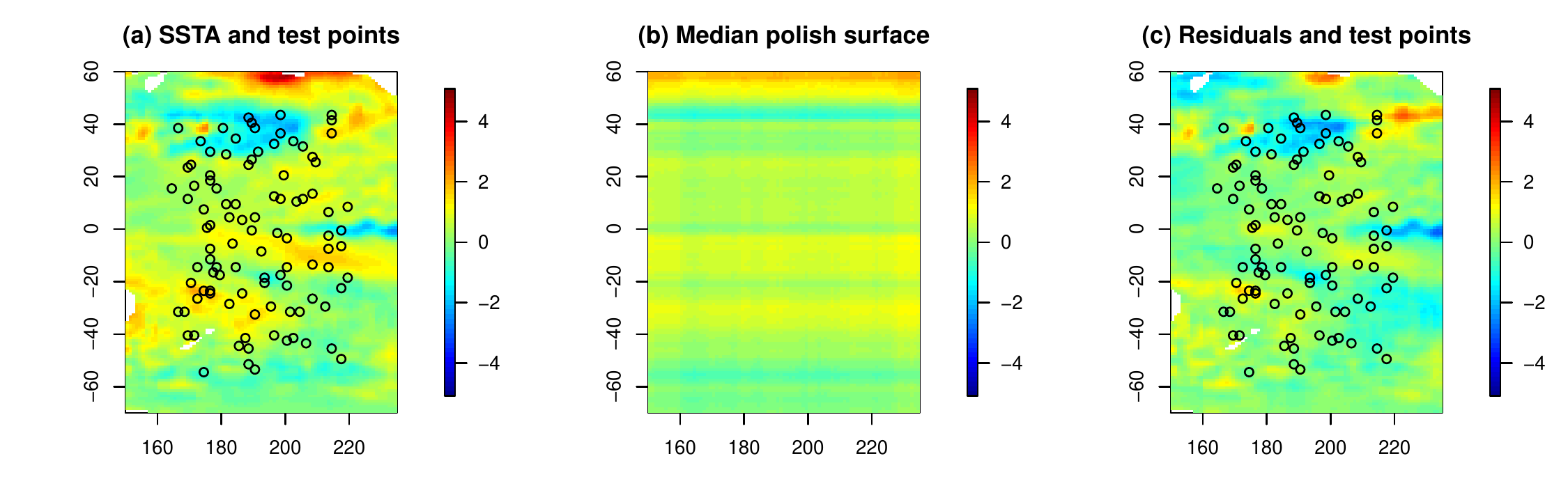}
\caption{(a) The SSTA data set (in $^{\circ}$C) where the black circles mark the $100$ test points. The white regions denote missing data due to land mass; (b) the fitted median polish surface; and (c) the median polish residuals.}
\label{fig:SSTAmp}
\end{figure}

\begin{figure}[tbp]
\centering
\includegraphics[width =3.6in, height = 4.2in, trim = 1in 0.2in 1in 0.2in]{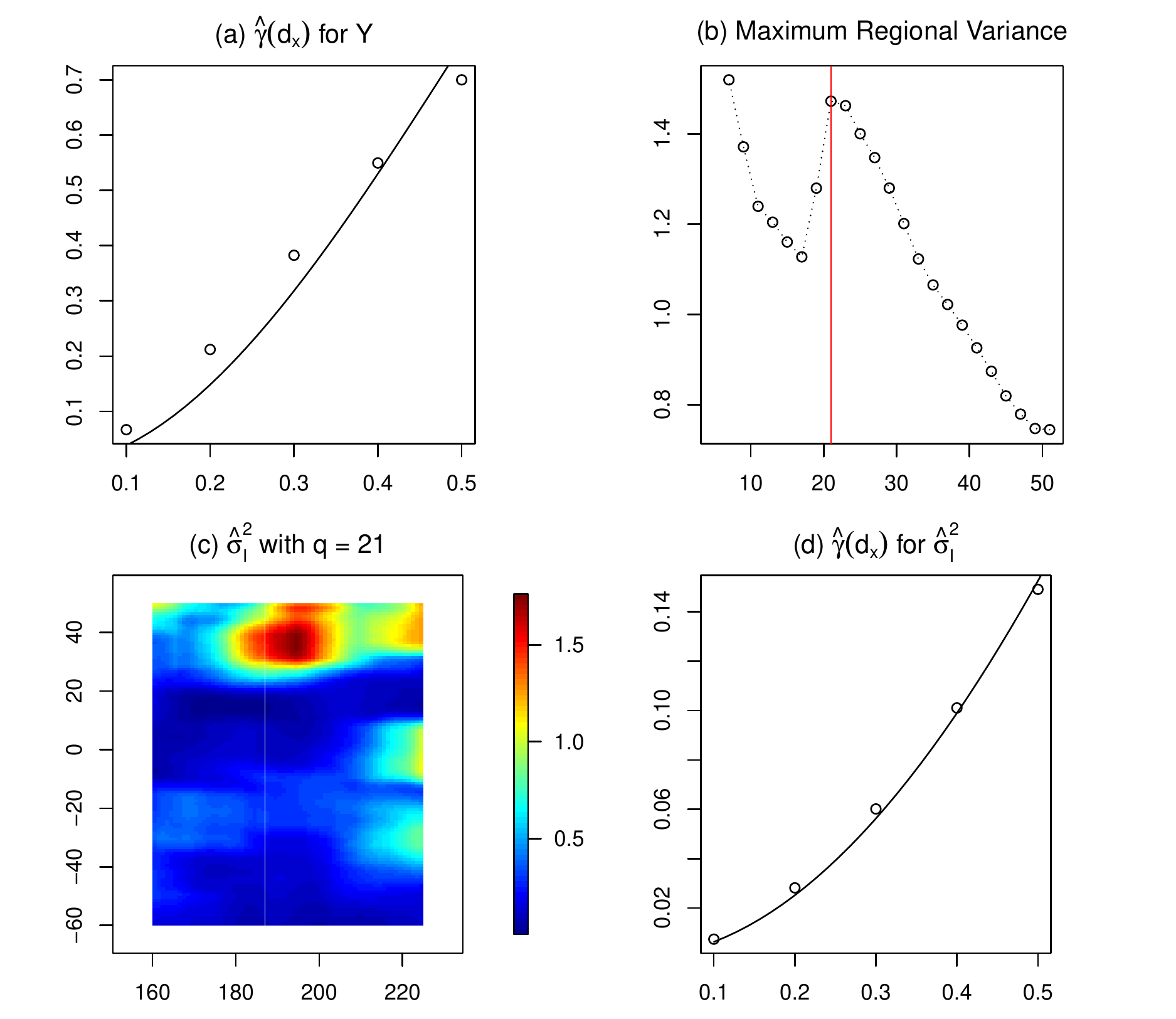}
\caption{(a) The least-squares fit of the normalised variogram of $Y$; (b) the MRV plot; (c) the estimated local variance ($\hat{\sigma}^{2}_{I}$) surface with the selected $q$ value ($21$); and (d) the least-squares fit of the normalised variogram of $\hat{\sigma}^{2}_{I}$.}
\label{fig:Fitplots}
\end{figure}

\clearpage
Here, the rate parameters $\lambda, \eta > 0$, $W$ is a homogeneous standard Gaussian basis independent of the stochastic volatility $\sigma(\bm{\xi})$ and $L$ is a homogeneous subordinator basis with the mean and variance of its seed being $a$ and $b>0$. In this case, it seems reasonable to use Gaussian kernels because of tractable covariances and since they can be inferred from physical ocean dynamics \cite[]{Higdon1998, Barnes1964}. 
\\
In our simulation experiments in Section \ref{sec:est}, we used the first spatial lag to fit the normalised variograms. For empirical data, however, better fits to to the empirical normalised variograms can be obtained by fitting a curve using least-squares. Figure \ref{fig:Fitplots}(a) shows the least squares fit for the estimated normalised variogram of the median polish residuals (i.e. $Y$) when we use $5$ spatial lags. The consistency of the resultant estimators can be derived by applying an edited version of Theorem 3.1 on page 70 of \cite{LLC2002} to our original proofs. 
\\
Based on the MRV plot in Figure \ref{fig:Fitplots}(b), we use $q = 21$ in the moving window step of the inference. The presence of a peak as well as the eventual decrease to the global variance in the MRV plot support the use of the VMMA for this data set. From the estimated local variance surface in Figure \ref{fig:Fitplots}(c), we identify a prominent volatility cluster near $40^{\circ}$N/$190^{\circ}$E. This corresponds to the region in Figure \ref{fig:SSTAmp}(a) where we observe adjacent high and low anomalies.
\\
With the estimated local variance surface, we continue our inference method. The least squares fit for the estimated normalised variogram of $\hat{\sigma}^{2}_{I}$ is shown in Figure \ref{fig:Fitplots}(d). This is much better than the variogram fit for $Y$. The final parameter estimates obtained are $\hat{\lambda} = 3.848, \hat{a} = 0.676, \hat{a}_{2} = 0.636, \hat{b} = 3.499 \text{ and } \hat{\eta} = 0.691$. 
\\
Despite the low value for $\hat{\eta}$ which implies large, diffuse volatility clusters, there is still some advantage of modelling the median polish residuals by a VMMA instead of a GMA. Using the estimated local variance surface and the median polish surface, we can construct $95\%$ confidence intervals for the SSTA values at our test points since we know that $Y(\mathbf{x})|\mathcal{F}^{\sigma} \sim N(0, \sigma^{2}_{I}(\mathbf{x}))$. At the same time, we can compute $95\%$ confidence intervals from the corresponding GMA model for which $\sigma^{2}_{I}(\mathbf{x}) = \hat{a}\hat{\lambda}/2\pi$ for all test points. Although the VMMA gives more narrow intervals when volatility is low (and wider intervals when volatility is high), we find that $93/100$ test points lie within the confidence intervals constructed by the VMMA while only $89/100$ of them lie within those constructed by the GMA. The benefit of modelling with a VMMA is expected to be enhanced under stronger heteroskedastic behaviour.

\section{Conclusion and further work} \label{sec:confurther}

In this paper, we focused on volatility modulated moving averages (VMMAs) and their ability to exhibit spatial heteroskedasticity. These processes extend the definition of a Gaussian moving average or process convolution by introducing a stationary stochastic volatility field in the integral. 
\\
After providing a summary of how such an integral is constructed in Section \ref{sec:btheory}, we derived several distributional properties of a VMMA in Section \ref{sec:tprop}. These were used to develop a two-step moments-matching estimation procedure in  Section \ref{sec:est}. In the first step, we relied on the stationarity of a VMMA together with its second-order properties; in the second step, we examined the conditional non-stationarity in the form of local variances and used the second-order properties of the conditional variance process. Consistency of the resulting estimators can be proved under suitable double asymptotics. 
\\
In Section \ref{sec:sim}, we developed a discrete convolution simulation algorithm for VMMAs and provided semi-explicit formulas for the mean squared error. We also derived an explicit formula for an upper bound which proved to be more useful in practice for deriving orders of convergence. Through experiments with simulated data, we tested our inference procedure and found that promising results were achieved under strong heteroskedasticity. As expected, the outlier analysis also revealed a small degree of sensitivity to the realisation of the multiplicative Gaussian noise in our VMMA.
\\
The application to sea surface temperature anomaly data in Section \ref{sec:emp} illustrates the benefits of using a VMMA instead of a GMA when spatial heteroskedasticity is present. Better prediction for missing values can be achieved through the estimated trend and the estimated local variances. This empirical example also shines light on the many ways one can adapt and improve the two-step moments-matching estimation method. Using least-squares in the fitting of the normalised variograms is one such extension. To further improve the variogram fits, generalising the method to the popular Mat\'ern kernel as well as anisotropic kernels will be helpful. One may also consider experimenting with different shapes of the moving window to capture that of the anisotropic volatility clusters. Incorporating kernel averaging in the moving window approach, such as that done for the local variogram estimation in \cite{FDR2016}, could also help to reduce the occurrence of overlapping clusters.
\\
So far, only point estimates of our parameters are considered and the uncertainty in these is not reflected in our prediction intervals. A key direction for further work is building a Bayesian framework around our moments-matching method to provide credible intervals. In this case, the moments-matching estimates could be useful as starting values or to motivate prior distributions. Composite likelihoods could also be used to reduce computational effort.

\section*{Appendix:}

\begin{proof}[Proof of Example \ref{eg:2oCov}]

For Model (\ref{eqn:weg}), we have $\Cov\left(\sigma^{2}(\bm{\xi}), \sigma^{2}(\bm{\xi}^{*})\right) =\frac{b\eta}{2\pi}\exp\left(-\frac{\eta\left(\bm{\xi}-\bm{\xi}^{*}\right)^{T}\left(\bm{\xi}-\bm{\xi}^{*}\right)}{2}\right)$, where $\Var\left(L'\right) = b$. So, with $\mathbf{w} = \mathbf{x}-\bm{\xi}$, $\mathbf{w}^{*} = \mathbf{x}^{*}-\bm{\xi}^{*}$ and $\mathbf{u} = \mathbf{x} - \mathbf{x}^{*}$:
\begin{align*}
\Cov(\sigma^{2}_{I}(\mathbf{x}), \sigma^{2}_{I}(\mathbf{x}^{*})) &= \int_{\mathbb{R}^{2}} \int_{\mathbb{R}^{2}} g^{2}(\mathbf{x}-\bm{\xi})g^{2}(\mathbf{x}^{*}-\bm{\xi}^{*})\Cov\left(\sigma^{2}(\bm{\xi}), \sigma^{2}(\bm{\xi}^{*})\right)\mathrm{d}\bm{\xi}\mathrm{d}\mathbf{s^{*}} \\
&= \frac{b\lambda^{4}\eta}{2\pi^{5}}\int_{\mathbb{R}^{2}} \int_{\mathbb{R}^{2}}\exp\left(-2\lambda\left[\mathbf{w}^{T}\mathbf{w}+\mathbf{w}^{*T}\mathbf{w}^{*}\right]-\frac{\eta\left(\mathbf{w}^{*}-\mathbf{w} + \mathbf{u}\right)^{T}\left(\mathbf{w}^{*}-\mathbf{w} + \mathbf{u}\right)}{2}\right) \mathrm{d}\mathbf{w}\mathrm{d}\mathbf{w}^{*}. 
\end{align*}
The exponent of the integrand is equal to:
\begin{equation*}
 -2\lambda(w_{1}^{2} + w_{1}^{*2}) - \frac{\eta}{2}\left[\left(w_{1}^{*} - w_{1}\right)^{2} + 2u_{1}\left(w_{1}^{*} - w_{1}\right)\right] -2\lambda(w_{2}^{2} + w_{2}^{*2}) - \frac{\eta}{2}\left[\left(w_{2}^{*} - w_{2}\right)^{2} + 2u_{2}\left(w_{2}^{*} - w_{2}\right)\right] -\frac{\eta}{2}\left(u_{1}^{2}+ u_{2}^{2}\right).
\end{equation*}
Focusing on the terms in $w_{1}$ and $w_{1}^{*}$:
\begin{align*}
&\int_{\mathbb{R}} \int_{\mathbb{R}}\exp\left(-2\lambda(w_{1}^{2} + w_{1}^{*2}) - \frac{\eta}{2}\left[\left(w_{1}^{*} - w_{1}\right)^{2} + 2u_{1}\left(w_{1}^{*} - w_{1}\right)\right]\right) \mathrm{d}w_{1}\mathrm{d}w_{1}^{*} \\
&= \int_{\mathbb{R}} \exp\left(-\left[2\lambda + \frac{\eta}{2}\right] w_{1}^{*2} -\eta u_{1} w_{1}^{*} + \frac{\left(\eta \left[u_{1} + w_{1}^{*}\right]\right)^{2}}{4\left(2\lambda + \frac{\eta}{2}\right)}\right) \int_{\mathbb{R}}\exp\left(-\left[2\lambda + \frac{\eta}{2}\right]\left[w_{1} -\frac{\eta \left[u_{1} + w_{1}^{*}\right]}{2\left(2\lambda + \frac{\eta}{2}\right)}\right]^{2}\right) \mathrm{d}w_{1} \mathrm{d}w_{1}^{*}\\
&= \sqrt{\frac{\pi}{2\lambda + \frac{\eta}{2}}} \int_{\mathbb{R}} \exp\left(-\left[2\lambda + \frac{\eta}{2}\right] w_{1}^{*2} -\eta u_{1} w_{1}^{*} + \frac{\left(\eta \left[u_{1} + w_{1}^{*}\right]\right)^{2}}{4\left(2\lambda + \frac{\eta}{2}\right)}\right) \mathrm{d}w_{1}^{*} \\
&= \sqrt{\frac{\pi}{2\lambda + \frac{\eta}{2}}} \exp\left(\frac{\eta^{2}u_{1}^{2}}{4\left(2\lambda + \frac{\eta}{2}\right)} + \frac{\lambda\left(4\lambda +2\eta\right)}{2\lambda + \frac{\eta}{2}} \left[\frac{\eta u_{1}}{4\lambda+2\eta}\right]^{2}\right) \int_{\mathbb{R}} \exp\left(-\frac{\lambda\left(4\lambda + 2\eta\right)}{2\lambda + \frac{\eta}{2}} \left[w_{1}^{*} - \frac{\eta u_{1}}{4\lambda+2\eta}\right]^{2}\right) \mathrm{d}w_{1}^{*} \\
&= \frac{\pi}{\sqrt{\lambda\left(4\lambda + 2\eta\right)}} \exp\left(\frac{\eta^{2}u_{1}^{2}}{4\lambda + 2\eta}\right).
\end{align*}
As the terms in $w_{2}$ and $w_{2}^{*}$ follow a similar structure, we have:
\begin{align}
\Cov(\sigma^{2}_{I}(\mathbf{x}), \sigma^{2}_{I}(\mathbf{x}^{*})) &= \frac{b\lambda^{3}\eta}{4\pi^{3}(2\lambda + \eta)}\exp\left(\frac{-\lambda\eta}{2\lambda + \eta} \left(\mathbf{x} - \mathbf{x}^{*}\right)^{T}\left(\mathbf{x} - \mathbf{x}^{*}\right)\right). \label{eqn:VCov}
\end{align}
On the other hand:
\begin{align*}
\mathbb{E}\left[\left(\int_{\mathbb{R}^{2}} g(\mathbf{x}-\bm{\xi})g(\mathbf{x}^{*}-\bm{\xi})\sigma^{2}(\bm{\xi})\mathrm{d}\bm{\xi}\right)^{2}\right] &=\int_{\mathbb{R}^{2}}\int_{\mathbb{R}^{2}} g(\mathbf{x}-\bm{\xi})g(\mathbf{x}^{*}-\bm{\xi})g(\mathbf{x}-\bm{\xi}^{*})g(\mathbf{x}^{*}-\bm{\xi}^{*})\mathbb{E}\left[\sigma^{2}(\bm{\xi})\sigma^{2}(\bm{\xi}^{*})\right]\mathrm{d}\bm{\xi}\mathrm{d}\bm{\xi}^{*} \\
&= \frac{\lambda^{4}}{\pi^{4}} \left(\frac{b\eta A'}{2\pi} + a^{2}B'\right), 
\end{align*}
where $A' = \int_{\mathbb{R}^{2}}\int_{\mathbb{R}^{2}} e^{-\lambda\left[\mathbf{w}^{T}\mathbf{w} + \mathbf{w}^{*T}\mathbf{w*} + \left(\mathbf{w}^{*} + \mathbf{u}\right)^{T}\left(\mathbf{w}^{*} + \mathbf{u}\right) + \left(\mathbf{w} - \mathbf{u}\right)^{T}\left(\mathbf{w} - \mathbf{u}\right)\right]-\frac{\eta}{2}\left(\mathbf{w}^{*} - \mathbf{w} + \mathbf{u}\right)^{T}\left(\mathbf{w}^{*} - \mathbf{w} + \mathbf{u}\right)}\mathrm{d}\mathbf{w}\mathrm{d}\mathbf{w}^{*}$, and \\
$B' = \int_{\mathbb{R}^{2}}\int_{\mathbb{R}^{2}} e^{-\lambda\left[\mathbf{w}^{T}\mathbf{w} + \mathbf{w}^{*T}\mathbf{w*} + \left(\mathbf{w}^{*} + \mathbf{u}\right)^{T}\left(\mathbf{w}^{*} + \mathbf{u}\right) + \left(\mathbf{w} - \mathbf{u}\right)^{T}\left(\mathbf{w} - \mathbf{u}\right)\right]}\mathrm{d}\mathbf{w}\mathrm{d}\mathbf{w}^{*}$.
\\
The exponent of the $A'$'s integrand is equal to:
\begin{align*}
&-2\lambda\left[w_{1}^{2} + w_{1}^{*2} + u_{1}w_{1}^{*} - u_{1}w_{1}\right] - \frac{\eta}{2}\left[\left(w_{1}^{*} - w_{1}\right)^{2} + 2 u_{1}\left(w_{1}^{*} - w_{1}\right)\right] -\left[2\lambda + \frac{\eta}{2}\right]u_{1}^{2}\\
& -2\lambda\left[w_{2}^{2} + w_{2}^{*2} + u_{2}w_{2}^{*} - u_{2}w_{2}\right] - \frac{\eta}{2}\left[\left(w_{2}^{*} - w_{2}\right)^{2} + 2 u_{2}\left(w_{2}^{*} - w_{2}\right)\right] -\left[2\lambda + \frac{\eta}{2}\right]u_{2}^{2}.
\end{align*}
Focusing on the terms in $w_{1}$ and $w_{1}^{*}$:
\begin{align*}
&\int_{\mathbb{R}}\int_{\mathbb{R}} \exp\left(-2\lambda\left[w_{1}^{2} + w_{1}^{*2} + u_{1}w_{1}^{*} - u_{1}w_{1}\right] - \frac{\eta}{2}\left[\left(w_{1}^{*} - w_{1}\right)^{2} + 2 u_{1}\left(w_{1}^{*} - w_{1}\right)\right]\right) \mathrm{d}w_{1}\mathrm{d}w_{1}^{*} \\
&= \int_{\mathbb{R}}\exp\left(-\left[2\lambda + \frac{\eta}{2}\right] w_{1}^{*2} -\left[2\lambda + \eta\right] u_{1}w_{1}^{*}\right)\int_{\mathbb{R}} \exp\left(-\left[2\lambda + \frac{\eta}{2}\right] w_{1}^{2}+\left[\left(2\lambda + \eta\right)u_{1} + \eta w_{1}^{*}\right]w_{1} \right) \mathrm{d}w_{1}\mathrm{d}w_{1}^{*} \\
&= \sqrt{\frac{\pi}{2\lambda + \frac{\eta}{2}}}\exp\left(\frac{\left(2\lambda + \eta\right)u_{1}^{2}}{2}\right)\int_{\mathbb{R}}\exp\left(- \frac{4\lambda^{2} + 2\lambda\eta}{2\lambda + \frac{\eta}{2}} \left[w_{1} - \frac{ u_{1}}{2}\right]^{2}\right) \mathrm{d}w_{1}^{*} \\
&= \frac{\pi}{\sqrt{4\lambda^{2}+ 2\lambda\eta}}\exp\left(\frac{\left(2\lambda + \eta\right)u_{1}^{2}}{2}\right).
\end{align*}
Since the terms in $w_{2}$ and $w_{2}^{*}$ follow the same form:
\begin{equation*}
A' = \frac{\pi^{2}}{4\lambda^{2}+ 2\lambda\eta}\exp\left(\left[\frac{\left(2\lambda + \eta\right)}{2}-\left(2\lambda + \frac{\eta}{2}\right)\right]\left(\mathbf{x}-\mathbf{x}^{*}\right)^{T}\left(\mathbf{x}-\mathbf{x}^{*}\right)\right) =\frac{\pi^{2}}{4\lambda^{2}+ 2\lambda\eta}\exp\left(-\lambda\left(\mathbf{x}-\mathbf{x}^{*}\right)^{T}\left(\mathbf{x}-\mathbf{x}^{*}\right)\right) .
\end{equation*}
The exponent of the $B'$'s integrand is equal to:
\begin{equation*}
-2\lambda\left[\left(w_{1}^{2} - u_{1}w_{1}\right) + \left(w_{1}^{*2} + u_{1}w_{1}^{*}\right) +\left(w_{2}^{2} - u_{2}w_{2}\right) + \left(w_{2}^{*2} + u_{2}w_{2}^{*}\right) + \left[u_{1}^{2}+ u_{2}^{2}\right]\right].
\end{equation*}
Focusing on the terms in $w_{1}$:
\begin{equation*}
\int_{\mathbb{R}} \exp\left(-2\lambda\left(w_{1}^{2} - u_{1}w_{1}\right) \right) \mathrm{d}w_{1} = \exp\left(\frac{\lambda u_{1}^{2}}{2}\right)\int_{\mathbb{R}} \exp\left(-2\lambda\left(w_{1} - \frac{u_{1}}{2}\right)^{2} \right) \mathrm{d}w_{1} = \sqrt{\frac{\pi}{2\lambda}}\exp\left(\frac{\lambda u_{1}^{2}}{2}\right).
\end{equation*}
Similarly, $\int_{\mathbb{R}} \exp\left(-2\lambda\left(w_{1}^{*2} + u_{1}w_{1}^{*}\right)\right) = \sqrt{\frac{\pi}{2\lambda}}\exp\left(\frac{\lambda u_{1}^{2}}{2}\right)$. Thus, $B' = \frac{\pi^{2}}{4\lambda^{2}}\exp\left(-\lambda \left(\mathbf{x} - \mathbf{x}^{*}\right)^{T}\left(\mathbf{x} - \mathbf{x}^{*}\right)\right)$.
\\
Since $\Cov(Y^{2}(\mathbf{x}), Y^{2}(\mathbf{x}^{*})) = \frac{2\lambda^{4}}{\pi^{4}} \left(\frac{b\eta A'}{2\pi} + a^{2}B'\right) + \Cov(\sigma^{2}_{I}(\mathbf{x}), \sigma^{2}_{I}(\mathbf{x}^{*}))$, we obtain the required result. 
\end{proof}

\begin{proof}[Proof of Theorem \ref{thm:mse}]
\begin{align*}
\mathbb{E}\left[|Y(\mathbf{x}) - Z(\mathbf{x})|^{2}\right] &= \mathbb{E}\left[\left|\int_{\mathbb{R}^{2}}g(\mathbf{x}-\bm{\xi})\sigma(\bm{\xi})W(\mathrm{d}\bm{\xi}) - \int_{\mathbb{R}^{2}}g_{\triangle}(\mathbf{x}, \bm{\xi})\sigma_{\triangle}(\bm{\xi})W(\mathrm{d}\bm{\xi})\right|^{2}\right] \\
&= \mathbb{E}\left[\left|\int_{\mathbb{R}^{2}}g(\mathbf{x}-\bm{\xi})\left(\sigma(\bm{\xi}) -  \sigma_{\triangle}(\bm{\xi})\right)W(\mathrm{d}\bm{\xi}) + \int_{\mathbb{R}^{2}}\left(g(\mathbf{x} - \bm{\xi})- g_{\triangle}(\mathbf{x}, \bm{\xi})\right)\sigma_{\triangle}(\bm{\xi}) W(\mathrm{d}\bm{\xi})\right|^{2}\right] \\
&= T1 + T2 + T3,
\end{align*}
where $T1 := \int_{\mathbb{R}^{2}}g^{2}(\mathbf{x}-\bm{\xi})\mathbb{E}\left[\left(\sigma(\bm{\xi}) -  \sigma_{\triangle}(\bm{\xi})\right)^{2}\right]\mathrm{d}\bm{\xi}$, $T2 :=  \int_{\mathbb{R}^{2}}\left(g(\mathbf{x} - \bm{\xi})- g_{\triangle}(\mathbf{x}, \bm{\xi})\right)^{2}\mathbb{E}\left[\sigma^{2}_{\triangle}(\bm{\xi})\right] \mathrm{d}\bm{\xi}$ and $T3:= 2\int_{\mathbb{R}^{2}}g(\mathbf{x}-\bm{\xi})\left(g(\mathbf{x} - \bm{\xi})- g_{\triangle}(\mathbf{x}, \bm{\xi})\right)\mathbb{E}\left[\sigma_{\triangle}(\bm{\xi})\left(\sigma(\bm{\xi}) -  \sigma_{\triangle}(\bm{\xi})\right)\right]\mathrm{d}\bm{\xi}$. In the calculations, the third equality follows from the independence of $\sigma^{2}$ and $W$ as well as the fact that $W$ is a homogeneous standard Gaussian basis. We simplify the three terms separately:
\begin{align*}
T1 &= \int_{\mathbb{R}^{2}}g^{2}(\mathbf{x}-\bm{\xi})\mathbb{E}\left[\sigma^{2}(\bm{\xi})\right]\mathrm{d}\bm{\xi}  +  \int_{\mathbb{R}^{2}}g^{2}(\mathbf{x}-\bm{\xi})\mathbb{E}\left[\sigma^{2}_{\triangle}(\bm{\xi})\right]\mathrm{d}\bm{\xi} - 2\int_{\mathbb{R}^{2}}g^{2}(\mathbf{x}-\bm{\xi})\mathbb{E}\left[\sigma(\bm{\xi})\sigma_{\triangle}(\bm{\xi})\right]\mathrm{d}\bm{\xi} \\
&=a\left(\int_{\mathbb{R}^{2}}g^{2}(\mathbf{x}-\bm{\xi})\mathrm{d}\bm{\xi}\right)\left(\int_{\mathbb{R}^{2}}h(\bm{\xi} - \mathbf{u})\mathrm{d}\mathbf{u} + \sum_{i, j = -\tilde{p}}^{\tilde{p}} h\left(i\triangle, j \triangle\right) \triangle^{2}\right) - 2\int_{\mathbb{R}^{2}}g^{2}(\mathbf{x}-\bm{\xi})\mathbb{E}\left[\sigma(\bm{\xi})\sigma_{\triangle}(\bm{\xi})\right]\mathrm{d}\bm{\xi}
\end{align*}
since $\mathbb{E}\left[\sigma^{2}(\bm{\xi})\right] =  \int_{\mathbb{R}^{2}}h(\bm{\xi} - \mathbf{u})\mathbb{E}\left[L(\mathrm{d}\mathbf{u}\right] = a\int_{\mathbb{R}^{2}}h(\bm{\xi} - \mathbf{u})\mathrm{d}\mathbf{u}$, and:
\begin{equation*}
\mathbb{E}\left[\sigma_{\triangle}^{2}(\bm{\xi})\right]= \int_{\mathbb{R}^{2}}h_{\triangle}(\bm{\xi}, \mathbf{u}) \mathbb{E}\left[L(\mathrm{d}\mathbf{u}) \right] =  a \sum_{i', j' = -\tilde{p}}^{\tilde{p}} h\left(i'\triangle, j' \triangle\right) \triangle^{2}.
\end{equation*}
Next, we focus on $T2$: 
\begin{align*}
T2 &= \left(a \sum_{i', j' = -\tilde{p}}^{\tilde{p}} h\left(i'\triangle, j' \triangle\right) \triangle^{2}\right) \left[\int_{\mathbb{R}^{2}}g^{2}(\mathbf{x} - \bm{\xi})\mathrm{d}\bm{\xi}  + \int_{\mathbb{R}^{2}}g_{\triangle}^{2}(\mathbf{x}, \bm{\xi}) \mathrm{d}\bm{\xi} - 2\int_{\mathbb{R}^{2}}g(\mathbf{x} - \bm{\xi})g_{\triangle}(\mathbf{x}, \bm{\xi}) \mathrm{d}\bm{\xi} \right]\\
&= \left[\left(\int_{\mathbb{R}^{2}}g^{2}(\mathbf{x} - \bm{\xi})\mathrm{d}\bm{\xi} -  \sum_{i,j = -p}^{p} g^{2}\left(i\triangle, j \triangle\right)\triangle^{2}\right) \right. \\
&\left.+ 2\sum_{i,j = -p}^{p} g\left(i\triangle, j \triangle\right)\left( g\left(i\triangle, j \triangle\right)\triangle^{2} - \int_{i\triangle- \frac{\triangle}{2}}^{i\triangle + \frac{\triangle}{2}} \int_{j\triangle- \frac{\triangle}{2}}^{j\triangle + \frac{\triangle}{2}} g(\mathbf{w})  \mathrm{d}\mathbf{w}\right)\right]\times \left(a \sum_{i', j' = -\tilde{p}}^{\tilde{p}} h\left(i'\triangle, j' \triangle\right) \triangle^{2}\right) ,
\end{align*}
by letting $ \mathbf{w} = \bm{\xi} - \mathbf{x}$ and since:
\begin{align*}
\int_{\mathbb{R}^{2}} g_{\triangle}^{2}(\mathbf{x}, \bm{\xi})\mathrm{d}\bm{\xi} &= \sum_{i,j = -p}^{p}\sum_{i',j' = -p}^{p}\int_{\mathbb{R}^{2}}\mathbf{1}_{\left[x_{1} + i\triangle- \frac{\triangle}{2}, x_{1} + i\triangle + \frac{\triangle}{2}\right)}(s_{1}) \mathbf{1}_{\left[x_{2} + j\triangle- \frac{\triangle}{2}, x_{2} + j\triangle + \frac{\triangle}{2}\right)}(s_{2}) \\
&\mathbf{1}_{\left[x_{1} + i'\triangle- \frac{\triangle}{2}, x_{1} + i'\triangle + \frac{\triangle}{2}\right)}(s_{1})\mathbf{1}_{\left[x_{2} + j'\triangle- \frac{\triangle}{2}, x_{2} + j'\triangle + \frac{\triangle}{2}\right)}(s_{2}) g\left(i\triangle, j \triangle\right)g\left(i'\triangle, j' \triangle\right)\mathrm{d}\bm{\xi} \\ 
&= \sum_{i,j = -p}^{p} g^{2}\left(i\triangle, j \triangle\right)\int_{\mathbb{R}^{2}}\mathbf{1}_{\left[x_{1} + i\triangle- \frac{\triangle}{2}, x_{1} + i\triangle + \frac{\triangle}{2}\right)}(s_{1}) \mathbf{1}_{\left[x_{2} + j\triangle- \frac{\triangle}{2}, x_{2} + j\triangle + \frac{\triangle}{2}\right)}(s_{2})\mathrm{d}\bm{\xi}\\ 
&= \sum_{i,j = -p}^{p} g^{2}\left(i\triangle, j \triangle\right)\triangle^{2}, \\
\text{and } \int_{\mathbb{R}^{2}}g(\mathbf{x} - \bm{\xi})g_{\triangle}(\mathbf{x}, \bm{\xi}) \mathrm{d}\bm{\xi} &= \int_{\mathbb{R}^{2}}g(\mathbf{x} - \bm{\xi})\sum_{i,j = -p}^{p} \mathbf{1}_{\left[x_{1} + i\triangle- \frac{\triangle}{2}, x_{1} + i\triangle + \frac{\triangle}{2}\right)}(s_{1}) \mathbf{1}_{\left[x_{2} + j\triangle- \frac{\triangle}{2}, x_{2} + j\triangle + \frac{\triangle}{2}\right)}(s_{2}) \\
&\times g\left(i\triangle, j \triangle\right) \mathrm{d}\bm{\xi} \\
&=  \sum_{i,j = -p}^{p} g\left(i\triangle, j \triangle\right) \int_{i\triangle- \frac{\triangle}{2}}^{i\triangle + \frac{\triangle}{2}} \int_{j\triangle- \frac{\triangle}{2}}^{j\triangle + \frac{\triangle}{2}} g(\mathbf{w})  \mathrm{d}w_{1}\mathrm{d}w_{2}. 
\end{align*}
Finally, we look at $T3$:
\begin{align*}
T3 &= 2\int_{\mathbb{R}^{2}}\left[g^{2}(\mathbf{x}-\bm{\xi}) - g(\mathbf{x} - \bm{\xi})g_{\triangle}(\mathbf{x}, \bm{\xi})\right]\left(\mathbb{E}\left[\sigma(\bm{\xi})\sigma_{\triangle}(\bm{\xi})\right] -  \mathbb{E}\left[ \sigma^{2}_{\triangle}(\bm{\xi})\right]\right)\mathrm{d}\bm{\xi} \\
&= 2\left[\int_{\mathbb{R}^{2}}g^{2}(\mathbf{x}-\bm{\xi})\mathbb{E}\left[\sigma(\bm{\xi})\sigma_{\triangle}(\bm{\xi})\right] \mathrm{d}\bm{\xi} - \int_{\mathbb{R}^{2}}g(\mathbf{x} - \bm{\xi})g_{\triangle}(\mathbf{x}, \bm{\xi})\mathbb{E}\left[\sigma(\bm{\xi})\sigma_{\triangle}(\bm{\xi})\right] \mathrm{d}\bm{\xi} \right. \\ &\left. -  a \sum_{i', j'= -\tilde{p}}^{\tilde{p}} h\left(i'\triangle, j'\triangle\right) \triangle^{2}\left(\int_{\mathbb{R}^{2}}g^{2}(\mathbf{x}-\bm{\xi})\mathrm{d}\bm{\xi} - \sum_{i, j = -p}^{p}g\left(i\triangle, j \triangle\right) \int_{i\triangle- \frac{\triangle}{2}}^{i\triangle + \frac{\triangle}{2}} \int_{j\triangle- \frac{\triangle}{2}}^{j\triangle + \frac{\triangle}{2}} g(\mathbf{w})  \mathrm{d}\mathbf{w}\right)\right]. 
\end{align*}
To obtain the required expression for $\mathbb{E}\left[\sigma(\bm{\xi})\sigma_{\triangle}(\bm{\xi})\right]$, we apply the following equality from Section 1.7 of \cite{Applebaum2009}:
$u^{\alpha} = \frac{\alpha}{\Gamma(1-\alpha)} \int_{0}^{\infty} (1 - e^{-ux})\frac{\mathrm{d} x}{x^{1+\alpha}}$, where $u\geq 0$ and $0<\alpha<1$. By setting $u = \sigma(\bm{\xi})$ and $u = \sigma_{\triangle}(\bm{\xi})$ separately with $\alpha = 1/2$, and using Fubini's Theorem:
\begin{align*}
\mathbb{E}\left[\sigma(\bm{\xi})\sigma_{\triangle}(\bm{\xi})\right] &= \mathbb{E}\left[\frac{1}{4\pi} \int_{0}^{\infty} \int_{0}^{\infty} \left(1 - e^{-\sigma^{2}\left(\bm{\xi}\right)x}\right) \left(1 - e^{-\sigma_{\triangle}^{2}\left(\bm{\xi}\right)y}\right)\frac{\mathrm{d}x}{x^{3/2}}\frac{\mathrm{d}y}{y^{3/2}}\right] \\ 
&=\frac{1}{4\pi}  \int_{0}^{\infty} \int_{0}^{\infty} \left(1 - \mathbb{E}\left[e^{-\sigma^{2}\left(\bm{\xi}\right)x}\right] - \mathbb{E}\left[ e^{-\sigma_{\triangle}^{2}\left(\bm{\xi}\right)y}\right] + \mathbb{E}\left[ e^{-\sigma^{2}\left(\bm{\xi}\right)x-\sigma_{\triangle}^{2}\left(\bm{\xi}\right)y}\right]  \right)\frac{\mathrm{d}x}{x^{3/2}}\frac{\mathrm{d}y}{y^{3/2}}. 
\end{align*}
Since $\sigma^{2}(\bm{\xi}) =\int_{\mathbb{R}^{2}} h(\bm{\xi} - \mathbf{u})L(\mathrm{d}\mathbf{u})$, its CGF can be expressed as: $C(\theta ; \sigma^{2}(\bm{\xi})) = \int_{\mathbb{R}^{2}}C(\theta h(\bm{\xi} - \mathbf{u}) ; L') \mathrm{d}\mathbf{u}$.   \\
By replacing $\theta$ by $ix$, we find that the Laplace exponent of $\sigma^{2}(\bm{\xi})$ is equal to $\int_{\mathbb{R}^{2}}\Psi_{L}(xh(\bm{\xi} - \mathbf{u}))d\mathbf{u}$. Thus, $\mathbb{E}\left[e^{-\sigma^{2}(\bm{\xi})x}\right] = e^{\int_{\mathbb{R}^{2}}\Psi_{L}(xh(\bm{\xi} - \mathbf{u}))d\mathbf{u}}$. The expressions for $\mathbb{E}\left[e^{-\sigma_{\triangle}^{2}(\bm{\xi})y}\right]$ and $\mathbb{E}\left[e^{-\sigma^{2}(\bm{\xi})x -\sigma_{\triangle}^{2}(\bm{\xi})y}\right]$ can be found analogously.
\end{proof}
\begin{proof}[Proof of Corollary \ref{cor:mseub}]
We obtain an upper bound for $T1$ in Theorem \ref{thm:mse}. By using the fact that the harmonic mean of $\sigma^{2}$ and $\sigma_{\triangle}^{2}$ is less than or equal to their geometric mean, and by applying Jensen’s inequality since $1/x$ is a convex function of $x$ for $x > 0$:
\begin{align*}
\mathbb{E}\left[\sigma(\bm{\xi})\sigma_{\triangle}(\bm{\xi})\right] &\geq \mathbb{E}\left[2\left(\frac{1}{\sigma^{2}(\bm{\xi})} + \frac{1}{\sigma_{\triangle}^{2}(\bm{\xi})}\right)^{-1}\right] \geq 2\left( \mathbb{E}\left[\frac{1}{\sigma^{2}(\bm{\xi})}\right] +  \mathbb{E}\left[\frac{1}{\sigma_{\triangle}^{2}(\bm{\xi})}\right]\right)^{-1} \geq 2 \left(\frac{1}{\mathbb{E}\left[\sigma^{2}(\bm{\xi})\right]} + \frac{1}{\mathbb{E}\left[\sigma_{\triangle}^{2}(\bm{\xi})\right]}\right)^{-1} \\ 
\Rightarrow  T1 &\leq  \left( a\int_{\mathbb{R}^{2}}g^{2}(\mathbf{x}-\bm{\xi})\mathrm{d}\bm{\xi}\right)\left[\int_{\mathbb{R}^{2}}h(\bm{\xi} - \mathbf{u})\mathrm{d}\mathbf{u} + \sum_{i, j = -\tilde{p}}^{\tilde{p}} h\left(i\triangle, j \triangle\right) \triangle^{2} \right. \\
&\left.- 4\left(\frac{1}{\int_{\mathbb{R}^{2}}h(\bm{\xi} - \mathbf{u})\mathrm{d}\mathbf{u}} + \frac{1}{ \sum_{i = -\tilde{p}}^{\tilde{p}} \sum_{j = -\tilde{p}}^{\tilde{p}} h\left(i\triangle, j \triangle\right) \triangle^{2}}\right)^{-1} \mathrm{d}\bm{\xi} \right]  \\
&= \left( a\int_{\mathbb{R}^{2}}g^{2}(\mathbf{x}-\bm{\xi})\mathrm{d}\bm{\xi}\right) \frac{\left(\int_{\mathbb{R}^{2}}h(\bm{\xi} - \mathbf{u})\mathrm{d}\mathbf{u} -  \sum_{i = -\tilde{p}}^{\tilde{p}} \sum_{j = -\tilde{p}}^{\tilde{p}} h\left(i\triangle, j \triangle\right) \triangle^{2}\right)^{2}}{\int_{\mathbb{R}^{2}}h(\bm{\xi} - \mathbf{u})\mathrm{d}\mathbf{u} + \sum_{i = -\tilde{p}}^{\tilde{p}} \sum_{j = -\tilde{p}}^{\tilde{p}} h\left(i\triangle, j \triangle\right) \triangle^{2}} \\
&:= T4. 
\end{align*}
Likewise, we obtain an upper bound for $T3$ in Theorem \ref{thm:mse}. By using the lower bound for $\mathbb{E}\left[\sigma(\bm{\xi})\sigma_{\triangle}(\bm{\xi})\right]$ attained previously as well as the fact that the arithmetic mean of $\sigma^{2}$ and $\sigma_{\triangle}^{2}$ is greater than or equal to their geometric mean, we obtain:
\begin{align*}
T3 &\leq    2\left[\frac{a}{2}\int_{\mathbb{R}^{2}}g^{2}(\mathbf{x}-\bm{\xi})\left[\int_{\mathbb{R}^{2}}h(\bm{\xi} - \mathbf{u})\mathrm{d}\mathbf{u} +  \sum_{i, j = -\tilde{p}}^{\tilde{p}} h\left(i\triangle, j \triangle\right) \triangle^{2}\right] \mathrm{d}\bm{\xi}  \right. \\ 
&\left. - 2 a\sum_{i, j = -p}^{p} g\left(i\triangle, j \triangle\right) \int_{x_{1} + i\triangle- \frac{\triangle}{2}}^{x_{1} + i\triangle + \frac{\triangle}{2}} \int_{x_{2} + j\triangle- \frac{\triangle}{2}}^{x_{2} + j\triangle + \frac{\triangle}{2}} g(\mathbf{x} - \bm{\xi})\left(\frac{1}{\int_{\mathbb{R}^{2}}h(\bm{\xi} - \mathbf{u})\mathrm{d}\mathbf{u}} + \frac{1}{ \sum_{i', j' = -\tilde{p}}^{\tilde{p}} h\left(i'\triangle, j' \triangle\right) \triangle^{2}}\right)^{-1}  \mathrm{d}\bm{\xi} \right. \\ 
&\left.-  a \sum_{i', j'= -\tilde{p}}^{\tilde{p}} h\left(i'\triangle, j'\triangle\right) \triangle^{2}\left(\int_{\mathbb{R}^{2}}g^{2}(\mathbf{x}-\bm{\xi})\mathrm{d}\bm{\xi} -  \sum_{i, j = -p}^{p}g\left(i\triangle, j \triangle\right) \int_{i\triangle- \frac{\triangle}{2}}^{i\triangle + \frac{\triangle}{2}} \int_{j\triangle- \frac{\triangle}{2}}^{j\triangle + \frac{\triangle}{2}} g(\mathbf{w})  \mathrm{d}\mathbf{w}\right)\right] \\
&=a\left[\int_{\mathbb{R}^{2}}h(\bm{\xi} - \mathbf{u})\mathrm{d}\mathbf{u} -  \sum_{i, j = -\tilde{p}}^{\tilde{p}} h\left(i\triangle, j \triangle\right) \triangle^{2}\right]\left[\int_{\mathbb{R}^{2}}g^{2}(\mathbf{x}-\bm{\xi})\mathrm{d}\bm{\xi}   - \frac{2\sum_{i', j'= -\tilde{p}}^{\tilde{p}} h\left(i'\triangle, j'\triangle\right) \triangle^{2}}{\sum_{i', j'= -\tilde{p}}^{\tilde{p}} h\left(i'\triangle, j'\triangle\right) \triangle^{2} + \int_{\mathbb{R}^{2}}h(\bm{\xi} - \mathbf{u})\mathrm{d}\mathbf{u}}\right. \\
&\left. \left(\sum_{i, j = -p}^{p}g\left(i\triangle, j \triangle\right) \int_{i\triangle- \frac{\triangle}{2}}^{i\triangle + \frac{\triangle}{2}} \int_{j\triangle- \frac{\triangle}{2}}^{j\triangle + \frac{\triangle}{2}} g(\mathbf{w})  \mathrm{d}\mathbf{w}\right)\right] \\
&=a\left[\int_{\mathbb{R}^{2}}h(\bm{\xi} - \mathbf{u})\mathrm{d}\mathbf{u} -  \sum_{i, j = -\tilde{p}}^{\tilde{p}} h\left(i\triangle, j \triangle\right) \triangle^{2}\right]\left[\left(\int_{\mathbb{R}^{2}}g^{2}(\mathbf{x}-\bm{\xi})\mathrm{d}\bm{\xi} - \sum_{i, j = -p}^{p}g^{2}\left(i\triangle, j \triangle\right) \triangle^{2} \right)\right.\\
&\left. + \sum_{i, j = -p}^{p}g\left(i\triangle, j \triangle\right) \left( g\left(i\triangle, j \triangle\right)\triangle^{2} - \int_{i\triangle- \frac{\triangle}{2}}^{i\triangle + \frac{\triangle}{2}} \int_{j\triangle- \frac{\triangle}{2}}^{j\triangle + \frac{\triangle}{2}} g(\mathbf{w})  \mathrm{d}\mathbf{w}\right) \right.\\
&\left.+ \left(\int_{\mathbb{R}^{2}}h(\bm{\xi} - \mathbf{u})\mathrm{d}\mathbf{u} - \sum_{i', j'= -\tilde{p}}^{\tilde{p}} h\left(i'\triangle, j'\triangle\right) \triangle^{2} \right) \left(\sum_{i, j = -p}^{p}g\left(i\triangle, j \triangle\right) \int_{i\triangle- \frac{\triangle}{2}}^{i\triangle + \frac{\triangle}{2}} \int_{j\triangle- \frac{\triangle}{2}}^{j\triangle + \frac{\triangle}{2}} g(\mathbf{w})  \mathrm{d}\mathbf{w}\right)\right] \\
&:= T5.
\end{align*}
By combining the upper bounds of $T1$ and $T3$ with $T2$, we obtain an upper bound for the simulation MSE. 
\end{proof}

\begin{proof}[Proof of Lemma \ref{lem:kcomp}]
Using Assumption \ref{a:bHm}, we can apply second-order Taylor expansions of $g^{2}(\mathbf{w})$ around $(i\triangle, j\triangle)$ for $-\lfloor R/\triangle \rfloor \leq i, j \leq \lfloor R/\triangle \rfloor$:
\begin{align*}
&\int_{-\left(R + \triangle/2\right)}^{R + \triangle/2}\int_{-\left(R + \triangle/2\right)}^{R + \triangle/2}g^{2}(\mathbf{w})\mathrm{d}\mathbf{w} -  \sum_{i,j = -\lfloor R/\triangle \rfloor}^{\lfloor R/\triangle \rfloor} g^{2}\left(i\triangle, j \triangle\right)\triangle^{2} \\
&=  \sum_{i,j = -\lfloor R/\triangle \rfloor}^{\lfloor R/\triangle \rfloor}\int_{i\triangle - \triangle/2}^{i\triangle + \triangle/2}\int_{j\triangle - \triangle/2}^{j\triangle + \triangle/2}\left(g^{2}(\mathbf{w}) -  g^{2}\left(i\triangle, j \triangle\right)\right)\mathrm{d}\mathbf{w} \\
&=  \frac{1}{2}\sum_{i,j = -\lfloor R/\triangle \rfloor}^{\lfloor R/\triangle \rfloor}\left(\int_{i\triangle - \triangle/2}^{i\triangle + \triangle/2}\int_{j\triangle - \triangle/2}^{j\triangle + \triangle/2} \left[\left(w_{1} - i\triangle\right)^{2}\left. \frac{\partial^{2}g^{2}(\mathbf{w})}{\partial^{2}w_{1}}\right|_{\mathbf{w} = \zeta_{1}(i, j, \triangle)} + \left(w_{2} - j\triangle\right)^{2} \left. \frac{\partial^{2}g^{2}(\mathbf{w})}{\partial^{2}w_{2}}\right|_{\mathbf{w} = \zeta_{2}(i, j, \triangle)} \right]\mathrm{d}\mathbf{w} \right),
\end{align*}
where $\zeta_{k}(i, j, \triangle) \in I(i, j, \triangle) = \left(i\triangle - \triangle/2, i\triangle + \triangle/2\right)\times \left(j\triangle - \triangle/2, j\triangle + \triangle/2\right)$ for $k = 1, 2$, and since the terms involving single powers of $\left(w_{1} - i\triangle\right)$ and $\left(w_{2} - j\triangle\right)$ integrate to zero.
\\
Let $s_{1}(i, j, \triangle)$ and $S_{1}(i, j, \triangle)$ be the infimum and supremum of $\frac{\partial^{2}g^{2}(\mathbf{w})}{\partial^{2}w_{1}}$ over $I(i, j, \triangle)$. Then: 
\begin{align*}
\frac{\triangle^{4}}{12}\sum_{i,j = -\lfloor R/\triangle \rfloor}^{\lfloor R/\triangle \rfloor}s_{1}(i, j, \triangle) &\leq \sum_{i,j = -\lfloor R/\triangle \rfloor}^{\lfloor R/\triangle \rfloor}\int_{i\triangle - \triangle/2}^{i\triangle + \triangle/2}\int_{j\triangle - \triangle/2}^{j\triangle + \triangle/2} \left(w_{1} - i\triangle\right)^{2}\left. \frac{\partial^{2}g^{2}(\mathbf{w})}{\partial^{2}w_{1}}\right|_{\mathbf{w} = \zeta_{1}(i, j, \triangle)} \mathrm{d}\mathbf{w} \\
&\leq \frac{\triangle^{4}}{12}\sum_{i,j = -\lfloor R/\triangle \rfloor}^{\lfloor R/\triangle \rfloor}S_{1}(i, j, \triangle),
\end{align*}
since $\int_{i\triangle - \triangle/2}^{i\triangle + \triangle/2}\int_{j\triangle - \triangle/2}^{j\triangle + \triangle/2} \left(w_{1} - i\triangle\right)^{2}\mathrm{d}\mathbf{w} = \triangle \left[\frac{\triangle^{3}}{12}\right] = \frac{\triangle^{4}}{12}$. 
\\
Since $\lim_{\triangle\rightarrow 0} \triangle^{2}\sum_{i,j = -\lfloor R/\triangle \rfloor}^{\lfloor R/\triangle \rfloor}s_{1}(i, j, \triangle) = \lim_{\triangle\rightarrow 0} \triangle^{2}\sum_{i,j = -\lfloor R/\triangle \rfloor}^{\lfloor R/\triangle \rfloor}S_{1}(i, j, \triangle) = \int_{-R}^{R} \int_{-R}^{R}  \frac{\partial^{2}g^{2}(\mathbf{w})}{\partial^{2}w_{1}}\mathrm{d}\mathbf{w} $:
\begin{equation*}
\lim_{\triangle\rightarrow 0} \frac{1}{\triangle^{2}}\sum_{i,j = -\lfloor R/\triangle \rfloor}^{\lfloor R/\triangle \rfloor}\int_{i\triangle - \triangle/2}^{i\triangle + \triangle/2}\int_{j\triangle - \triangle/2}^{j\triangle + \triangle/2} \left(w_{1} - i\triangle\right)^{2}\left. \frac{\partial^{2}g^{2}(\mathbf{w})}{\partial^{2}w_{1}}\right|_{\mathbf{w} = \zeta_{1}(i, j, \triangle)} \mathrm{d}\mathbf{w} = \frac{1}{12}\int_{-R}^{R} \int_{-R}^{R}  \frac{\partial^{2}g^{2}(\mathbf{w})}{\partial^{2}w_{1}}\mathrm{d}\mathbf{w},
\end{equation*}
by the Sandwich Theorem. Using similar arguments for $\sum_{i,j = -\lfloor R/\triangle \rfloor}^{\lfloor R/\triangle \rfloor}\int_{i\triangle - \triangle/2}^{i\triangle + \triangle/2}\int_{j\triangle - \triangle/2}^{j\triangle + \triangle/2} \left(w_{2} - j\triangle\right)^{2}\left. \frac{\partial^{2}g^{2}(\mathbf{w})}{\partial^{2}w_{2}}\right|_{\mathbf{w} = \zeta_{2}(i, j, \triangle)} \mathrm{d}\mathbf{w}$:
\begin{equation*}
\int_{-\left(R + \triangle/2\right)}^{R + \triangle/2}\int_{-\left(R + \triangle/2\right)}^{R + \triangle/2}g^{2}(\mathbf{w})\mathrm{d}\mathbf{w} -  \sum_{i,j = -\lfloor R/\triangle \rfloor}^{\lfloor R/\triangle \rfloor} g^{2}\left(i\triangle, j \triangle\right)\triangle^{2} = O(\triangle^{2}). 
\end{equation*}
By analogous arguments, one obtain the results for $g\left(i\triangle, j \triangle\right)\triangle^{2} - \int_{i\triangle- \frac{\triangle}{2}}^{i\triangle + \frac{\triangle}{2}} \int_{j\triangle- \frac{\triangle}{2}}^{j\triangle + \frac{\triangle}{2}} g(\mathbf{w})  \mathrm{d}\mathbf{w}$ and \\ $\int_{-\left(\widetilde{R} + \triangle/2\right)}^{\widetilde{R} + \triangle/2}\int_{-\left(\widetilde{R} + \triangle/2\right)}^{\widetilde{R} + \triangle/2}h(\mathbf{w})\mathrm{d}\mathbf{w} -  \sum_{i,j = -\widetilde{R}/\triangle}^{\widetilde{R}/\triangle}  h\left(i\triangle, j \triangle\right) \triangle^{2}$. 
\end{proof}

\begin{proof}[Proof of Theorem \ref{thm:l2con}]
With $\mathbb{E}[\cdot|\mathcal{F}^{\sigma}]$ denoting the expectation conditional on $\sigma^{2}$,  the MSE of our estimator is:
\\
\begin{align}
&\mathbb{E}\left[\left| \frac{1}{(2Q+1)^{2}} \sum_{l = -Q}^{Q} \sum_{k = -Q}^{Q} Y^{2}(\mathbf{x} + (l, k)\triangle)-\sigma_{I}^{2}(\mathbf{x}) \right|^{2}\right] \nonumber\\
&= \mathbb{E}\left[\left| \frac{1}{(2Q+1)^{2}} \sum_{l = -Q}^{Q} \sum_{k = -Q}^{Q} \left[Y^{2}(\mathbf{x} + (l, k)\triangle)-\sigma_{I}^{2}(\mathbf{x} + (l, k)\triangle) - \left(\sigma_{I}^{2}(\mathbf{x})-\sigma_{I}^{2}(\mathbf{x} + (l, k)\triangle) \right) \right]\right|^{2}\right] \nonumber\\
&=  \frac{1}{(2Q+1)^{4}} \mathbb{E}\left[\left| \sum_{l = -Q}^{Q} \sum_{k = -Q}^{Q} \left[Y^{2}(\mathbf{x} + (l, k)\triangle)-\sigma_{I}^{2}(\mathbf{x} + (l, k)\triangle) \right]\right|^{2}\right]  \label{eqn:MSEt1}\\
& - \frac{2}{(2Q+1)^{4}} \mathbb{E}\left[ \left(\sum_{l = -Q}^{Q} \sum_{k = -Q}^{Q} \left(Y^{2}(\mathbf{x} + (l, k)\triangle)-\sigma_{I}^{2}(\mathbf{x} + (l, k)\triangle)  \right)\right)\sum_{l = -Q}^{Q} \sum_{k = -Q}^{Q} \left[\sigma_{I}^{2}(\mathbf{x})-\sigma_{I}^{2}(\mathbf{x} + (l, k)\triangle) \right] \right] \label{eqn:MSEt2}\\
&+  \frac{1}{(2Q+1)^{4}} \mathbb{E}\left[ \left| \sum_{l = -Q}^{Q} \sum_{k = -Q}^{Q}  \left[\sigma_{I}^{2}(\mathbf{x})-\sigma_{I}^{2}(\mathbf{x} + (l, k)\triangle) \right]\right|^{2}\right]. \label{eqn:MSEt3}
\end{align}
When we simplify term (\ref{eqn:MSEt1}), we have:
\begin{align}
&\frac{1}{(2Q+1)^{4}} \mathbb{E}\left[ \sum_{\substack{l', k', l, k\\ = -Q}}^{Q} \left[Y^{2}(\mathbf{x} + (l, k)\triangle)-\sigma_{I}^{2}(\mathbf{x} + (l, k)\triangle) \right] \left[Y^{2}(\mathbf{x} + (l', k')\triangle)-\sigma_{I}^{2}(\mathbf{x} + (l', k')\triangle) \right]\right] \nonumber\\
&= \frac{1}{(2Q+1)^{4}}\sum_{\substack{l', k', l, k\\ = -Q}}^{Q} \mathbb{E}\left[\mathbb{E}\left[\left[ Y^{2}(\mathbf{x} + (l, k)\triangle)-\sigma_{I}^{2}(\mathbf{x} + (l, k)\triangle) \right] \left[Y^{2}(\mathbf{x} + (l', k')\triangle)-\sigma_{I}^{2}(\mathbf{x} + (l', k')\triangle)\right] |\mathcal{F}^{\sigma}\right]\right] \nonumber\\
&=  \frac{1}{(2Q+1)^{4}}\sum_{\substack{l', k', l, k\\ = -Q}}^{Q} \mathbb{E}\left[\Cov\left(Y^{2}(\mathbf{x}+ (l, k)\triangle), Y^{2}(\mathbf{x} + (l', k')\triangle) |\mathcal{F}^{\sigma}\right) \right] \label{eqn:CCcondition}\\
&= \frac{2\sum_{\substack{l', k', l, k\\ = -Q}}^{Q}  \mathbb{E}\left[\left(\int_{\mathbb{R}^{2}} g(\mathbf{x} + (l, k)\triangle-\bm{\xi})g(\mathbf{x} + (l', k')\triangle-\bm{\xi})\sigma^{2}(\bm{\xi})\mathrm{d}\bm{\xi}\right)^{2}\right]}{(2Q+1)^{4}} \rightarrow 0, \nonumber
\end{align}
by our assumption and Corollary 2.
\\
Since $\mathbb{E}\left[Y^{2}(\mathbf{x} + (l, k)\triangle) |\mathcal{F}^{\sigma}\right] = \sigma_{I}^{2}(\mathbf{x} + (l, k)\triangle)$ for $-Q \leq l, k \leq Q$, we also find that term (\ref{eqn:MSEt2}) is equal to:
\begin{equation*}
- \frac{2}{(2Q+1)^{4}} \mathbb{E}\left[\left(\sum_{l = -Q}^{Q} \sum_{k = -Q}^{Q} \left[\sigma_{I}^{2}(\mathbf{x})-\sigma_{I}^{2}(\mathbf{x} + (l, k)\triangle) \right]\right)\sum_{l = -Q}^{Q} \sum_{k = -Q}^{Q}\mathbb{E}\left[  \left(Y^{2}(\mathbf{x} + (l, k)\triangle)-\sigma_{I}^{2}(\mathbf{x} + (l, k)\triangle)  \right)|\mathcal{F}^{\sigma}\right] \right], 
\end{equation*}
which vanishes to zero. Simplifying term (\ref{eqn:MSEt3}), we get:
\begin{align}
&\frac{\sum_{\substack{l', k', l, k\\ = -Q}}^{Q}\left[\mathbb{E}\left[\sigma_{I}^{4}(\mathbf{x})\right] -  \mathbb{E}\left[ \sigma_{I}^{2}(\mathbf{x})\sigma_{I}^{2}(\mathbf{x} + (l, k)\triangle) \right]  -  \mathbb{E}\left[ \sigma_{I}^{2}(\mathbf{x})\sigma_{I}^{2}(\mathbf{x} + (l', k')\triangle) \right]+ \mathbb{E}\left[\sigma_{I}^{2}(\mathbf{x} + (l, k)\triangle) \sigma_{I}^{2}(\mathbf{x} + (l', k')\triangle) \right]\right]}{(2Q+1)^{4}} \nonumber \\
&=  \frac{1}{(2Q+1)^{4}}\sum_{\substack{l', k', l, k\\ = -Q}}^{Q}\left[\Var\left[\sigma_{I}^{2}(\mathbf{x})\right] -  \Cov\left(\sigma_{I}^{2}(\mathbf{x}), \sigma_{I}^{2}(\mathbf{x} + (l, k)\triangle) \right)  -  \Cov\left( \sigma_{I}^{2}(\mathbf{x}), \sigma_{I}^{2}(\mathbf{x} + (l', k')\triangle) \right)\right. \nonumber \\
&\left.+ \Cov\left(\sigma_{I}^{2}(\mathbf{x} + (l, k)\triangle), \sigma_{I}^{2}(\mathbf{x} + (l', k')\triangle) \right)\right] \text{ since } \sigma^{2}_{I} \text{ is stationary,} \nonumber\\
&=  \frac{1}{(2Q+1)^{4}}\sum_{\substack{l', k', l, k\\ = -Q}}^{Q}\left[\Var\left[\sigma_{I}^{2}(\mathbf{x})\right] -  \Cov\left(\sigma_{I}^{2}(\mathbf{x}), \sigma_{I}^{2}(\mathbf{x} + (l, k)\triangle) \right)\right] \label{eqn:MSEt4} \\
&+ \frac{1}{(2Q+1)^{4}}\sum_{\substack{l', k', l, k\\ = -Q}}^{Q} \left[ \Cov\left(\sigma_{I}^{2}(\mathbf{x} + (l, k)\triangle), \sigma_{I}^{2}(\mathbf{x} + (l', k')\triangle) \right)-  \Cov\left( \sigma_{I}^{2}(\mathbf{x}), \sigma_{I}^{2}(\mathbf{x} + (l', k')\triangle) \right) \right] \label{eqn:MSEt5} 
\end{align}
Since we assumed that $C(d_{x_{1}}, d_{x_{2}}) := \Cov\left(\sigma_{I}^{2}(\mathbf{x}), \sigma_{I}^{2}(\mathbf{x} + (d_{x_{1}}, d_{x_{2}}))\right)$ has a finite gradient over $\mathbb{R}^{2}$, i.e. $\sigma_{I}^{2}$ has a second order mean squared derivative (page 27 of \cite{Adler2010}), we can use the Mean Value Theorem and there exists points $\{\mathbf{c}_{(l, k)}: -Q \leq l, k \leq Q\}$ in $\mathbb{R}^{2}$ such that: 
\begin{align*}
&\left|\frac{1}{(2Q+1)^{4}}\sum_{\substack{l', k', l, k\\ = -Q}}^{Q}\left[\Var\left[\sigma_{I}^{2}(\mathbf{x})\right] -  \Cov\left(\sigma_{I}^{2}(\mathbf{x}), \sigma_{I}^{2}(\mathbf{x} + (l, k)\triangle) \right)\right]\right| \\
&\leq \frac{1}{(2Q+1)^{2}}\sum_{l, k = -Q}^{Q}\left|\Var\left[\sigma_{I}^{2}(\mathbf{x})\right] -  \Cov\left(\sigma_{I}^{2}(\mathbf{x}), \sigma_{I}^{2}(\mathbf{x} + (l, k)\triangle) \right)\right| \\
&\leq \frac{1}{(2Q+1)^{2}}\sum_{l, k = -Q}^{Q}\left|\nabla C(c_{(l, k)})\right|\left| (l, k)\triangle\right| \text{ where } \nabla C \text{ denotes the gradient of C with respect to } d_{x_{1}} \text{ and } d_{x_{2}},\\
&\leq \frac{K\triangle}{(2Q+1)^{2}}\sum_{l, k = -Q}^{Q} \sqrt{l^{2} + k^{2}} \text{ where } K \text{ is the maximum absolute value of the gradients, } \\
&\leq \frac{K\triangle}{(2Q+1)^{2}}\sum_{l, k = -Q}^{Q} (|l| + |k|) \text{ by the Triangle Inequality,}\\
&= \frac{2K\triangle Q(Q+1)}{(2Q+1)^{2}},
\end{align*}
which converges to zero if $Q$ tends to infinity and $\triangle$ behaves like $O(Q^{-\tilde{r}})$ for $\tilde{r}>0$. 
\\
Similarly, for term (\ref{eqn:MSEt5}), there exists points $\{\mathbf{c}_{(l, k, l', k')}: -Q \leq l, k, l', k' \leq Q\}$ such that: 
\begin{align*}
 &\left|\frac{1}{(2Q+1)^{4}}\sum_{\substack{l', k', l, k\\ = -Q}}^{Q} \left[ \Cov\left(\sigma_{I}^{2}(\mathbf{x} + (l, k)\triangle), \sigma_{I}^{2}(\mathbf{x} + (l', k')\triangle) \right)-  \Cov\left( \sigma_{I}^{2}(\mathbf{x}), \sigma_{I}^{2}(\mathbf{x} + (l', k')\triangle) \right) \right]\right| \\
&\leq \frac{1}{(2Q+1)^{4}}\sum_{\substack{l', k', l, k\\ = -Q}}^{Q} \left| \Cov\left(\sigma_{I}^{2}(\mathbf{x} + (l, k)\triangle), \sigma_{I}^{2}(\mathbf{x} + (l', k')\triangle) \right)-  \Cov\left( \sigma_{I}^{2}(\mathbf{x}), \sigma_{I}^{2}(\mathbf{x} + (l', k')\triangle) \right) \right|\\
&\leq  \frac{1}{(2Q+1)^{4}}\sum_{\substack{l', k', l, k\\ = -Q}}^{Q} \left|\nabla f(c_{(l, k, l', k')})\right|\left| (l, k)(\triangle)\right| \\
&\leq \frac{K'\triangle}{(2Q+1)^{2}}\sum_{l, k = -Q}^{Q} \sqrt{l^{2} + k^{2}} \text{ where } K' \text{ is the maximum absolute value of the gradients, } \\
&= \frac{2K'\triangle Q(Q+1)}{(2Q+1)^{2}}. 
\end{align*}
So, term (\ref{eqn:MSEt5}) also converges to zero  if $Q$ tends to infinity and $\triangle$ behaves like $O(Q^{-\tilde{r}})$ for $\tilde{r}>0$. Since under the latter conditions the MSE of our local variance estimator decreases to zero, we have proved that it converges to the true local variance in the $\mathcal{L}_{2}$ sense.
\end{proof}

\begin{proof}[Proof of Corollary \ref{cor:convar}]
Since $\mathcal{L}_{2}$ convergence implies convergence in probabilty, each point at which we compute $\hat{\sigma}_{I}^{2}$ is associated with a sequence $\{ \hat{\sigma}_{I}^{2}(\mathbf{x}, Q) \}$ converges in probability towards $\sigma_{I}^{2}(\mathbf{x})$. By Theorem 2.7(vi) of \cite{VDV1998}, this means that the vector $\left(\hat{\sigma}_{I}^{2}(\mathbf{x}_{1}, Q), \dots, \hat{\sigma}_{I}^{2}(\mathbf{x}_{M}, Q)\right)$ converges in probability to $\left(\sigma_{I}^{2}(\mathbf{x}_{1}), \dots, \sigma_{I}^{2}(\mathbf{x}_{M})\right)$. 
\\
Furthermore, since convergence in probability implies convergence in distribution (also established in Theorem 2.7(ii) of \cite{VDV1998}), we have that the limit of the joint distribution of $\left(\hat{\sigma}_{I}^{2}(\mathbf{x}_{1}, Q), \dots, \hat{\sigma}_{I}^{2}(\mathbf{x}_{M}, Q)\right)$ is the joint distribution of $\left(\sigma_{I}^{2}(\mathbf{x}_{1}), \dots, \sigma_{I}^{2}(\mathbf{x}_{M})\right)$. This means that the mean, variance and normalised variogram of the local variance estimator converge to those of the true conditional variance.
\end{proof}

\section*{Acknowledgements}

We would like to thank Mikko Pakkanen, Claudia Kl\"uppelberg and Carsten Chong for helpful discussions. M. Nguyen is grateful to Imperial College for her PhD scholarship which supported this research. A.E.D. Veraart acknowledges financial support by a Marie Curie FP7 Integration Grant within the 7th European Union Framework Programme.

\bibliographystyle{agsm} 
\bibliography{refs}

\vspace{2mm}
Michele Nguyen, Department of Mathematics, Imperial College London, 180 Queen's Gate, SW7 2AZ London, UK. \\
Email: michele.nguyen09@imperial.ac.uk

\end{document}